\documentclass[aps,10pt,groupedaddress,prx,superscriptaddress,twocolumn,balancelastpage]{revtex4-2}

\usepackage[tighter]{newpxtext}
\usepackage{newpxmath}

\RequirePackage[utf8]{inputenc}
\RequirePackage[english]{babel}
\RequirePackage[T1]{fontenc}
\RequirePackage{amsmath}
\RequirePackage{amsthm}
\RequirePackage{amssymb}
\RequirePackage{graphicx}
\RequirePackage{subcaption}
\RequirePackage{amsfonts}
\RequirePackage{amsbsy}
\RequirePackage{colonequals}
\RequirePackage{nccmath}
\RequirePackage{mathtools}
\RequirePackage{stmaryrd}
\SetSymbolFont{stmry}{bold}{U}{stmry}{m}{n}
\RequirePackage{bm}
\RequirePackage{dsfont}
\RequirePackage{lipsum}
\RequirePackage[hidelinks]{hyperref}
\RequirePackage{ragged2e}
\RequirePackage{braket}
\RequirePackage{xifthen}
\RequirePackage{makecell}
\RequirePackage[inline]{enumitem}
\RequirePackage{booktabs}
\RequirePackage{tikz}
\newcommand{\tikzcmfont}{\fontfamily{cmr}\selectfont}

\tikzset{
  cm text/.style={
    every node/.append style={font=\tikzcmfont}
  }
}

\RequirePackage[mode=build]{standalone}
\RequirePackage{tcolorbox}
\RequirePackage{mdframed}
\RequirePackage{geometry}
\RequirePackage{titletoc}
\RequirePackage[createShortEnv]{proof-at-the-end}
\RequirePackage[classfont=sanserif,langfont=sanserif,funcfont=sanserif]{complexity}
\RequirePackage[ruled,lined,linesnumbered]{algorithm2e}
\SetKwComment{Comment}{$\triangleright$\ }{}

\RequirePackage{colortbl}

\RequirePackage{booktabs}
\RequirePackage{pifont}
\RequirePackage{subcaption}

\RequirePackage{multirow}

\RequirePackage{hhline}

\RequirePackage{physics}
\RequirePackage{wasysym}

\usetikzlibrary{arrows,3d,shapes,calc,decorations.pathreplacing,decorations.markings,positioning,intersections,shapes.symbols,positioning}

\DeclareMathAlphabet{\mymathbb}{U}{BOONDOX-ds}{m}{n}

\theoremstyle{plain}
\newtheorem{theorem}             {Theorem}
\newtheorem{proposition}{Proposition}

\newtheorem{lemma}{Lemma}

\newtheorem{corollary}{Corollary}
\newtheorem{conjecture}{Conjecture}
\newtheorem*{theorem*}    {Theorem}
\newtheorem*{proposition*}{Proposition}
\newtheorem*{lemma*}      {Lemma}
\newtheorem*{corollary*}  {Corollary}
\newtheorem*{conjecture*} {Conjecture}

\theoremstyle{definition}
\newtheorem{definition}{Definition}

\newtheorem*{definition*}{Definition}
\newtheorem*{example*}   {Example}

\theoremstyle{remark}

\newtcolorbox{mybox}[2][]{%attach boxed title to top center
               = {yshift=-8pt},
  colback      = cyan!6!white,
  colframe     = cyan!1!black,
  halign       = flush left,
  fonttitle    = \bfseries\sffamily,
  colbacktitle = cyan!50!black,
  title        = #2,#1,
  }

\newcommand{\phsp}[2]{%
\ifthenelse{\isempty{#1}}
	{\rule[-.5\baselineskip]{0pt}{.5\baselineskip}}%
	{\rule[#1\baselineskip]{0pt}{#2\baselineskip}}%
}

\newcommand{\Id}{\mathds{I}}

\newcommand{\vbra}[1]{\langle\!\langle #1 |}
\newcommand{\vket}[1]{| #1 \rangle\!\rangle}

\makeatletter
\newcommand{\leqnomode}{\tagsleft@true\let\veqno\@@leqno}
\newcommand{\reqnomode}{\tagsleft@false\let\veqno\@@eqno}
\makeatother

\makeatletter
\newcommand{\proglabel}[2]{%
   \protected@write \@auxout {}{\string \newlabel {#1}{{#2}{\thepage}{#2}{#1}{}} }%
   \hypertarget{#1}{#2}
}
\makeatother

\newcommand{\Cm}{L}
\newcommand{\Dm}{R}

\usepackage{mathrsfs}

\usepackage[nameinlink]{cleveref} % 'nameinlink' to mimic \autoref's behavior
\crefname{chapter}{Chapter}{Chapters}
\crefname{section}{Section}{Sections}
\crefname{algorithm}{Algorithm}{Algorithms}
\crefname{line}{Line}{Lines}
\crefname{equation}{Eq.}{Eqs.}
\crefname{figure}{Fig.}{Figs.}
\crefname{table}{Table}{Tables}
\crefname{appendix}{Appendix}{Appendices}
\crefname{theorem}{Theorem}{Theorems}
\crefname{corollary}{Corollary}{Corollaries}
\crefname{lemma}{Lemma}{Lemmas}
\crefname{conjecture}{Conjecture}{Conjectures}
\crefname{proposition}{Proposition}{Propositions}
\crefname{definition}{Definition}{Definitions}
\crefname{footnote}{Footnote}{Footnotes}

\let\autoref\cref

\newcommand{\pr}[2][]{
	\mathop{
		\ifx &#1&
		\mathrm{Pr}
		\else
			\mathrm{Pr}_{#1}
		\fi
		\left[#2\right]}
}

\newcommand{\e}[2][]{
	\mathop{
		\ifx &#1&
			\mathbb{E}
		\else
			\underset{#1}{\mathbb{E}}
		\fi
		\left[#2\right]}
}
\renewcommand{\var}[2][]{
	\mathop{
		\ifx &#1&
			\mathrm{Var}
		\else
			\underset{#1}{\mathrm{Var}}
		\fi
		\left[#2\right]}
}

\newcommand{\s}{{\bm{s}}}

\newcommand{\x}{\bm{x}}

\newcommand{\dbra}[1]{\langle\!\langle #1 |}
\newcommand{\dket}[1]{| #1 \rangle\!\rangle}

\def\multiset#1#2{\ensuremath{\left(\kern-.3em\left(\genfrac{}{}{0pt}{}{#1}{#2}\right)\kern-.3em\right)}}

\usepackage{pgfplots}
\usepackage{ytableau}
\pgfplotsset{compat=newest}
\usepgfplotslibrary{groupplots}
\usepgfplotslibrary{polar}
\usepgfplotslibrary{smithchart}
\usepgfplotslibrary{statistics}
\usepgfplotslibrary{dateplot}
\usepgfplotslibrary{ternary}
\usetikzlibrary{arrows.meta}
\usetikzlibrary{backgrounds}
\usepgfplotslibrary{patchplots}
\usepgfplotslibrary{fillbetween}
\pgfplotsset{%
    layers/standard/.define layer set={%
        background,axis background,axis grid,axis ticks,axis lines,axis tick labels,pre main,main,axis descriptions,axis foreground%
    }{
        grid style={/pgfplots/on layer=axis grid},%
        tick style={/pgfplots/on layer=axis ticks},%
        axis line style={/pgfplots/on layer=axis lines},%
        label style={/pgfplots/on layer=axis descriptions},%
        legend style={/pgfplots/on layer=axis descriptions},%
        title style={/pgfplots/on layer=axis descriptions},%
        colorbar style={/pgfplots/on layer=axis descriptions},%
        ticklabel style={/pgfplots/on layer=axis tick labels},%
        axis background@ style={/pgfplots/on layer=axis background},%
        3d box foreground style={/pgfplots/on layer=axis foreground},%
    },
}

\renewcommand{\geq}{\geqslant}
\renewcommand{\leq}{\leqslant}

\definecolor{lightblue}{RGB}{84,189,220}
\definecolor{navy}{RGB}{47,60,126}
\definecolor{quandelablue}{HTML}{435BEC}
\definecolor{quandelared}{HTML}{F07362}
\definecolor{darkviolet}{RGB}{99,56,142}
\definecolor{darkgreen}{RGB}{39,174,96}

\definecolor{c1}{HTML}{ffbf00}
\definecolor{c2}{HTML}{e83f6f}
\definecolor{c3}{HTML}{2274a5}

\renewcommand{\selectlanguage}[1]{}

\newclass{\BosonPCC}{BosonP}

\newcommand{\per}[1]{\mathrm{Per}\left(#1\right)}

\RequirePackage{anyfontsize}

\usepackage{scalerel}

\definecolor{darkblue}{HTML}{224C98}
\definecolor{orange}{HTML}{FFAB0D}

\makeatletter

\newif\ifinappendix
\inappendixfalse

\pretocmd{\appendix}{\inappendixtrue}{}{}

\let\orig@section\section
\let\orig@subsection\subsection
\let\orig@subsubsection\subsubsection

\renewcommand{\section}{%
  \@ifstar
    {\orig@section*}%
    {\app@section}%
}

\newcommand{\app@section}[1]{%
  \orig@section{#1}%
  \ifinappendix
    \addcontentsline{atoc}{section}{%
      \protect\numberline{\thesection}#1%
    }%
  \fi
}

\renewcommand{\subsection}{%
  \@ifstar
    {\orig@subsection*}%
    {\app@subsection}%
}

\newcommand{\app@subsection}[1]{%
  \orig@subsection{#1}%
  \ifinappendix
    \addcontentsline{atoc}{subsection}{%
      \protect\numberline{\thesubsection}#1%
    }%
  \fi
}

\renewcommand{\subsubsection}{%
  \@ifstar
    {\orig@subsubsection*}%
    {\app@subsubsection}%
}

\newcommand{\app@subsubsection}[1]{%
  \orig@subsubsection{#1}%
  \ifinappendix
    \addcontentsline{atoc}{subsubsection}{%
      \protect\numberline{\thesubsubsection}#1%
    }%
  \fi
}

\makeatother

\usepackage[toc,page]{appendix}
\usepackage[tight]{minitoc}
\usepackage{tocloft}

\doparttoc % Tell to minitoc to generate a toc for the parts
\faketableofcontents % Run a fake tableofcontents command 

\usepackage{apptools}

\AtAppendix{\setcounter{theorem}{0}}
\AtAppendix{\setcounter{lemma}{0}}
\AtAppendix{\setcounter{corollary}{0}}
\AtAppendix{\setcounter{definition}{0}}

\hypersetup{
    colorlinks=true,
    linkcolor=navy,
    citecolor=navy,
    urlcolor=navy,
    breaklinks=true,
}

\begin{document}

\title{Boson sampling beyond the dilute regime: second moments and anti-concentration}

\author{Hela Mhiri}
\affiliation{Laboratoire d’Informatique de Paris 6, CNRS, Sorbonne Université, Paris, France}
\affiliation{Institute of Physics, Ecole Polytechnique Fédérale de Lausanne (EPFL), Lausanne, Switzerland}

\author{Hugo Thomas}
\affiliation{Laboratoire d’Informatique de Paris 6, CNRS, Sorbonne Université, Paris, France}
\affiliation{Quandela, 7 rue Léonard de Vinci, Massy, France}
\affiliation{DIENS, \'Ecole Normale Supérieure, PSL University, CNRS, INRIA, Paris, France}

\author{Léo Monbroussou}
\affiliation{Laboratoire d’Informatique de Paris 6, CNRS, Sorbonne Université, Paris, France}
\affiliation{School of Informatics, University of Edinburgh, Edinburgh, United Kingdom}
\affiliation{Terra Quantum GmbH, Munich, Germany}

\author{Ulysse Chabaud}
\affiliation{DIENS, \'Ecole Normale Supérieure, PSL University, CNRS, INRIA, Paris, France}

\author{Zoë Holmes}
\affiliation{Institute of Physics, Ecole Polytechnique Fédérale de Lausanne (EPFL), Lausanne, Switzerland}
\affiliation{Centre for Quantum Science and Engineering, École Polytechnique Fédérale de Lausanne (EPFL), Lausanne, Switzerland}

\author{Elham Kashefi}
\affiliation{Laboratoire d’Informatique de Paris 6, CNRS, Sorbonne Université, Paris, France}
\affiliation{School of Informatics, University of Edinburgh, Edinburgh, United Kingdom}

\begin{abstract}
Boson sampling is a leading candidate for demonstrating quantum advantage in photonic systems. Despite significant experimental and theoretical progress, a characterization of its output statistics remains incomplete. 
This is especially true beyond the dilute regime, where photon collisions and bunching become significant. The associated saturated regime, characterized by mode number growing linearly with photon number, or more generally sub-quadratically, is precisely the regime of greatest experimental interest.
As a consequence, anti-concentration of the output distribution--a key ingredient in hardness arguments--remains poorly understood in boson sampling.
In this work, we leverage  representation-theoretic tools to address this gap, obtaining closed-form expressions for second moments of generic particle-number-preserving bosonic observables. We express these quantities in terms of Hilbert–Schmidt norms of projections onto irreducible components of the operator space and show that these projection norms admit compact analytical expressions by exploiting the underlying symmetry structure.
Focusing on Fock state output probabilities, we further establish anti-concentration beyond the dilute regime. Together with recent complexity-theoretic results, our findings strengthen hardness guarantees for boson sampling in experimentally interesting settings.

\end{abstract}

\maketitle

\section{Introduction}

Boson sampling is a promising paradigm for demonstrating quantum advantage in sampling tasks \cite{aaronson_computational_2011}. 
 Over the past decade, photonic implementations of boson sampling have steadily improved in scale and performance \cite{broome_photonic_2013,tillmann_experimental_2013,loredo_boson_2017,wang_highefficiency_2017,hoch_reconfigurable_2022,young_atomic_2024}. In parallel, significant theoretical progress has been made in understanding the computational complexity of these systems, including refined hardness results with minimal conjectures and extensions to experimentally relevant regimes \cite{bouland_complexitytheoretic_2023,bouland2025exponentialimprovementsaveragecasehardness,Bouland_2022}.
Despite these advances, a complete characterization of the statistical properties of the output distribution is still lacking.  In particular, the second moment of the output distribution of such circuits have not been derived in full generality. 

The study of second moments is motivated by their central role in understanding the concentration properties of expectation values of generic observables evolved under random interferometers. Specifically, second moment computations arise in a wide range of quantum information processing tasks, including quantum advantage experiments \cite{aaronson_computational_2011,Hangleiter_2018,Hangleiter_2023,Ehrenberg_2025,Ehrenberg_2025_tansition,oliveira_immanants_2021}, certification \cite{bouland_complexity_2019,martinezcifuentes2024linearcrossentropycertificationquantum, huang_statistical_2017, walschaers_statistical_2016, phillips_benchmarking_2019, giordani_experimental_2018}, randomized benchmarking \cite{heinrich_randomized_2023,arienzo_bosonic_2024,wilkens_benchmarking_2024}, classical shadows \cite{thomas_shedding_2025,aaronson_shadow_2018,west_real_2024,west2026classicalshadowsarbitrarygroup}, classical simulation  \cite{Shchesnovich_2019,angrisani2025simulatingquantumcircuitsarbitrary,Martinez_2025} and  concentration phenomena in variational quantum algorithms \cite{mcclean_barren_2018, sharma_trainability_2022, patti_entanglement_2021, ortiz_marrero_entanglement-induced_2021, holmes_barren_2021, schatzki_theoretical_2024, west_provably_2024, garcia-martin_quantum_2024, liu_laziness_2024, larocca_barren_2025,fontana_characterizing_2024,larocca_diagnosing_2022,diaz2023showcasingbarrenplateautheory,cerezo_cost_2021,mhiri_constrained_2025,mhiri_unifying_2025, thabet_when_2026}.

One particularly important property of boson sampling, captured by second-moment quantities, is the anti-concentration of the output distribution. Broadly, it asserts that, on average, output probabilities are not too small compared to their mean for a non-negligible fraction of outcomes \cite{aaronson_computational_2011}.
In hardness arguments for random  sampling schemes, and in particular boson sampling, a key challenge is the gap between conjectured average-case hardness and what can be rigorously established \cite{movassagh2023hardness,bouland2025exponentialimprovementsaveragecasehardness,bouland_complexitytheoretic_2023}. Anti-concentration, together with multiplicative-error hardness, helps bridge this gap by enabling reductions from additive to multiplicative error guarantees \cite{Hangleiter_2023,bouland_complexitytheoretic_2023}.
Moreover, it constrains the structure of the output distribution by ruling out polynomial sparsity \cite{Pashayan_2020}, a regime that admits efficient classical simulation of boson sampling experiments \cite{lim2025classicalalgorithmsestimatingexpectation}.

Existing analyses of second moments and therefore anti-concentration have primarily focused on the dilute regime, where the number of modes grows at least quadratically with the number of photons. In this limit, photon collisions are rare and the so-called \emph{hiding} property ensures that small submatrices of random interferometers behave as random Gaussian matrices with independent entries, so that the statistics of output probabilities can be treated using random matrix theory \cite{aaronson_computational_2011,Ehrenberg_2025,Ehrenberg_2025_tansition,nezami2021permanentrandommatricesrepresentation}. However, the hiding property breaks down in the  experimentally relevant \emph{saturated} regime, in which the number of modes scales linearly with the number of photons and collisions occur with significant probability \cite{bouland_complexitytheoretic_2023}. This renders current techniques for analysing boson sampling statistics based on the hiding property ineffective in experimentally relevant settings.

In this work, we address this gap by providing: (i) a general account of second moment computations in linear optics and (ii) a proof of the anti-concentration conjecture beyond the dilute regime. Our main contributions are detailed hereafter.

(i) We develop a representation-theoretic framework for linear optics (see \autoref{sec:reptheory}) to derive closed-form expressions of second moments for expectation values of particle-number-preserving observables.
Concretely, we express these second moments in terms of Hilbert–Schmidt norms of projections onto irreducible components of the operator space (\autoref{prop:sec_moment}), and we show that these norms can be computed analytically using the underlying representation-theoretic structure of linear optics (\autoref{th:recursive_irrep_norm}). At a high level, the operator space exhibits a rich symmetry structure, which organizes it into irreducible components that can be identified recursively and used to derive compact expressions for their corresponding projection norms.

(ii) We then leverage this representation-theoretic framework to study anti-concentration of boson sampling. In particular, combining it with combinatorial techniques, we derive a compact closed-form expression for the second moment of averaged output probabilities, corresponding to the \emph{normalized average outcome collision probability} (\autoref{th:P2}). This quantity quantifies the probability of sampling the same outcome twice from a random interferometer and thus provides a standard measure of anti-concentration \cite{aaronson_computational_2011,Hangleiter_2023,Ehrenberg_2025}. Moreover, this quantity is also of independent interest, as it naturally arises in linear cross-entropy benchmarking, where it characterizes the score of the ideal distribution \cite{Ehrenberg_2025}.
By analysing its asymptotic scaling, we prove anti-concentration in the presence of photon collisions (\autoref{cj:scalingP}). This result complements recent complexity-theoretic advances on the hardness of boson sampling in the saturated regime, where the number of modes scales linearly with the number of particles \cite{bouland_complexitytheoretic_2023}, and provides further evidence toward bridging existing gaps in hardness arguments.

\section{Background and preliminary results}\label{sec:moments}

In this section, we review the key concepts and tools underlying our analysis. We begin by introducing the representation-theoretic framework relevant to passive linear-optical systems, with particular emphasis on group actions and the decomposition of bosonic operator spaces (see \autoref{Fig_Intro.pdf}). We then discuss the role of anti-concentration in complexity-theoretic arguments for sampling problems, focusing on boson sampling, and highlight its connection to second-moment quantities as well as prior work addressing this question.

\subsection{Representation theory framework for passive linear optics}
\label{sec:reptheory}

In this section, we review the representation-theoretic framework underlying passive linear-optical transformations of bosonic systems. Readers who are already familiar with this material, or prefer to avoid these technical details, may skip directly to \autoref{sec:results}, as the main results can be understood without relying heavily on this material. 

In \autoref{sec:bosonic}, we introduce the photonic homomorphism, which induces an action of the group $U(m)$ on the operator space of particle-number-preserving bosonic observables. 

In \autoref{sec:howe_duality}, we show that this space decomposes into fixed particle-number sectors, on which the adjoint action defines a unitary representation that further decomposes into irreducible components, as we recall in \autoref{lemma:irrep_decomp}.  

We then uncover, in \autoref{lemma:sl2_commutation_main}, an additional $\mathfrak{sl}_2(\mathbb{C})$ structure acting on this operator space, generated by natural raising and lowering maps that connect different photon-number sectors. As detailed in \autoref{sec:howe} and \autoref{th:howe_duality}, this action commutes with the $U(m)$ action, giving rise to a $U(m)$--$\mathfrak{sl}_2(\mathbb{C})$ Howe duality that organizes the operator space into a ladder of irreducible components. This structure yields a recursive construction of these components and, crucially, enables the efficient evaluation of quantities such as second moments through compact expressions for projection norms.

Further details and derivations are provided in \autoref{app:rep}. While some of these results are likely known, we include them here for completeness, as we were unable to find a self-contained treatment adapted to our setting.

\begin{figure*}[t!]
    \centering
    
    \includegraphics[width=1.0\linewidth]{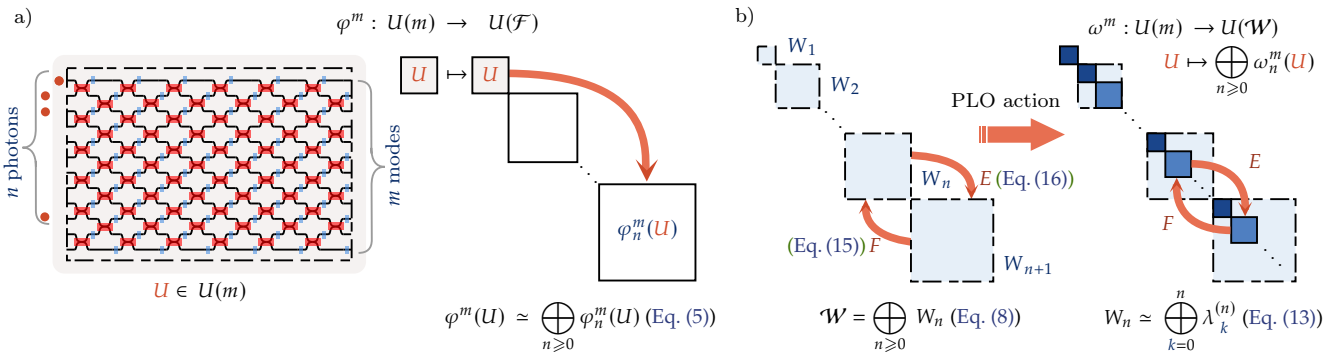}
    \caption{\justifying Schematics of passive linear optics transformations from a representation theory perspective. a) Linear-optical circuit made of beam-splitters and phase-shifters on $m$ spatial modes described by a unitary in $U(m)$, with single photons as input (left); action of a linear interferometer with $n$ input photons through the photonic homomorphism $\varphi$ giving rise to a block-diagonal decomposition of the Fock space under passive linear optics transformations (right). b) Decomposition of the operator space of particle-number-preserving observables into irreps under the action of the adjoint representation $\omega$ (ignoring the trivial  subspace $W_0$), same copy appearing in different particle sector have same color and connected via raising and lowering maps.}
    \label{Fig_Intro.pdf}
\end{figure*}

\subsubsection{Fock space and passive linear optics}\label{sec:bosonic}

The Hilbert space of states of indistinguishable photons in $m$ modes is given by the bosonic Fock space $\mathcal{F} = \bigoplus_{n \geq 0} \mathcal{F}_n$, where each sector $\mathcal{F}_n = \mathrm{Sym}^n(\mathbb{C}^m)$ describes configurations of $n$ identical particles. This symmetric structure is  realized using bosonic creation and annihilation operators satisfying
\begin{equation}
    [a_i,a_j^\dagger]= \delta_{i,j}, \qquad [a_i,a_j]=[a_i^\dagger,a_j^\dagger]=0.
\end{equation}

An orthonormal basis of $\mathcal{F}_n$ is given by occupation-number states
\begin{equation}
    \ket{S} = \prod_{i=1}^m \frac{(a_i^\dagger)^{s_i}}{\sqrt{s_i!}}\ket{\boldsymbol{0}}, 
    \quad \sum_{i=1}^m s_i = n,
\end{equation}
where $\ket{\boldsymbol{0}}$ corresponds to the vacuum state on $m$ modes. These occupation-number states are
labelled by configurations $S \in \Phi_m^n$ defined as
\begin{equation}
    \Phi_m^n := \{(s_1,\dots,s_m)\,|\,\sum_{i=1}^m s_i=n\} \;, | \Phi_m^n| =  \binom{n+m-1}{n}\;.
\end{equation}

In the following, we fix the number of modes $m$ and omit it from the notation; whenever not stated explicitly, all objects are assumed to act on $m$ modes.

Passive linear-optical transformations form a subclass of Gaussian operations that preserve particle number. Also known as linear interferometers, they are generated by quadratic Hamiltonians of the form $a_j^\dagger a_i$ (together with their Hermitian conjugates), which physically correspond to beam splitters and phase shifters \cite{makarov_theory_2022,pan_multiphoton_2012}. These operators close under commutation with their Hermitian combinations generating  the Lie algebra $\mathfrak{u}(m)$.
Exponentiation yields the unitary group $U(m)$, which characterizes all passive linear optical transformations on $m$ modes \cite{garcia-escartin_multiple_2019,parellada_nogo_2023,parellada2025liealgebraicinvariantsquantum}.
 
Concretely, a linear interferometer is fully characterized by its action on the single-particle space $\mathcal{F}_1 = \mathbb{C}^m$, given by an $m \times m$ unitary matrix. This action lifts to the full Fock space via a unitary representation, known as the photonic homomorphism \cite{aaronson_computational_2011},
\begin{equation}\label{eq:photonic_homomorphism}
    \varphi: U(m) \rightarrow U(\mathcal{F})\;.
\end{equation}
Moreover, particle-number conservation implies that this representation decomposes into invariant photon-number sectors and is block-diagonal in the Fock basis, i.e.
\begin{equation}\label{eq:phi_n}
\varphi(U) \simeq \bigoplus_{n \geq 0} \varphi_n(U)\;,    
\end{equation}
where $\varphi_n(U):= U^{\otimes n}|_{\mathcal{F}_n}$ is the maximally symmetric  irreducible representation of the group $U(m)$ of the $n$-particle Fock sector $\mathcal{F}_n$ \cite{VilenkinKlimyk1995,arienzo_bosonic_2024}.

The statistics of output distributions generated by linear interferometers are governed by this symmetric tensor structure (stemming from photon indistinguishability). 
This structure underlies the appearance of matrix permanents in the expression of output probabilities \cite{aaronson_computational_2011}. Concretely, for $U \in U(m)$,  matrix amplitudes of the photonic representation $\varphi_n(U)$, expressed in the Fock basis, are given by permanent of submatrices of $U$, namely
\begin{equation}\label{eq:permanent}
    \mel{T}{\varphi_n(U)}{S} = \frac{\per{U_{S, T}}}{\sqrt{\prod_{i=1}^{m}s_i! \prod_{j=1}^{m}t_j!}},
\end{equation}
where $U_{S, T}$ is obtained from $U$ by taking $s_i$ copies of the $i$-th row
and $t_j$ copies of the $j$-th column \cite{aaronson_computational_2011}, and
${\rm Per(A) = \sum_{\sigma \in S_n} \prod_{i=1}^n A_{i,\sigma(i)}}$ is the matrix permanent function for an $n \times n$ matrix \cite{marcus_permanents_1965}. 

Another determining factor in the output statistics is the relative scaling of the number of photons $n$ and modes $m$. In particular, one distinguishes two regimes.

\begin{definition}[Dilute and saturated regimes]\label{def:regimes}
For a passive linear optical systems with $n$ photons in $m$ modes:
\begin{enumerate}
    \item The \emph{dilute} (or collision-free, no-bunching) regime corresponds to $m = \Omega(n^2)$, where photons are sparsely distributed and collisions are suppressed.
    
    \item The \emph{saturated} (or linear, low-mode) regime corresponds to $m = \Theta(n)$, where output configurations with substantial collisions become typical.
\end{enumerate}
\end{definition}

This distinction is captured by the bosonic birthday paradox, which identifies when collisions become significant in typical output configurations (see \autoref{sec:P2}, \autoref{eq:collision_free,eq:quadratic,eq:intermediate1,eq:intermediate2,eq:linear}) \cite{aaronson_computational_2011,arkhipov_bosonic_2012}. In what follows, we refer to the regime between dilute and saturated (see \autoref{def:regimes}) as the collision, or \emph{intermediate saturated}, regime. While the choice of regime strongly affects the statistics of output probabilities in boson sampling, the representation-theoretic analysis we develop in the remainder of the paper applies uniformly across all regimes.

\subsubsection{Structure of particle-number-preserving bosonic operators and adjoint representation}\label{sec:howe_duality}

An important class of operators relevant to our analysis is the subalgebra 
$\mathcal{W}$ of particle-number–preserving operators, spanned by normally ordered monomials containing equal numbers of creation and annihilation operators,
\begin{equation}\label{eq:subalgebra_commute}
    \mathcal{W} = {\rm Span} \{  a_{j_1}^\dagger \dots a_{j_d}^\dagger a_{i_1} \dots a_{i_d} \;|\;  d\geq 0\}\;.
 \end{equation}
 By construction, operators in 
$\mathcal{W}$ preserve the total photon number and hence act block-diagonally on the Fock Hilbert space,
\begin{equation}\label{eq:decomp_big_W}
    \mathcal{W} \simeq \bigoplus_{n \geq 0} W_n\;,
\end{equation}
where $W_n = \mathcal{F}_n \otimes \mathcal{F}_n^*$ denotes the  space of operators acting on the $n$--photon sector  $\mathcal{F}_n$.

Passive linear optics transformations, described by the group \(U(m)\), act on the bosonic operator space by conjugation via the adjoint representation  $\omega$ given by
\begin{equation}\label{eq:adjoint_rep}
    \omega(U) = \varphi(U) (\cdot)\varphi(U)^\dagger \simeq \varphi(U) \otimes \varphi(U)^*\;.
\end{equation}

Restricting to particle-number-preserving operators \(\mathcal{W}\), the adjoint representation preserves each block $W_n$ introduced in \autoref{eq:decomp_omega_m}, and thus decomposes as
\begin{align}
    \omega(U) \simeq \bigoplus_{n \geq 0} \varphi_n(U) \otimes \varphi_n(U)^* = \bigoplus_{n \geq 0} \omega_n(U)\;, \label{eq:decomp_omega_m}
\end{align}
where $W_n = \mathcal{F}_n \otimes \mathcal{F}_n^*$ is the carrier space of ${\omega_n(U)=\varphi_n(U) \otimes \varphi_n(U)^*}$, and $\simeq$ denotes a unitary change of basis \cite{arienzo_bosonic_2024}.  

While the representation $\varphi_n$  is irreducible (as the maximally symmetric representation of $U(m)$), the corresponding adjoint representation $\omega_n$ acting on $W_n$ is completely reducible by Maschke’s theorem \cite{Folland2016}. The precise structure of this decomposition was established in \cite{arienzo_bosonic_2024}.  

\begin{lemma}[Complete reducibility of passive linear optics adjoint representation, Lemma T1 in \cite{arienzo_bosonic_2024}]\label{lemma:irrep_decomp}
    The space \(W_n\) of operators acting on $n$ particles, equipped with the adjoint representation \(\omega_n\) introduced in \autoref{eq:decomp_omega_m}, decomposes into \(n+1\) inequivalent, multiplicity-free \(U(m)\)-modules,
    \begin{equation}\label{eq:decomp_Wn}
     W_n \simeq \bigoplus_{k=0}^n \lambda_k^{(n)} \;.
 \end{equation}
Each subspace $\lambda_k^{(n)}$  carries an irreducible $U(m)$-representation labelled by $0\leq k\leq n$, and is associated with a Young diagram with $m-1$ rows, where the first row contains $2k$ boxes and the remaining rows each contain $k$ boxes.

The dimension of each irreducible component $\lambda_k^{(n)}$, denoted by $d_k^{(n)}$, is given by (Proposition T2 in \cite{arienzo_bosonic_2024})
 \begin{equation}\label{eq:irrep_dim}
    d_k^{(n)} = \frac{2k+m-1}{m-1} \binom{k+m-2}{k}^2.
\end{equation}
\end{lemma}

\autoref{lemma:irrep_decomp} shows that the operator space $W_n$ associated with a fixed photon-number sector decomposes into multiplicity-free, inequivalent irreducible $U(m)$-representations labelled by $k=0,\dots,n$. Combining \autoref{eq:decomp_omega_m} and \autoref{eq:decomp_Wn}, we obtain the decomposition of the full operator space $\mathcal{W}$,
\begin{equation}\label{eq:full_W_decomp}
    \mathcal{W} = \bigoplus_{n\geq 0} \bigoplus_{k=0}^n \lambda_k^{(n)}\;.
\end{equation}

Moreover, the Young diagram description implies that, for each fixed $k$, the irreducible components $\lambda_k^{(n)}$ are independent of the photon number $n$ for all $n \geq k$. This suggests reorganizing the decomposition in \autoref{eq:full_W_decomp} as  
\begin{equation}\label{eq:ladder_decomp_main}
    \mathcal{W} \simeq \bigoplus_{k \geq 0} \bigoplus_{n \geq k}  \lambda_k^{(n)}\;.
\end{equation}

This reorganization highlights that irreducible components of the same type (i.e. same label $k$) appear across different photon-number sectors, suggesting the presence of an additional structure relating them. To make this relation explicit, it is useful to introduce maps connecting adjacent sectors, which will play a central role in the constructive identification of the irreducible components.

A natural way to relate adjacent sectors is through operations that remove or add a photon in each mode. Specifically, we define two linear maps on \(\mathcal{W}\), referred to as the \emph{lowering} and \emph{raising} maps, denoted respectively by $\Cm$ and $\Dm$,
\begin{align}
     \Cm (\cdot) & = \sum_{s=1}^m a_s (\cdot) a_s^\dagger \;,\label{eq:C_map}\\
     \Dm(\cdot)  & = \sum_{s=1}^m a_s^\dagger (\cdot) a_s.\label{eq:D_map}
\end{align}

These maps admit a simple interpretation in terms of photon number: when acting on operators supported on a fixed sector $W_n$, $\Cm$ decreases the photon number by one, while $\Dm$ increases it by one. They therefore provide a natural mechanism for relating operator spaces associated with different photon numbers. 

Beyond this operational role, these maps exhibit a deeper algebraic structure. In particular, together with the commutator $H := [\Cm, \Dm]$, they generate a representation of the Lie algebra $\mathfrak{sl}_2(\mathbb{C})$ on $\mathcal{W}$, as we formalize in \autoref{lemma:sl2_commutation_main}. Recall that $\mathfrak{sl}_2(\mathbb{C})$ is the Lie algebra  of traceless $2 \times 2$ matrices, with canonical basis
\begin{equation}\label{eq:conanoical_sl2}
      e = \begin{pmatrix} 0 & 1 \\ 0 & 0 \end{pmatrix}, \quad
    f = \begin{pmatrix} 0 & 0 \\ 1 & 0 \end{pmatrix}, \quad
    h = \begin{pmatrix} 1 & 0 \\ 0 & -1 \end{pmatrix}.
\end{equation}
satisfying the commutation relations ${[h,e]= 2e}$,  ${[h,f] = -2f, \;[e,f]=h}$.  

\begin{lemma}{(\(\mathfrak{sl}_2(\mathbb{C})\)-action on \(\mathcal{W}\)).}\label{lemma:sl2_commutation_main} 
Consider the lowering and raising maps introduced in \autoref{eq:C_map} and \autoref{eq:D_map} respectively, together with their commutator $H=[\Cm,\Dm]$. The triple $\{\Cm,\Dm,H\}$ generates a representation of the Lie algebra $\mathfrak{sl}_2(\mathbb{C})$ on the operator space $\mathcal{W}$. Namely, they satisfy the commutation relations
\begin{align}
    [H, \Dm]= 2 \Dm\;, \quad [H, \Cm]= -2 \Cm \;, \quad [\Cm,\Dm]=H\;.
\end{align}  

In particular, this defines a representation $\pi: \mathfrak{sl}_2(\mathbb{C}) \rightarrow \mathrm{End}(\mathcal{W})$ given by  
\begin{equation}\label{eq:pi_sl2_rep}
    \pi(e)= \Dm \;, \quad \pi(f)= - \Cm \;, \quad \pi(h)=H\;,
\end{equation}
where $\mathrm{End}(\mathcal{W})$ denotes the space of endomorphisms on $\mathcal{W}$ and    $\{e,f,h\}$ is the $\mathfrak{sl}_2$ triple introduced in \autoref{eq:conanoical_sl2}.
\end{lemma}

\autoref{lemma:sl2_commutation_main} shows that $\mathcal{W}$ carries an additional symmetry generated by $\Cm$ and $\Dm$, acting alongside the adjoint action of passive linear optics. We note that the minus sign in $\pi(f)$ ensures that the map $\pi$ preserves the $\mathfrak{sl}_2$ commutation relations. The full proof of \autoref{lemma:sl2_commutation_main} is given in \autoref{app:subsec_reptheory}.

In the following section, we exploit this additional structure to constructively identify the $U(m)$-irreducible representations appearing in the decomposition of each fixed particle-number operator space $W_n$ described in \autoref{lemma:irrep_decomp}.

\subsubsection{$U(m)-\mathfrak{sl}_2(\mathbb{C})$ Howe duality and iterative identification of irreducible representations}\label{sec:howe}

On top of the adjoint $U(m)$ action on $\mathcal{W}$ and its complete reducibility within each fixed particle-number sector established in \autoref{lemma:irrep_decomp}, the additional $\mathfrak{sl}_2(\mathbb{C})$ action introduced in \autoref{lemma:sl2_commutation_main} provides a natural way to relate the irreducible components appearing across different sectors. As suggested by the decomposition in \autoref{eq:ladder_decomp_main}, irreps with the same label $k$ recur for all $n\geq k$, and we show below that the maps $\Cm$ and $\Dm$ act as ladder maps connecting these occurrences across particle-number sectors.

A key structural feature underlying this organisation is that the adjoint $U(m)$ action and the $\mathfrak{sl}_2(\mathbb{C})$ action commute, as formally established in \autoref{th:commuting_action} (see \autoref{app:howe_proof}). As a consequence, the decomposition of $\mathcal{W}$ into $U(m)$--irreducible components is compatible with the action of the lowering and raising maps: each family of irreps labelled by $k$ is stable under the  $\mathfrak{sl}_2(\mathbb{C})$ action and therefore carries an $\mathfrak{sl}_2(\mathbb{C})$-module.

This compatibility indicates that the multiplicities observed in the decomposition of $\mathcal{W}$ are not arbitrary, but instead organised by the $\mathfrak{sl}_2(\mathbb{C})$ symmetry, which relates copies of the same $U(m)$--irrep across different particle number sectors.
This interplay between the two commuting actions is characteristic of a dual pair (Howe duality) setting, in which   $\mathfrak{sl}_2(\mathbb{C})$ and $U(m)$ generate mutual centralizers \cite{debie2016howedualitypolynomialsolutions,brito2024quantummetaplectichoweduality}. Concretely, this duality implies a multiplicity-free decomposition of the operator space, where each irreducible component indexed by $k$ is uniquely identified by a pair consisting of a $U(m)$ irrep and an $\mathfrak{sl}_2(\mathbb{C})$ irrep.
In the following, we formalize this structure and make explicit how it constrains the decomposition of $\mathcal{W}$.

\begin{proposition}[\texorpdfstring{$U(m)$--$\mathfrak{sl}_2(\mathbb{C})$}{U(m)-sl2} Howe-duality structure on $\mathcal W$ (informal)]\label{th:howe_duality}
Let $\mathcal W$ be the space of particle-number-preserving operators defined in \autoref{eq:subalgebra_commute}. Then the commuting actions of $U(m)$ and $\mathfrak{sl}_2(\mathbb{C})$ induce a graded decomposition of the form
\begin{equation}\label{eq:irreps}
    \mathcal W = \bigoplus_{k\geq 0} \mathcal W_k,
    \qquad
    \mathcal W_k = \bigoplus_{n\geq k} \lambda_k^{(n)},
\end{equation}
where each $\lambda_k^{(n)}$ is an irreducible $U(m)$-submodule satisfying
\[
\lambda_k^{(n)} \cong \lambda_k^{(k)} \quad \text{for all } n \ge k.
\]

Moreover, each subspace $\mathcal W_k$ is organized by the $\mathfrak{sl}_2(\mathbb{C})$-action into a ladder structure of the form
\begin{equation}\label{eq:ladder}
   \lambda_k^{(k)} \xrightarrow{\Dm} \lambda_k^{(k+1)} \xrightarrow{\Dm} \lambda_k^{(k+2)} \xrightarrow{\Dm} \cdots,
\end{equation}
where the lowest-weight component is given by the kernel of the lowering map $\Cm$ within the $k$-particle operator space, i.e.
\begin{equation}\label{eq:lowest_weight}
    \lambda_k^{(k)} = \ker(\Cm)\cap W_k.
\end{equation}
\end{proposition}

\autoref{th:howe_duality} reveals a correspondence between irreducible $U(m)$-representations and $\mathfrak{sl}_2$-modules appearing in $\mathcal{W}$, arising from an underlying  Howe-duality structure.
A formal statement of this duality is given in \autoref{th:howe_duality_app}, with its proof in
 \autoref{app:howe_proof}. While closely related to known $\mathfrak{gl}(m)$–$\mathfrak{gl}(2)$ Howe duality, the present construction is derived independently and tailored to the bosonic setting \cite{Howe1995PerspectivesOI,nezami2021permanentrandommatricesrepresentation}. 
 
 In particular, each $U(m)$-type $\lambda_k$ arises from a primitive (lowest-weight subspace of $\mathfrak{sl}_2(\mathbb{C})$ irrep) component $\lambda_k^{(k)}$ and generates, through repeated application of the raising map $\Dm$, the subspace $\mathcal{W}_k$, which carries the corresponding $\mathfrak{sl}_2(\mathbb{C})$-module.
 Conversely, the lowering map $\Cm$ identifies these primitive components via its kernel, as given in \autoref{eq:lowest_weight}.

Although this structure induces a global decomposition of $\mathcal{W}$ involving generally infinite-dimensional $\mathfrak{sl}_2$-modules, the underlying ladder structure yields a recursive and constructive description within each fixed photon-number sector $W_n$.
In particular, each irreducible component $\lambda_k^{(n)}$ can be obtained by iteratively applying the raising map $\Dm$ to its associated primitive  component $\lambda_k^{(k)}$.
In this way, the $\mathfrak{sl}_2$ ladder structure organizes the decomposition of each finite-dimensional sector $W_n$ into its irreducible $U(m)$-submodules, leading to the explicit decomposition

\begin{equation}\label{eq:decomp_easy}
W_n = \bigoplus_{k=0}^n   \Dm^{k}(\lambda_{n-k}^{(n-k)})\; .
\end{equation}

The identification of irreducible components of $W_n$ through this ladder structure will be key in establishing second moments expressions, as will be further detailed in \autoref{sec:iter_norms}.

\subsection{Anti-concentration of boson sampling and hardness arguments}\label{sec:ac_back}

In boson sampling, and more broadly in random  sampling schemes, anti-concentration refers to the property that for a typical unitary drawn from a random ensemble of interest, a non-negligible fraction of outcome probabilities are of the same order of magnitude as the uniform distribution. This implies that, on average over the unitary ensemble, the output distribution is sufficiently well spread over the output space, thus not concentrated on a few events \cite{Hangleiter_2023}. 

This property plays a significant role in  complexity-theoretic arguments for the hardness of approximate sampling. While anti-concentration is neither strictly necessary nor sufficient in existing proof techniques, it provides an important consistency check supporting hardness conjectures \cite{Hangleiter_2023, aaronson_computational_2011,bouland_complexitytheoretic_2023}. In particular, such arguments typically reduce approximate sampling to estimating output probabilities to within an exponentially small \emph{additive} error $\varepsilon$  via Stockmeyer’s approximate counting algorithm \cite{aaronson_computational_2011,movassagh2023hardness}. 
If such estimates were hard on average, this would rule out efficient classical sampling under standard complexity assumptions.
 However, current techniques  establish  hardness only for much smaller errors $\varepsilon' \ll \varepsilon$, leaving  a persistent robustness gap \cite{Hangleiter_2023,movassagh2023hardness,bouland_complexitytheoretic_2023}.

Anti-concentration is invoked to bridge this gap by converting additive guarantees into  multiplicative ones.
Indeed, since the additive error scales as $1/|\Phi_m^n|$, such a conversion  requires that typical probabilities are not too small.
In this sense, anti-concentration identifies the subset of instances for which additive estimates can be meaningfully upgraded to multiplicative approximations. Combined with average-case multiplicative hardness, this mechanism closes the robustness gap \cite{aaronson_computational_2011,bouland_complexitytheoretic_2023}, as we further detail in \autoref{app:anticon_hardness}.

Concretely, for a given interferometer $U \in U(m)$ and an output configuration $S \in \Phi_m^n$, we define
\begin{equation}\label{eq:prob_S}
    p_U(S) = \Tr\!\left[\ketbra{S_{\mathrm{cf}}}{S_{\mathrm{cf}}}\, \varphi_n(U)^\dagger \ketbra{S}{S}\, \varphi_n(U)\right]\;,
\end{equation}
as the probability of observing outcome $S$ after evolving the collision-free input state $\ket{S_{\mathrm{cf}}} = \ket{1,\dots,1,0,\dots,0}$ under the action of the interferometer $U$ on the $n$-particle sector.

\begin{definition}[Anti-concentration of Boson sampling]\label{def:anticon_def}
    Let $U \sim U(m)$ be a Haar-random interferometer and let $S \sim \Phi_m^n$ denote a uniformly sampled output configuration. 

The output distribution of a boson sampler is said to be anti-concentrated if
\begin{equation}\label{eq:weak_ac_main}
    \Pr_{\substack{U \sim U(m)\\ S \sim \Phi_m^n}}
    \left[
    p_U(S) \geq \frac{\alpha}{|\Phi_m^n|}
    \right]
    \geq \gamma(\alpha),
\end{equation}
for some constants $\alpha,\gamma(\alpha) > 0$.

Moreover, the output distribution  is said to be \emph{weakly} anti-concentrated if $\gamma(\alpha)= \Omega(n^{-a})$ with $a = O(1)$.

\end{definition}

This formulation asserts that   a fraction $\gamma(\alpha)$  of  instances given by a pair $(U,S)$ have probabilities  at least a constant fraction of the mean value $1/|\Phi_m^n|$. Depending on the scaling of $\gamma(\alpha)$, one can distinguish different degrees of anti-concentration. The standard notion corresponds to $\gamma(\alpha) = \Theta(1)$, while a weaker form allows for polynomial decay.  

Anti-concentration on a constant fraction of instances imposes strong structural constraints on the output distribution. It prevents the distribution from being \emph{polynomially-sparse}, i.e., supported on only polynomially many outcomes \cite{Pashayan_2020}. Since such sparse distributions admit efficient classical sampling algorithms \cite{lim2025classicalalgorithmsestimatingexpectation}, this rules out these simulation strategies for typical interferometers.

A natural quantity often used for diagnosing anti-concentration is the normalized average outcome collision probability
\begin{equation}\label{eq:outcome_collision}
    P_2(m,n) := |\Phi_m^n| \sum_{S \in \Phi_m^n} \mathbb{E}_{U \sim U(m)}\!\left[p_U(S)^2\right].
\end{equation}
 This quantity measures the probability that two independent samples from the same output distribution coincide. Equivalently, it corresponds to the squared $\ell_2$-norm of the probability vector and thus quantifies the concentration of the distribution: $P_2(m,n) = 1$ for the uniform distribution and increases as the distribution becomes more peaked (reaching $|\Phi_m^n|$ in the extremal case). It also coincides with the ideal linear cross-entropy benchmarking score, providing an operationally relevant diagnostic of the distribution \cite{Ehrenberg_2025}.

Characterizing the scaling of $P_2$ directly implies  anti-concentration via inequalities such as Paley--Zygmund:

\begin{equation}\label{eq:PZ}
     \Pr_{\substack{U \sim U(m)\\ S \sim \Phi_m^n}}
    \left[
    p_U(S) \geq \alpha \mathbb{E}[ p_U(S)]
    \right] \geq (1-\alpha)^2 \frac{\mathbb{E}[ p_U(S)]^2}{\mathbb{E}[ p_U(S)^2]},
\end{equation}
where expected values are taken over Haar random interferometers and uniformly sampled output configurations. 

As detailed in \autoref{sec:P2}, 
the ratio  of moments  appearing in the right hand side of \autoref{eq:PZ} can be expressed as

\begin{equation}\label{eq:paley_back}
    \frac{\mathbb{E}[ p_U(S)]^2}{\mathbb{E}[ p_U(S)^2]} = \frac{1}{P_2(m,n)},
\end{equation}
so that $P_2(m,n)$ provides a direct quantitive measure of anti-concentration.

So far, analytical results on anti-concentration in boson sampling  has been restricted the dilute regime where the number of modes scales at least quadratically in the number of photons.

For Fock-state boson sampling,   weak  anti-concentration was established  in the seminal work of Aaronson and Arkhipov \cite[Theorem 54]{aaronson_computational_2011}, and later extended to Gaussian boson sampling in \cite{Ehrenberg_2025_tansition,Ehrenberg_2025}, both in the collision-free regime. Progress toward stronger notions of anti-concentration  has been made in  \cite{nezami2021permanentrandommatricesrepresentation} via higher-moments computation, however no complete analytical charcterization has been reached.

These  results rely on the fact that in the dilute regime where $m = \Omega(n^2)$, the output distribution is dominated by collision-free configurations as dictated by the so-called bosonic birthday paradox  \cite{aaronson_computational_2011,arkhipov_bosonic_2012}. In this regime, $P_2(m,n)$ can be approximated by restricting the sum in \autoref{eq:outcome_collision} to collision-free configurations, rather than all Fock states, as detailed in \cite{Ehrenberg_2025} and further discussed in \autoref{sec:P2}. 

Moreover, since the set of collision-free states is invariant under mode permutations and Haar invariance over $U(m)$ is preserved through the photonic homomorphism $\varphi(U(m))$, all collision-free output configurations contribute equally on average. Consequently, it suffices to evaluate the second moment for a single reference configuration, e.g., the collision-free input state $\ket{S_{\mathrm{cf}}}$ \cite{Ehrenberg_2025}.
This yields the approximation
\begin{equation}
    P_2(m,n) \simeq \binom{m}{n}^2 \, \mathbb{E}_{U \sim U(m)}\!\left[p_U(S_{\mathrm{cf}})^2\right].
\end{equation}
The remaining expectation can be evaluated using Gaussian approximations for submatrices of Haar-random unitaries,  due to the so-called \emph{hiding} property. This reduces the problem to computing moments of permanents of i.i.d.\ Gaussian matrices, which can be analysed using tools from random matrix theory \cite{aaronson_computational_2011}. We note that while the hiding property has been proven in regimes such as $n^3 = O(m)$ \cite{Leverrier_2018} and $n^2 = O(m)$ for random orthogonal matrices \cite{jiang_distances_2019}, it breaks down in the  experimentally relevant   \emph{saturated} regime, in which the number of modes scales linearly with the number of photons \cite{bouland_complexitytheoretic_2023,young_atomic_2024,Wang_2019}.

In low-mode regimes where photon collisions occur with significant probability, these simplifications no longer apply. Collision events become prevalent, and the hiding property breaks down. As a result, the second moments of output probabilities must be analyzed over all configurations to accurately evaluate $P_2(m,n)$.
At the same time, standard proof techniques for establishing the complexity-theoretic hardness of boson sampling also fail in this \emph{saturated} regime, where the number of modes scales linearly with the number of photons \cite{aaronson_computational_2011,bouland2025exponentialimprovementsaveragecasehardness}. In \cite{bouland_complexitytheoretic_2023}, new techniques are introduced to address this regime, based on analogous reductions and average-case hardness assumptions for approximating output probabilities (see \autoref{app:anticon_hardness} for more details). However, the robustness gap persists in this setting, and anti-concentration remains a key ingredient toward closing it.
Despite this progress, anti-concentration in the collision-dominated regime is currently supported only by numerical evidence. In the next section, we address this gap by establishing anti-concentration for boson sampling beyond the dilute regime.

\section{Main results}\label{sec:results}

 In this section, we present our main results on the evaluation of second moments and their application to anti-concentration in boson sampling. Our approach builds on the representation-theoretic framework developed in \autoref{sec:reptheory}, which provides a multiplicity-free decomposition of the fixed particle-number operator space into irreducible components together with an associated ladder structure. 

In \autoref{sec:iter_norms},  we show that the second moment of expectation values of particle-number-preserving observables for Haar-random interferometers admits a simple decomposition in terms of Hilbert--Schmidt norms of projections onto irreducible components, as detailed in \autoref{prop:sec_moment}. To evaluate these quantities, we introduce in \autoref{th:iterative_removal_main} a recursive projection method based on the underlying $\mathfrak{sl}_2$ structure, which allows us to isolate irreducible contributions without constructing explicit projectors. This leads to closed-form expressions for the corresponding projection norms given in \autoref{th:recursive_irrep_norm}, which can be efficiently computed for a broad class of observables, including Fock states.

In \autoref{subsec:anticon}, we apply these results to output probabilities of boson sampling circuits to establish anti-concentration. By expressing the normalized average outcome collision probability in terms of second moments, we obtain a closed-form expression valid for arbitrary system sizes in \autoref{th:P2}. Analysing its asymptotic scaling, given in \autoref{thm:asymptoticP2}, allows us to prove anti-concentration in regimes beyond the dilute regime, including in the linear regime where photon collisions are prevalent, as stated in \autoref{cj:scalingP}.

\subsection{The second moment of expectation values}
\label{sec:iter_norms}

We consider expectation values of particle-number–preserving observables of the form  
\begin{equation}\label{eq:linear_form}
    f_U(\rho,O) = \Tr[\rho \varphi(U)^\dagger O \varphi(U)].
\end{equation}

When restricting to states and observables with a fixed number of particles \(n\)—a physically well-motivated setting in passive linear optics, encompassing scenarios such as photon-number-resolving measurements or fixed-photon-number inputs—it suffices to consider the action of passive linear optics on the corresponding operator space \(W_n\), governed by the representation \(\varphi_n\) and its adjoint action \(\omega_n\), introduced in \autoref{eq:phi_n} and \autoref{eq:decomp_omega_m}, respectively. In this case, the expectation value reduces to  
\begin{align}
    f_U(\rho,O) &= \Tr[ \rho^{(n)}  \varphi_n(U)^\dagger O^{(n)}  \varphi_n(U)] \label{eq:linear_form_n}\\
    &= \dbra{O^{(n)}} \omega_n(U) \dket{\rho^{(n)}}\;,\label{eq:vectorized_loss}
\end{align}
where \(\rho^{(n)}\) and \(O^{(n)}\) denote the projections of the initial state and observable onto the \(n\)-particle Fock sector \(\mathcal{F}_n\) and $\dket{\cdot}$ denotes their vectorized form.

Exploiting the decomposition of \(W_n\) into irreducible components \(\lambda_k^{(n)}\), for \(0 \leq k \leq n\), given in \autoref{lemma:irrep_decomp}, we derive closed-form expressions for the second moment of expectation values of the form in \autoref{eq:linear_form}, averaged over Haar-random interferometers $U \in U(m)$. Since passive linear optics transformations are described by the unitary group $U(m)$, where $m$ denotes the number of modes, sampling Haar-random interferometers is experimentally feasible using standard photonic architectures \cite{mezzadri_how_2007, carolan_universal_2015, moralis-pegios_perfect_2024, clements_optimal_2016}. Consequently, analysing concentration phenomena under Haar randomness is directly relevant for realistic implementations.

 Doing so, we obtain the following proposition.

\begin{proposition}\label{prop:sec_moment}
    Consider a particle-number preserving bosonic observable $O \in \mathcal{W}$ and its  expectation value $f_U(\rho,O)$ of the form in \autoref{eq:linear_form} with respect to an  initial state $\rho$ on $m$ modes evolved under passive linear optics transformations. Further, assume that either the initial state or the observable acts exactly on $n$ photons.  Then, the second moment over  Haar-random interferometers $U \sim U(m)$ of $f_U(\rho,O)$ is given by  

    \begin{equation}\label{eq:sec_m_exp}
        \e[U \sim U(m)]{f_U(\rho,O)^2} =  \sum_{k=0}^n \frac{1}{d_{k}^{(n)}}\|P_k^{(n)}(\rho)\|_2^2 \|P_k^{(n)}(O)\|_2^2,
    \end{equation}
    where $P_k^{(n)}$ is a projection map onto the isotypic component $\lambda_k^{(n)}$ introduced in  \autoref{lemma:irrep_decomp} and $d_{k}^{(n)}$ is its dimension given by \autoref{eq:irrep_dim}. 
\end{proposition}

\autoref{prop:sec_moment} establishes that, for a particle-number-preserving observable, the  second moment of its expectation value over Haar random interferometers reduces to the computation of projection norms of the initial state and observable into the irreducible representations appearing in \autoref{lemma:irrep_decomp}. This result follows directly from twirling the composition of the adjoint representation and the $n$-particle photonic homomorphism over the group $U(m)$ by means of Schur's lemma \cite{bartlett_reference_2007}.
Specifically, averaging over Haar-random interferometers within a fixed photon-number sector
eliminates inter-irrep coherences, giving rise to the Hilbert--Schmidt norms of irrep projections
appearing in \autoref{eq:sec_m_exp}.
A detailed proof of \autoref{prop:sec_moment} is provided in \autoref{app:sec_moment_proof}.

To obtain a closed-form expression for the second moment, it remains to compute the Hilbert--Schmidt norms of the projections onto the irreducible components, i.e. $\|P_k^{(n)}(\rho)\|_2^2$  and $\|P_k^{(n)}(O)\|_2^2$.

A standard approach to constructing bases for irreducible components relies on Clebsch–Gordan theory, which decomposes the tensor product representation $\omega_n(U)$ into irreducible $U(m)$-modules via a change of basis defined by Clebsch–Gordan coefficients \cite{VilenkinKlimyk1995}. These coefficients provide explicit orthonormal bases for each subspace $\lambda_k^{(n)}$ and yield corresponding projection operators. 
While this construction gives a complete and explicit description of the decomposition, it operates at the level of basis elements and becomes increasingly cumbersome as the dimension of the irreducible component grows \cite{arienzo_bosonic_2024}. We defer the reader to \autoref{sec:CG} for more details about projections with Clebsch–Gordan coefficients.

 By exploiting the representation-theoretic characterization of the particle-number preserving operator space established in \autoref{sec:moments}, we introduce a recursive projection method that bypasses the explicit construction of such bases and leverages the underlying $\mathfrak{sl}_2$ ladder structure to isolate irreducible contributions directly. Rather than relying on Clebsch–Gordan decompositions, this approach is based on invariant maps that act diagonally on irreducible components—effectively isolating each irrep without requiring an explicit basis—thereby providing a basis-independent route to evaluating quantities of interest, most notably Hilbert–Schmidt norms.

\begin{proposition}[Iterative projection onto irreducible components]
\label{th:iterative_removal_main}
Let $O \in W_n$ be an observable acting on the $n$-particle Fock sector. Denote by $P_k^{(n)}(O)$ the projection of $O$ onto the irreducible components $\lambda_k^{(n)}$ for each $0 \leq k \leq n$ introduced in \autoref{lemma:irrep_decomp}.

Then,  the projections $P_k^{(n)}(O)$ can be computed recursively as
\begin{equation}
    P_k^{(n)}(O)
    =
    \frac{1}{\alpha_{k,k+1,n}}
    \!\left(
        \Dm^{\,n-k}\!\circ\!\Cm^{\,n-k}(O)
        -
        \sum_{r=0}^{k-1} \alpha_{r,k+1,n}\, P_r^{(n)}(O)
    \right),
\end{equation}
where $\Cm^{\,n-k}$ and $\Dm^{\,n-k}$ denotes the application of the lowering and raising map respectively introduced in \autoref{eq:C_map} and \autoref{eq:D_map}, iteratively for $n-k$ times. For all $0 \leq r < k \leq n$, the coefficients $\alpha_{r,k+1,n}$ are given by
\begin{equation}
    \alpha_{r,k+1,n}
    =
    \frac{(n-r)! \, (n+m+r-1)!}{(k-r)! \, (k+m+r-1)! }.
\end{equation}
\end{proposition}

 \autoref{th:iterative_removal_main} provides a recursive procedure to project a given operator acting on $n$ particles onto the irreducible components of the operator space. It follows from the Howe duality established in \autoref{th:howe_duality}, which yields the explicit decomposition of fixed particle-number operator subspaces in \autoref{eq:decomp_easy}. Concretely, the iterative application of the lowering map $\Cm$ progressively removes contributions from higher irreducible components. By subtracting the contributions of the lower ones, one isolates the desired irrep, which can then be lifted back to the appropriate particle-number sector via successive applications of the raising map $\Dm$. The proof of \autoref{th:iterative_removal_main} is provided in \autoref{app:iterative_removal}.

 The iterative projection procedure introduced in \autoref{th:iterative_removal_main} exploits the underlying representation-theoretic structure through the raising and lowering maps, enabling the direct isolation of irreducible components without explicit basis constructions or full projectors, and in particular avoiding the computation of Clebsch–Gordan coefficients. This recursive structure allows for the efficient evaluation of Hilbert–Schmidt norms and leads to the following corollary on irreducible component norms.

\begin{corollary}[Irrep norms]\label{th:recursive_irrep_norm}
Let $O \in W_n$ be an Hermitian operator acting on the $n$-particle Fock sector, and let $P_k^{(n)}(O)$ denote its projection onto the irreducible component $\lambda_k^{(n)}$ introduced in \autoref{lemma:irrep_decomp}. 
For each $0 \le l \le n$, define
\begin{equation}\label{eq:g_k_th_main}
    g_l(O)
    :=
    \Tr\!\left[\left(\Cm^{\,n-l}(O)\right)^2\right],
\end{equation}
where $\Cm^{\,n-l}$ denotes the $(n-l)$-fold application of the lowering map defined in \autoref{eq:C_map}.

Then, for each $0 \le k \le n$, the Hilbert--Schmidt norm of $P_k^{(n)}(O)$ is given by
\begin{equation}
    \Tr\!\left[\left(P_k^{(n)}(O)\right)^2\right]
    =
    \frac{1}{\alpha_{k,k+1,n}}
    \sum_{l=0}^{k}
    \frac{(-1)^{k-l}  g_l(O)}{(k-l)!  (k+l+m-1)_{k-l}}
    ,
\end{equation}
where $(x)_p := x(x+1)\cdots(x+p-1)$ denotes the Pochhammer symbol, and
\begin{equation}
    \alpha_{k,k+1,n}=(n-k)! \, (m+2k)_{\,n-k}.
\end{equation}
\end{corollary}

\autoref{th:recursive_irrep_norm} provides a closed-form expression for the Hilbert--Schmidt norm of each irreducible component of an operator acting on the $n$-photon sector. 
The quantities $g_l(O)$ defined in \autoref{eq:g_k_th_main} can be efficiently evaluated for a wide range of observables expressed in the $n$-photon Fock basis over $m$ modes. In particular, closed-form expressions of $g_l(O)$ for Fock states are derived in \autoref{lemma:puR_fock}, together with the proof of \autoref{th:recursive_irrep_norm}  in \autoref{app:irrep_norms}.

Combined with the second moment expression established in \autoref{prop:sec_moment}, the irrep norms computed in \autoref{th:recursive_irrep_norm} yield a compact closed form expression of the second moment of expectation values of particle-number-preserving observables. In particular, this enables the computation of second moments of output probabilities of a boson sampling experiment without relying on the hiding assumption \cite{nezami2021permanentrandommatricesrepresentation}.

\subsection{Anti-concentration of boson sampling}\label{subsec:anticon}

We now turn to the analysis of anti-concentration beyond the dilute regime, where photon collisions occur with significant probability and the hiding property no longer applies. In this setting, the output distribution cannot be reduced to collision-free configurations or treated using permutation invariance, and a direct evaluation of second moments across all output configurations is required.

As discussed in \autoref{sec:ac_back},  a natural quantity capturing anti-concentration is the normalized average outcome collision probability
\begin{equation}\label{eq:P2_recall}
    P_2(m,n) := |\Phi_m^n| \sum_{S \in \Phi_m^n} \mathbb{E}_{U \sim U(m)}\!\left[p_U(S)^2\right],
\end{equation}
which is governed by second moments of output probabilities and provides a proxy for anti-concentration via the Paley–Zygmund inequality, with further details in \autoref{sec:P2}. 
In principle, computing $P_2(m,n)$ requires summing these contributions over all output configurations, leading to a cumbersome expression.

Building on the framework developed in \autoref{sec:iter_norms}, we show that $P_2(m,n)$ can be evaluated exactly for arbitrary numbers of modes and photons.
Specifically, our analysis of second moments and reduced irrep purities yields a closed-form expression for $P_2(m,n)$ over Haar-random interferometers, valid in both the dilute and (intermediate) saturated regimes, which we state in the following \autoref{th:P2}. The detailed derivation is provided in \autoref{app:P2}.

\begin{theorem}[Normalized average outcome-collision probability]\label{th:P2}
The normalized average outcome-collision probability for linear optical
circuits, $P_2(m,n)$, defined in \autoref{eq:P2_recall}, is given by
 \begin{equation}\label{eq:P2_thm}
     P_2(m,n) =  \sum_{k=0}^n \frac{(m+k-1)_{k}}{(m+n)_k} {}_2 F_1(-k,n-k+1;2-m-2k;-1)\;.
 \end{equation}
 The shorthand $(x)_p:= x (x+1) \dots (x+p-1)$ denotes the Pochhammer symbol and ${}_2 F_1$ is the hypergeometric function defined as 
 \begin{equation}
     {}_2 F_1(-k,a;b;z) = \sum_{p=0}^k (-1)^p \binom{k}{p} \frac{(b)_p}{(c)_p} z^p,
 \end{equation}
 for any positive integer $k$.
\end{theorem}

\autoref{th:P2} is obtained by applying the general second-moment expression of \autoref{prop:sec_moment} to a collision-free input state and an arbitrary output configuration $S \in \Phi_m^n$ and then using the general irrep-norm expression established in \autoref{th:recursive_irrep_norm} evaluated on Fock states along with combinatorial identities. 
The resulting expression consists of a sum of contributions, each associated with an irreducible component of the $n$-particle operator space $W_n$. This decomposition of the outcome collision probability which is a relevant quantity in linear cross entropy benchmarking parallels analysis used in filtered randomized benchmarking, where noise contributions are diagnosed within each individual irrep \cite{arienzo_bosonic_2024}.

\begin{figure}[t!]
    \includegraphics[width=1.0\linewidth]{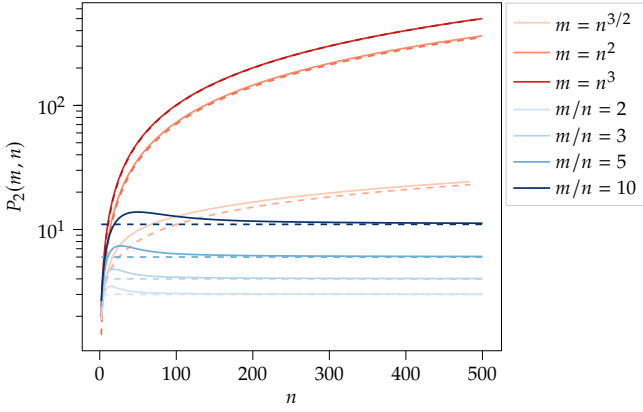}
    \caption{\justifying Numerical evaluation of the normalized average outcome collision probability $P_2(m,n)$ using \autoref{eq:P2_thm}, illustrating the asymptotic scaling established in \autoref{thm:asymptoticP2}. Solid lines show $P_2(m,n)$ in the dilute regime for $m = n^3, n^2$, the intermediate saturated regime for $m = n^{3/2}$, and the saturated (linear) regime for $m = c n$ with $c = 2, 3, 5, 10$. Dotted lines indicate the corresponding asymptotic scalings: $\Theta(n)$ for the dilute regime and  $m/n+1$ for the other regimes.}
    \label{fig:Concentration_with_bounds}
\end{figure}

Since the asymptotic behaviour of $P_2(m,n)$ as a function of $m$ and $n$ directly governs anti-concentration, we further analyse its scaling using the closed-form expression of \autoref{th:P2}, leading to the results stated in \autoref{thm:asymptoticP2}.

\begin{theorem}[Asymptotic behaviour of the normalized average outcome-collision probability]\label{thm:asymptoticP2}
Let $n \in \mathbb N$ and for constants $c\geq 1$ and $\beta \geq 1$, fix $m = cn^\beta$.  Then the normalized average outcome-collision probability $P_2(m,n)$ satisfies the following asymptotic behaviour.

\begin{enumerate}

\item In the dilute regime with $\beta \geq 2$, the collision probability grows linearly with $n$, i.e.
\begin{equation}
    P_2(m,n) = \kappa n + o(1),
\end{equation}
where $\kappa$ is a constant in $n$ that depends on $c$ and $\beta$.

\item In the intermediate saturated regime with
 $\beta<2$, the collision probability is dominated by the ratio $m/n$, i.e.
\begin{equation}
    P_2(m,n)=\frac{m}{n}+1+o(1),
\end{equation}
In particular, in the linear regime $m=cn$, it holds that 
\begin{equation}\label{eq:P2mIsCn}
    P_2(m,n) =  c + 1 + o(1).
\end{equation}
\end{enumerate}
\end{theorem}

The proof of \autoref{thm:asymptoticP2}, given in \autoref{app:P2}, exploits a connection between the hypergeometric function $_2F_1$ and moments of the Beta distribution,  and further provides refined control over the asymptotic convergence as detailed in \autoref{th:asymptotoic_P2_app}. More precisely, this relation enables us  to find the elegant expression 
 \begin{equation}
    P_2(m, n) = (n+m-1)\int_0^{\pi/2}\! d\theta \cos(\theta)^{n+m-2} \sin((n+1)\theta),
\end{equation}
and conclude with asymptotic expansions of Dawson function \cite{ARMSTRONG196761}.
\autoref{thm:asymptoticP2} identifies two distinct scaling behaviours for the averaged outcome-collision probability, depicted in \autoref{fig:Concentration_with_bounds}. For $m=cn^\beta$ with $\beta\geq 2$, we recover the scaling expected in the dilute regime \cite[Theorem 54]{aaronson_computational_2011}, but here derived without relying on the hiding property. By contrast, for $\beta<2$, the collision probability scales linearly with the ratio $m/n$, up to a polynomially decaying correction. This yields a new characterization of $P_2(m,n)$ in regimes where photon collisions are prevalent and, in particular, implies that it remains asymptotically constant in the linear regime.

With this in hand, plugging \autoref{eq:P2mIsCn} in the Paley--Zygmund inequality yields the following \autoref{cj:scalingP}, with the full proof given in \autoref{app:P2}.

\begin{corollary}[Anti-concentration of Boson Sampling in the saturated regime]\label{cj:scalingP}
    Let $m= c n^\beta $ with $c \geq 1$ and $1\leq \beta < 2$. Then in the asymptotic limit, the output  distribution of a Fock state boson sampler is anti-concentrated, i.e.
 \begin{equation}\label{eq:final_ac}
    \Pr_{\substack{S \sim \Phi_m^n \\U\sim U(m)}} \left[p_U(S) \geq  \frac{\alpha}{|\Phi_m^n|}\right] \geq   \frac{(1-\alpha)^2}{m/n+1 + o(1)}.
\end{equation}    
\end{corollary}

\autoref{cj:scalingP} shows that in the saturated regime corresponding to $m=c n^\beta $ with $c \geq 1$ and ${1\leq \beta < 2}$, anti-concentration holds with probability scaling as ${(1+m/n)^{-1}}$. For $\beta>1$, this yields probability scaling as $n^{\beta-1}$ corresponding to the weak-form of anti-concentration according to \autoref{def:anticon_def}. For $\beta=1$, the probability remains constant, yielding standard anti-concentration in the saturated regime. 
Therefore, this result resolves in the affirmative the anti-concentration conjecture of \cite{bouland_complexitytheoretic_2023}, complementing their proof of complexity-theoretic hardness for boson sampling in the linear regime.

\section{Conclusions}\label{sec:conclusion}

In this work, we developed a representation-theoretic framework for computing second moments of expectation values under passive linear-optical transformations. Our approach relies on the decomposition of the operator space into irreducible components and expresses second moments in terms of Hilbert--Schmidt norms of their projections. Crucially, this method avoids both the hiding property and the explicit use of permanent-based descriptions, providing a direct and general route to analytical expressions valid across all parameter regimes.

We then leveraged this framework to study output probabilities of Fock-state boson sampling circuits. By deriving a closed-form expression for the normalized average outcome collision probability and analysing its asymptotic scaling, we established anti-concentration in the linear regime, where photon collisions are prevalent. This provides analytical evidence supporting hardness conjectures in experimentally relevant settings and rules out classical simulation strategies based on sparsity of the output distribution. The same techniques can also be extended to boson sampling with Gaussian inputs, enabling the computation of moments and the characterization of anti-concentration in that setting across all mode regimes.

Beyond this application, our results provide a general toolkit for evaluating second moments of particle-number-preserving observables. These quantities play a central role in a variety of contexts such as linear cross-entropy benchmarking \cite{martinezcifuentes2024linearcrossentropycertificationquantum,Ehrenberg_2025}. They are also relevant for quantifying the variance of estimators arising randomized benchmarking \cite{arienzo_bosonic_2024} and  classical shadows \cite{thomas_shedding_2025} for passive linear optics.
In addition, they provide a natural framework for analysing concentration phenomena in variational algorithms based on passive linear-optical transformations, which we leave for future work.
 More broadly, the projection norms introduced here may serve as useful invariants or monotones in resource-theoretic settings under passive linear optics \cite{mamon2026orbitdimensionslineargaussian,migdal_multiphoton_2014,parellada2025liealgebraicinvariantsquantum,draux_invariants_2025}.

Several open directions follow from this work. A first question is whether similar representation-theoretic techniques can be extended to higher moments \cite{nezami2021permanentrandommatricesrepresentation,sierant2026theorymatchgatecommutant, braccia2026commutantfermionicgaussianunitaries}, which would provide access to finer statistical features of boson sampling distributions and further strengthen complexity-theoretic arguments. Existing approaches going through this route are currently limited to the dilute regime \cite{nezami2021permanentrandommatricesrepresentation}. Another important direction is to relax the fixed particle-number assumption on both observables and input states. This would require handling representations with multiplicities, but would enable the treatment of more general measurement schemes, such as homodyne and heterodyne detection, along with arbitrary bosonic input states. Finally, extending the framework to incorporate realistic experimental imperfections—such as photon loss and partial distinguishability—remains an important challenge for connecting these results to practical photonic implementations.

\emph{Note added.} Upon completion of this work, we became aware of independent work by Kolarovszki et al. \cite{kolarovszki_anticoncentration_xeb} that studies anti-concentration and linear cross entropy benchmarking in Boson Sampling using different group-theoretic techniques and proves anti-concentration in the saturated regime.
\medskip

\begin{acknowledgments}
The authors acknowledge Varun Upreti, Verena Yacoub, Eliott Mamon, Dmitry Grinko and Pierre-Emmanuel Emeriau for fruitful discussions. HM acknowledges support from the grant ANR-22-PETQ-0007.
HT acknowledges support from the European Commission as
part of the EIC accelerator program under the grant agreement 190188855 for
SEPOQC project, the Horizon-CL4 program under the grant agreement 101135288
for EPIQUE project, and the the CIFRE grant N. 2023/1746.
LM and EK acknowledge support from the EPSRC Quantum Advantage Pathfinder research (EP/X026167/1), and the Hub for Quantum Computing via Integrated and Interconnected Implementations (QCI3, EP/Z53318X/1) programs within the UK’s National Quantum Computing Center. ZH acknowledges support from the Sandoz Family Foundation-Monique de Meuron program for Academic Promotion.

\end{acknowledgments}
\begingroup
\makeatletter
\let\addcontentsline\@gobblethree % don't write TOC/minitoc entry
\endgroup
\medskip

\paragraph*{Code Availability.}

The code for irrep decomposition, and second moments calculation with all the figures presented is available at the following GitHub repository: \url{https://github.com/quantumsoftwarelab/Representation_Theory_Framework-Passive_Linear_Optics}.

\bibliographystyle{apsrev4-2}
\bibliography{references}
\appendix

\onecolumngrid
\newpage

\begin{center}
 {\Large \textbf{Appendix} }
\end{center}

\begingroup
\makeatletter
\@starttoc{atoc}
\makeatother
\endgroup

\section{Technical preliminaries}\label{sec:Preliminaries}

In this section, we recall standard-representation theoretic tools that will be relevant for our main analysis.

\begin{theorem}[$G$-twirling \cite{bartlett_reference_2007}]\label{th:G_twirl}
Let $G$ be a compact group and $H$ a Hilbert space. Let $g \mapsto T(g)$ be a unitary representation of $G$ on $H$. The \emph{$G$-twirling} operation $\mathcal{G}$ is defined as the CPTP map:
\begin{equation}
    \mathcal{G}(\rho) = \int_{G} T(g) \rho T^\dagger(g) \, \mathrm{d}g,
\end{equation}
where $\mathrm{d}g$ is the normalized Haar measure on $G$. 

Under the decomposition of $H$ into irreducible representations (irreps) of $G$,
\begin{equation}
    H = \bigoplus_{q} \mathcal{M}_q \otimes \mathcal{N}_q,
\end{equation}
where $\mathcal{M}_q$ is the carrier space of the $q$-th irrep and $\mathcal{N}_q$ is the corresponding multiplicity space, the action of $\mathcal{G}$ is given by
\begin{equation}
    \mathcal{G} = \sum_{q} (\Dm_{\mathcal{M}_q} \otimes \mathcal{I}_{\mathcal{N}_q}) \circ \mathcal{P}_q,
\end{equation}
where
\begin{itemize}
    \item $\mathcal{P}_q(\rho) = \Pi_q \rho \Pi_q$ is the projection superoperator onto the sector $H_q = \mathcal{M}_q \otimes \mathcal{N}_q$, with $\Pi_q$ being the orthogonal projection onto $H_q$.
    \item $\Dm_{\mathcal{M}_q}$ is the completely depolarizing channel on $\mathcal{M}_q$, defined by $\Dm_{\mathcal{M}_q}(\sigma) = \frac{\text{Tr}(\sigma)}{d_q} \mathbb{I}_{\mathcal{M}_q}$, where $d_q = \dim(\mathcal{M}_q)$.
    \item $\mathcal{I}_{\mathcal{N}_q}$ is the identity superoperator on the multiplicity space $\mathcal{N}_q$.
\end{itemize}
\end{theorem}

\begin{lemma}\label{lemma:lemma1}
 Let $G$ be a compact group and let $R : G \to \mathrm{GL}(V)$ be a completely reducible representation decomposing into $K$ inequivalent multiplicity free irreducible representations, i.e. 
 \begin{equation}
      V \stackrel{R(g)} {\simeq}  \bigoplus_{r=0}^K V_r \;.
 \end{equation}
 Let $\rho : G \to \mathrm{GL}(W)$ be an irreducible representation on $W$, and let $T :W \rightarrow V$ be a non-zero $G-$intertwiner, i.e.
 \begin{equation}
      T \circ \rho(g) = R(g) \circ T \;, \forall g \in G
 \end{equation}
 Then there exists a unique index $r \in \{0,\dots,K\}$ such that $W$ is isomorphic to $V_r$ through an isomorphism $\tilde{T}: W \rightarrow V_r$ and the intertwiner $T$ is a scalar multiple of $\tilde{T}$.

\end{lemma}

\begin{proof}
Since $T$ is an intertwiner, its kernel $\ker(T) \subseteq W$ is a $G$-invariant subspace of the irreducible representation $W$. Hence,
\begin{equation}
    \ker(T) \in \{ \{0\},\, W \}.
\end{equation}

If $\ker(T)=W$, then $T=0$, contradicting the assumption that $T\neq 0$. Therefore $\ker(T)=\{0\}$, and $T$ is injective. In particular, $T$ induces a $G$-equivariant isomorphism
\begin{equation}\label{eq:iso1}
    W \simeq \operatorname{Im}(T).
\end{equation}

Since $T$ is an intertwiner, the image $\operatorname{Im}(T)\subseteq V$ is also $G$-invariant. Moreover, the multiplicity free decomposition of the vector space $V$ into irreducible representations $V_r$
implies that there is a subset $\Lambda \subseteq \{0,\dots, K\}$ of size $|\Lambda| \geq 1$ such that 
\begin{equation}
    Im(T) \simeq \bigoplus_{r \in \Lambda} V_r
\end{equation}

Since $W \simeq \operatorname{Im}(T)$ and $W$ is irreducible, the direct sum above must contain exactly one summand. Hence $|\Lambda|=1$, and there exists a unique index $s$ such that
\begin{equation}\label{eq:iso2}
    W \simeq V_s.
\end{equation}

This implies that the irreducible representation $\rho$ acting on $W$ is equivalent to  the  sub-representation $R_s$ of $R$ acting irreducibly on  $V_s$. Moreover since the irreducible components $V_r$ are pairwise inequivalent for all $r \in \{0,\dots, K\}$, it follows that for every $ r\neq s$, the representation  $\rho$ and $R_r$ are inequivalent.

For each $r$, define the induced map $T_r$ as the projection of the intertwiner $T$ onto the irreducible subspace $V_r$, namely $T_r := P_r \circ T$. The equivariance of $T_r$ follows directly from the equivariance of $T$.

By Schur's Lemma, we then obtain
\begin{align}
    T_r = 0 \;, \forall r\neq s\;,
    T_s = \lambda \tilde{T}
\end{align}
where $\tilde{T} : W \rightarrow V_s$ is an isomorphism and $\lambda \in \mathbb{C}$.
Consequently,  $T= T_s$ and therefore $T$  is a scalar multiple of the isomorphism $\tilde{T}$.

\end{proof}

\section{Howe-duality for passive linear optics}\label{app:rep}

In this section, we develop the representation-theoretic framework governing the adjoint action of passive linear optics on the space of photon-number–preserving operators. We identify a $U(m)$--$\mathfrak{sl}_2(\mathbb{C})$ Howe duality that structures this action and enables a systematic decomposition of the operator space, which will be central to our analysis of second moments.

The presentation assumes only basic familiarity with representation theory. For completeness and to fix notation, we recall and re-derive several standard results adapted to our setting.

The section is organized as follows:

\begin{enumerate}
    \item In  \autoref{app:subsec_reptheory}, we introduce the adjoint action of passive linear-optical transformations on the space of photon-number-preserving bosonic operators and recall its decomposition into irreducible components, as established in \cite{arienzo_bosonic_2024}. 
    We then equip this space with a natural $\mathfrak{sl}_2(\mathbb{C})$ representation, defined via the raising and lowering maps introduced in \autoref{lemma:sl2_commutation}. In \autoref{th:commuting_action}, we prove that the resulting $\mathfrak{sl}_2(\mathbb{C})$ and 
  $U(m)$ actions  commute. As a consequence, the operator space decomposes into modules that simultaneously carry irreducible representations of $U(m)$  and representations of $\mathfrak{sl}_2(\mathbb{C})$.
\item In \autoref{app:properties}, we establish key structural properties of the raising and lowering maps. In particular, \autoref{prop:app_interwiner_Ck} shows that these maps act as \(U(m)\)-intertwiners between adjacent photon-number sectors. Furthermore, in \autoref{prop:adjoint_Ck_Dk}, we show that these maps are adjoint with respect to the Hilbert–Schmidt inner product.

\item In \autoref{app:exp_irreps}, we exploit these properties to explicitly identify the \(U(m)\)-irreducible representations appearing in the decomposition of each fixed photon-number operator sector. More precisely, in \autoref{prop:recursive_decomp}, we derive a recursive decomposition relating adjacent photon-number sectors. We then identify, in \autoref{prop:kernel_top_irrep}, the top irreducible component of each sector with  the kernel of the lowering map restricted to that space, thereby corresponding to 
the lowest-weight subspace of an appropriate $\mathfrak{sl}_2(\mathbb{C})$ representation. In \autoref{prop:lower_irreps}, we show that all remaining irreducible components within a photon-number sector can be generated by successive applications of the raising map to these primitive lowest-weight subspaces of different $\mathfrak{sl}_2(\mathbb{C})$ representations, yielding a complete description of the decomposition. Finally, in \autoref{prop:CD_DC}, we compute the eigenvalues of compositions of the raising and lowering maps on each $U(m)$--irreducible component, which will be crucial for constructing  the corresponding projections.

\item In \autoref{app:howe_proof}, we combine the structural and representation-theoretic results established above to prove the \(U(m)\)–\(\mathfrak{sl}_2(\mathbb{C})\) Howe duality on the space of photon-number–preserving bosonic operators.
\end{enumerate}

Henceforth, we denote $\mathfrak{sl}_2(\mathbb{C})$ simply by $\mathfrak{sl}_2$.

\subsection{Commuting $U(m)-\mathfrak{sl}_2$ structure on the photon-number-preserving bosonic operator space }\label{app:subsec_reptheory}

In this work, we consider photon-number-preserving bosonic operators acting on $m$ modes and further study their properties under passive linear optics transformations.

Formally, the algebra $\mathcal{W}$ of photon-number-preserving bosonic operators is spanned by normally ordered monomials containing an equal number of creation and annihilation operators denoted by $d$ and refered to as the operator degree,
\begin{equation}\label{eq:subalgebra_commute_app}
    \mathcal{W} = \text{Span} \{  a_{j_1}^\dagger \dots a_{j_d}^\dagger a_{i_1} \dots a_{i_d} \;|\;  d\geq 0\}\;.
 \end{equation}

 By construction, all operators in  $\mathcal{W} \subset \mathcal{F} \otimes  \mathcal{F}^*$ commute with the total photon number operator  $\hat{N}= \sum_{i=1}^m a_i^\dagger a_i$. Consequently, $\mathcal{W}$ preserves each photon-number sector $\mathcal{F}_n$ of the Fock space and decomposes accordingly as
 \begin{equation}\label{eq:photon_number_decomp}
      \mathcal{W}  = \bigoplus_{n \geq 0} W_n\;,
 \end{equation}
 where  $W_n:= \mathcal{F}_n \otimes \mathcal{F}_n^*$ denotes the  space of operators acting on the $n$-photon sector $\mathcal{F}_n$.

 A normally ordered monomial of degree $d$  acts non-trivially only on sectors $\mathcal{F}_n$ with $n \geq d$. This implies that, for a fixed photon number $n$, the subspace $W_n$ is spanned by monomials of degree up to $n$ projected onto the finite dimensional space $\mathcal{F}_n$, i.e.
 \begin{equation}\label{eq:span_Wn}
     W_n = \text{Span}\{ P_n (a_{j_1}^\dagger \dots a_{j_d}^\dagger a_{i_1} \dots a_{i_d}) P_n\;, d\leq n\}\;,
 \end{equation}
 where $P_n$ denotes the orthogonal projector  onto  $\mathcal{F}_n$ (in the Fock basis).

Passive linear optics transformations described by the group $U(m)$ acts on the Fock space through the photonic representation $\varphi$ \cite{aaronson_computational_2011}. This induces an adjoint action on the operator space $\mathcal{W}$, i.e.
\begin{align}\label{eq:adjoint_app}
     \omega(U)(X) = \varphi(U) (\cdot)  \varphi^\dagger(U) \;, X \in \mathcal{W} \;, U \in U(m)\;.
\end{align}

For convenience, we consider the vectorised representation of the adjoint map given by
\begin{equation}
    \omega(U) \simeq \varphi(U) \otimes \varphi^*(U)\;,
\end{equation}
acting of the vectorized form of the operator space $\mathcal{W}$.

Given that passive linear optics transformations preserve total photon number, this representation respects the photon-number decomposition in \autoref{eq:photon_number_decomp} and acts block diagonally as  
\begin{align}\label{eq:omega_app}
    \omega(U) \simeq \bigoplus_{n\geq 0} \varphi_n(U) \otimes  \varphi_n^*(U) := \bigoplus_{n\geq 0} \omega_n(U) \;, \quad
    \mathcal{W} = \bigoplus_{n \geq 0} W_n\;.
\end{align}

Each sector $W_n$ thus carries the $U(m)$--representation  given by 
\begin{equation}
    \omega_n(U) := \varphi_n(U) \otimes  \varphi_n^*(U).
\end{equation}

While $\varphi_n$ is an irreducible representation of $U(m)$ acting on the $n$-photon Hilbert space $\mathcal{F}_n$, the induced adjoint representation
$\omega_n$ on $W_n$ is indeed reducible. This reflects the presence of a richer internal structure in the photon-number preserving operator space $\mathcal{W}$, which can be resolved by decomposing each sector $W_n$ into irreducible components.

To derive such a decomposition, the authors of
Ref.\ \cite{arienzo_bosonic_2024} employ the Young diagram representation of  $\varphi_n$, corresponding to a  single row with $n$ boxes,
\begin{center}
\begin{equation}\label{eq:Young_irrep}
    \begin{tikzpicture}[baseline=(current bounding box.center), scale=0.5]
        
        \node at (-1.2, 0.5) {$\varphi_n \equiv$};

        \foreach \x in {0,1,2,3,4} {
            \draw (\x,0) rectangle (\x+1,1);
        }
        
        \draw [decorate, decoration={brace, amplitude=5pt, raise=2pt, mirror}] (0,0) -- (5,0) 
            node [midway, yshift=-12pt] {\small $n$};

    \end{tikzpicture} 
    \; .
    \end{equation}
\end{center}

Using Littlewood--Richardson's rule---a general tool to classify the decomposition of tensor product representations---they construct a decomposition of the representation  $\omega_n = \varphi_n \otimes \varphi_n^*$ acting on the subspace $W_n$  into $n+1$ irreducible representations. The irreducible component labeled by $0 \leq k \leq n$ is
associated with the Young diagram:

\begin{center}
\begin{equation}\label{eq:irrep_diagram}
    \begin{tikzpicture}[baseline=(current bounding box.center), scale=0.5]
        
        \node at (-4.5, -0.5) {$\lambda_k^{(n)} \equiv$};

        \draw (0,0) rectangle (1,1);
        \node at (1.5,0.5) {\small $\dots$};
        \draw (2,0) rectangle (3,1);
        
        \draw (3,0) rectangle (4,1);
        \node at (4.5,0.5) {\small $\dots$};
        \draw (5,0) rectangle (6,1);
        
        \draw (0,-2) rectangle (1,-1);
        \node at (1.5,-1.5) {\small $\dots$};
        \draw (2,-2) rectangle (3,-1);
        
        \node at (0.5, -0.5) {\small $\vdots$};
        \node at (1.5, -0.5) {\small $\ddots$};
        \node at (2.5, -0.5) {\small $\vdots$};
        
        \draw [decorate, decoration={brace, amplitude=4pt, raise=2pt}] (0,1) -- (2.9,1) 
            node [midway, yshift=12pt] {\small $k$};
        \draw [decorate, decoration={brace, amplitude=4pt, raise=2pt}] (3.1,1) -- (6,1) 
            node [midway, yshift=12pt] {\small $k$};
            
        \draw [decorate, decoration={brace, amplitude=4pt, raise=2pt, mirror}] (-0.2,1) -- (-0.2,-2) 
            node [midway, xshift=-28pt] {\small $m-1$};

        \node at (6.5, -0.5) {,};

    \end{tikzpicture}
    \end{equation}
\end{center}
and its dimension is obtained via the Weyl dimension formula, as detailed in \cite{arienzo_bosonic_2024} (See proposition T2),
\begin{equation}\label{eq:irrep_dim_app}
     d_k^{(n)} = \frac{2k+m-1}{m-1} \binom{k+m-2}{k}^2.
\end{equation}

Specifically, the fixed photon number operator subspace $W_n$ decomposes  under the representation $\omega_n$ as 
\begin{equation}\label{eq:comp_red}
    W_n \simeq \bigoplus_{k=0}^n \lambda_k^{(n)}\;.
\end{equation}

Combining this decomposition with the photon number decomposition of the full operator space $\mathcal{W}$  given in \autoref{eq:photon_number_decomp}, we obtain
\begin{equation}\label{eq:decomp1_app}
    \mathcal{W} \simeq \bigoplus_{n \geq 0} \bigoplus_{k=0}^n  \lambda_k^{(n)}.
\end{equation}

From the Young diagram description, one observes that for each fixed $k$, the irreducible representations $\lambda_k^{(n)}$ are independent of the
photon number $n$ for all $n \geq k$. 
This suggests
reorganizing the decomposition in \autoref{eq:decomp1_app} as
\begin{equation}\label{eq:ladder_decomp}
    \mathcal{W} \simeq \bigoplus_{k \geq 0} \bigoplus_{n \geq k}  \lambda_k^{(n)}  .
\end{equation}

This reorganization reveals that, for each fixed $k$, the representations $\lambda_k^{(n)}$ appearing in different photon-number sectors can be viewed as multiple copies of the same abstract $U(m)$-irrep, each occurring with multiplicity one in every sector $n \geq k$. This further suggests that these components should be organized into a single structure extending across photon-number sectors, rather than treated independently.

As we show below, this structure is captured by an additional symmetry of the operator space $\mathcal{W}$, which relates these copies in a systematic way. In particular, $\mathcal{W}$ carries an action of the Lie algebra $\mathfrak{sl}_2$, whose generators map between different photon-number sectors while preserving the $U(m)$-type label $k$. 
This perspective will be made precise along this section, where we formalize this organization in terms of $\mathfrak{sl}_2$ modules.

First, we begin by showing that the action of $\mathfrak{sl}_2$ on the bosonic operator space $\mathcal{W}$
 is generated by the raising and lowering maps $\Cm$ and $\Dm$, which we recall here for completeness:
 \begin{align}
     \Cm (\cdot) &: = \sum_{s=1}^m a_s (\cdot) a_s^\dagger \;,\label{eq:lowering_map_A}\\
     \Dm(\cdot)&:= \sum_{s=1}^m a_s^\dagger (\cdot) a_s \;.\label{eq:raising_map_A}
 \end{align}

With the following \autoref{lemma:sl2_commutation}, we prove that these maps along with their commutator $H= [\Cm,\Dm]$ do indeed satisfy the  $\mathfrak{sl}_2$  commutation relations, hence generating an $\mathfrak{sl}_2$--representation.

\begin{lemma}[$\mathfrak{sl}_2$ commutation relations of the lowering and raising maps]\label{lemma:sl2_commutation} Consider the lowering and raising maps introduced in \autoref{eq:lowering_map_A} and \autoref{eq:raising_map_A}, respectively, along with their commutator map $H$ defined as $H= [\Cm,\Dm]$. The triple $\{\Cm,\Dm,H\}$ generates a representation of the $\mathfrak{sl}_2$ Lie algebra on the operator space $\mathcal{W}$.
Namely, these maps satisfy the  $\mathfrak{sl}_2$ commutation relations, i.e.
\begin{align}\label{eq:com_maps_app}
    [H, \Dm]= 2 \Dm\;, [H, \Cm]= -2 \Cm \;, [\Cm,\Dm]=H\;.
\end{align}  

In particular, the $\mathfrak{sl}_2$ Lie algebra acts on $\mathcal{W}$ via the representation $\pi: \mathfrak{sl}_2 \rightarrow \rm End(\mathcal{W})$ given by 
\begin{equation}\label{eq:pi_sl2_rep_app}
    \pi(e)= E  \;, \pi(f)=- F \;, \pi(h)=H\;.
\end{equation}
where $\{e,f,h\}$ define the canonical basis of traceless $2 \times 2$ matrices given by
\begin{equation}
      e = \begin{pmatrix} 0 & 1 \\ 0 & 0 \end{pmatrix}\;, \quad
    f = \begin{pmatrix} 0 & 0 \\ 1 & 0 \end{pmatrix}\;, \quad
    h = \begin{pmatrix} 1 & 0 \\ 0 & -1 \end{pmatrix}\;.
\end{equation}
satisfying the canonical  commutation relations $[h,e]= 2e,  \; [h,f] = -2f, \;[e,f]=h$.
 
\end{lemma}

\begin{proof}
We begin the proof by establishing that the raising and lowering maps satisfy the commutation relations established in \autoref{eq:com_maps_app}.
   To this end, we start by showing that $[H, \Dm]= 2 \Dm$. To this end, we expand the nested commutator $[H, \Dm]$. Specifically, we get
    \begin{align}
        [H, \Dm]&= [[\Cm,\Dm], \Dm]\\
        &= [\Cm \Dm - \Dm \Cm, \Dm]\\
        &= \Cm \Dm^2 + \Dm^2 \Cm -2  \Dm \Cm  \Dm\;.
    \end{align}
    
    For $X \in \mathcal{W}$, the three terms appearing in the expansion of the commutator $[H, \Dm]$ are given by 
    \begin{align}
        \Dm  \Cm  \Dm (X) &= \sum_{i,j,k} a_i^\dagger a_k a_j^\dagger X a_j a_k^\dagger  a_i, \label{eq:term1}\\
        \Cm  \Dm^2(X) &= \sum_{i,j,k} a_k a_i^\dagger  a_j^\dagger X a_j   a_i a_k^\dagger,\label{eq:term2}\\
        \Dm^2\Cm(X)&=  \sum_{i,j,k}      a_i^\dagger  a_j^\dagger a_k X a_k^\dagger a_j   a_i\;.  \label{eq:term3}
    \end{align}
    In what follows, we further simplify the expression of the above maps by mainly using the canonical bosonic commutation relations:
    \begin{equation}\label{eq:BCR}
        [a_i,a_j^\dagger]= \delta_{i,j},\, [a_i^\dagger,a_j^\dagger]= [a_i,a_j]=0.
    \end{equation}

    The  term given  in \autoref{eq:term2} can be expressed as 
    \begin{align}
         \Cm \Dm^2(X) &= \sum_{i,j,k} a_k a_i^\dagger  a_j^\dagger X a_j   a_i a_k^\dagger \\
         &= \sum_{i,j,k} ( a_i^\dagger a_k + \delta_i^k \Id)  a_j^\dagger X a_j   ( a_k^\dagger a_i + \delta_i^k \Id)\\
         &= \Dm \Cm \Dm(X) + m \Dm(X) + \{ \hat{N}, \Dm(X)\}\\
         &=  \Dm \Cm \Dm(X) + m \Dm(X) + 2\hat{N}\Dm(X)\;,
    \end{align}
    where in the last equality we used the fact that $E(X) \in \mathcal{W}$, hence it commutes with the total number operator $\hat{N}$.

    Similarly, the  term in \autoref{eq:term3} can be written as 
    \begin{align}
         \Dm^2 \Cm(X)&=  \sum_{i,j,k}      a_i^\dagger  a_j^\dagger a_k X a_k^\dagger a_j   a_i   \\
         &=  \sum_{i,j,k}      a_i^\dagger  ( a_k a_j^\dagger - \delta _k^j \Id) X (a_j a_k^\dagger - \delta_j^k \Id )   a_i\\
         &= \Dm \Cm \Dm(X) + m \Dm(X) - \sum_{i,j=1}^m (a_i^\dagger X a_j a_j^\dagger  + a_i^\dagger a_j a_j^\dagger X a_i)\\
         &= \Dm \Cm \Dm(X) + m \Dm(X) - 2m \Dm(X) - 2 \Dm(\hat{N}X)\\
          &= \Dm \Cm \Dm(X) - m \Dm(X)  - 2 \Dm(\hat{N}X)\;.
    \end{align}

    Moreover, we have 
    \begin{align}
        \Dm ( \hat{N}X) &= \sum_k a_k^\dagger  \hat{N} X a_k\\
        &=   \sum_k   \hat{N} a_k^\dagger X a_k-  \sum_k    a_k^\dagger X a_k\\
        &= (\hat{N} - \Id ) \Dm(X)\;,
    \end{align}
    where we used the identity $[a_k^\dagger, \hat{N}]=-a_k^\dagger$.
    
    Hence, we finally get
    \begin{align}\label{eq:proof_D_Sl2}
        [H, \Dm]= [[\Cm,\Dm], \Dm] = 2 \Dm \;.
    \end{align}

     Now, we focus on proving the commutation relation involving the lowering map, i.e. $[H, \Cm]= -2 \Cm$.

    \begin{align}
        [H,\Cm]&= [[\Cm,\Dm],\Cm]\\
        &= [\Cm \Dm -\Dm \Cm,\Cm]\\
        &= 2 \Cm\Dm \Cm - \Dm \Cm^2-\Cm^2 \Dm\;. \label{eq:expand}
    \end{align}

    Similarly to the proof for the commutator with $\Dm$, we develop each term appearing in \autoref{eq:expand} aside.

    Starting by the term $\Dm \Cm^2$, we get
    \begin{align}
        \Dm \Cm^2 (X) &= \sum_k \sum_{i,j} a_k^\dagger  a_i a_j X    a_j^\dagger a_i^\dagger  a_k\\
        &= \sum_k \sum_{i,j} (a_k  a_i^\dagger - \delta_{k,i}\Id) a_j X    a_j^\dagger (a_i  a_k^\dagger - \delta_{k,i}\Id)\\
        &= \sum_k \sum_{i,j} a_k  a_i^\dagger a_j X    a_j^\dagger a_i  a_k^\dagger + \sum_k \sum_{i,j} \delta_{k,i} a_j X    a_j^\dagger -  \sum_k \sum_{i,j} \delta_{k,i} a_k  a_i^\dagger a_j X    a_j^\dagger - \sum_k \sum_{i,j} \delta_{k,i} a_j X    a_j^\dagger a_i  a_k^\dagger\\
        &= \Cm\Dm \Cm(X) + m \Cm(X) - (m \Id + \hat{N}) \Cm(X)- \Cm(X) (m \Id + \hat{N})\\
        &= \Cm \Dm \Cm(X) - m \Cm(X) - 2 \hat{N}\Cm(X)\;,
    \end{align}
    where in the last equality we used the fact that $\Cm(X) \in \mathcal{W}$, hence commuting with $ \hat{N}$.

    A similar derivation can be done for the term $\Cm^2 \Dm$. Precisely, we have
    \begin{align}
        \Cm^2 \Dm(X) &= \sum_{i,j} \sum_k a_j a_i   a_k^\dagger X a_k     a_i^\dagger a_j^\dagger\\
       &=  \sum_{i,j} \sum_k a_j (a_i^\dagger   a_k + \delta_{i,k} \Id) X (a_k^\dagger  a_i+ \delta_{k,i} \Id) a_j^\dagger\\
       &= \sum_{i,j} \sum_k a_j a_i^\dagger   a_k X a_k^\dagger  a_i a_j^\dagger + \sum_{i,j} \sum_k \delta_{i,k} a_i X a_i^\dagger + \sum_{i,j} \sum_k \delta_{i,k} a_j a_i^\dagger   a_k X a_j^\dagger + \sum_{i,j} \sum_k  \delta_{i,k} a_j X a_k^\dagger  a_i a_j^\dagger\\
       &= \Cm\Dm \Cm(X) + m \Cm(X) + 2\Cm(\hat{N}X)\\
       &= \Cm \Dm \Cm(X) + m \Cm(X) + 2 (\hat{N}+ \Id) \Cm(X)\;,
    \end{align}
    where in the last equality we used the commutation relation $[a_k,\hat{N}]= a_k$.

Having established that the rasing and lowering maps satisfy the  commutation relations in \autoref{eq:pi_sl2_rep_app}, this naturally defines a representation of the algebra $\mathfrak{sl}_2$ on the operator space $\mathcal{W}$ given by
\begin{equation}
     \pi(e)= \Dm  \;, \pi(f)=- \Cm \;, \pi(h)=H\;.
\end{equation}
where we recall that the triple $\{e,f,h\}$ form a basis of the Lie algebra $\mathfrak{sl}_2$ given by 
\begin{equation}\label{eq:conanoical_sl2_app}
      e = \begin{pmatrix} 0 & 1 \\ 0 & 0 \end{pmatrix}, \quad
    f = \begin{pmatrix} 0 & 0 \\ 1 & 0 \end{pmatrix}, \quad
    h = \begin{pmatrix} 1 & 0 \\ 0 & -1 \end{pmatrix}.
\end{equation}
satisfying the  commutation relations $[h,e]= 2e,  \; [h,f] = -2f, \;[e,f]=h$.  This concludes the proof
\end{proof}

The $\mathfrak{sl}_2$ action introduced in \autoref{lemma:sl2_commutation} provides a mechanism to further understand the copy structure of the $U(m)$-irreps $\lambda_k^{(n)}$ appearing across different photon-number sectors
, where the raising and lowering maps connect different photon-number sectors by adding and subtracting a photon in each mode, respectively.

Together with the adjoint $U(m)$ action given by the representation $\omega$, this endows the operator space $\mathcal{W}$ with a joint algebraic structure. In particular, the two actions of $\mathfrak{sl}_2$ and $U(m)$ commute, so each ladder $\bigoplus_{n \geq k} \lambda_k^{(n)}$ is invariant under $\mathfrak{sl}_2$ and hence organizes into $\mathfrak{sl}_2$-modules. The following theorem formalizes this structure.

 \begin{theorem}[Commuting $U(m)$ and $\mathfrak{sl}_2$ actions]\label{th:commuting_action}
The operator space $\mathcal{W}$ of photon-number-preserving operators, defined in \autoref{eq:subalgebra_commute_app}, carries two commuting actions:
\begin{itemize}
    \item a unitary representation of $U(m)$, denoted by $\omega(U)$, of the form in \autoref{eq:adjoint_app},
    \item a representation of $\mathfrak{sl}_2$ generated by the operators
    $\Cm$, $\Dm$, and $H := [\Cm,\Dm]$, defined in \autoref{eq:lowering_map_A} and \autoref{eq:raising_map_A}.
\end{itemize}
In particular, $\mathcal{W}$ admits a decomposition as a $U(m)$-module of the form
\begin{equation}
    \mathcal{W} \simeq \bigoplus_{\nu \in \widehat{U(m)}} V_\nu \otimes M_\nu\;,
\end{equation}
where $\widehat{U(m)}$ denotes the set of equivalence classes of irreducible representations of $U(m)$, $V_\nu$ are finite-dimensional irreducible $U(m)$-modules, and $M_\nu$ are the corresponding multiplicity spaces.

Moreover, each multiplicity space $M_\nu$ is invariant under the $\mathfrak{sl}_2$-action and thus carries an $\mathfrak{sl}_2$-module.
\end{theorem}

\begin{proof}

In \autoref{lemma:sl2_commutation}, we have already established that the lowering and raising  maps $\Cm$ and $\Dm$ introduced in \autoref{eq:lowering_map_A} and \autoref{eq:raising_map_A} generate along with their commutator $H=[\Cm,\Dm]$ an $\mathfrak{sl}_2$ representation on the operator space $\mathcal{W}$.

On the other hand, the group $U(m)$ acts on the same vector space via the adjoint representation of the photonic homomorphism $\omega(U)$, defined in \autoref{eq:adjoint_app}. For convenience, we henceforth consider the corresponding Lie algebra representation of the $U(m)$-representation $\omega(U)$, given by $d\omega:\mathfrak{u}(m) \rightarrow \rm End(\mathcal{W})$. For any $h\in \mathfrak{u}(m)$ and $O\in \mathcal{W}$, we have 
\begin{align}
    d\omega(h)\cdot O := [d\varphi(h),O] = \sum_{i,j=1}^m h_{i,j} [a_i^\dagger a_j,O]\;,
\end{align}
where $d\varphi(h) = \sum_{i,j=1}^m h_{i,j} a_i^\dagger a_j $ denotes the induced Lie algebra representation of the $U(m)$ representation  $\varphi$  \cite{garcia-escartin_multiple_2019}.

We now show that the  $\mathfrak{sl}_2$ action commutes with the $U(m)$ action. To this end, it suffices to prove that for all $h\in \mathfrak{u}(m)$, 
\begin{align}
    \Cm \circ d\omega(h) &=  d\omega(h) \circ \Cm\;,\label{eq:first_com}\\
     \Dm \circ d\omega(h) &=  d\omega(h) \circ \Dm\;,\label{eq:second_com}\\
      H \circ d\omega(h) &=  d\omega(h) \circ H\;.\label{eq:third_com}
\end{align}

Let $O \in \mathcal{W}$. 
We begin by proving the commutation relation in \autoref{eq:first_com}.  Precisely we have,
\begin{align}
    (\Cm \circ d\omega(h))\cdot O &= \sum_{s=1}^m a_s[d\varphi(h),O] a_s^\dagger\\
         &=   \sum_{i,j=1}^m h_{i,j} \sum_{s=1}^m  a_s [a_i^\dagger a_j,O] a_s^\dagger \\
    &=  \sum_{i,j=1}^m h_{i,j} \sum_{s=1}^m  (a_s a_i^\dagger a_j O a_s^\dagger -  a_s Oa_i^\dagger a_j  a_s^\dagger) \\
    &=  \sum_{i,j=1}^m h_{i,j} \sum_{s=1}^m  (( a_i^\dagger a_s + \delta_s^i \Id) a_j O a_s^\dagger -  a_s Oa_i^\dagger (a_s^\dagger a_j + \delta_s^j \Id )) \\
    &= \sum_{i,j=1}^m h_{i,j} \sum_{s=1}^m [a_i^\dagger a_j , a_s O  a_s^\dagger]\\
    &= \sum_{i,j=1}^m h_{i,j}  [a_i^\dagger a_j , \Cm(O)]\\
    &=  (  d\omega(h) \circ \Cm )\cdot O\;,
\end{align}
where in the fourth equality, we simply used the bosonic commutation relations.

Similarly for the raising map, we prove the commutation relation in \autoref{eq:second_com} as follows.
\begin{align}
    (\Dm \circ d\omega(h)) \cdot O  
    &= \sum_{s=1}^m \sum_{i,j=1}^m h_{i,j}  a_s^\dagger [a_i^\dagger a_j,O] a_s \\
    &= \sum_{s=1}^m \sum_{i,j=1}^m h_{i,j} (a_s^\dagger a_i^\dagger a_j O a_s - a_s^\dagger O a_i^\dagger a_j  a_s)\\
    &= \sum_{s=1}^m \sum_{i,j=1}^m h_{i,j} ( a_i^\dagger (a_j a_s^\dagger - \delta_{s,j} \Id) O a_s - a_s^\dagger O (a_s a_i^\dagger - \delta_{s,i} \Id) a_j )\\
     &= \sum_{s=1}^m \sum_{i,j=1}^m h_{i,j} ( a_i^\dagger a_j a_s^\dagger  O a_s - a_s^\dagger O a_s a_i^\dagger  a_j
     - \delta_{s,j} a_i^\dagger   O a_s 
     +  \delta_{s,i} a_s^\dagger O  a_j)\\
     &= \sum_{s=1}^m \sum_{i,j=1}^m h_{i,j}  [a_i^\dagger a_j , a_s^\dagger O a_s] -  \sum_{i,j=1}^m h_{i,j} a_i^\dagger O a_j - \sum_{i,j=1}^m h_{i,j} a_i^\dagger O a_j \\
     &= \sum_{i,j=1}^m h_{i,j} [a_i^\dagger a_j ,\Dm(O)]\\
     &=  (d\omega(h) \circ \Dm) \cdot O\;.
\end{align}

The commutation relation for the commutator map $H= [\Cm,\Dm]$ follows directly from \autoref{eq:first_com} and \autoref{eq:second_com}. Precisely, we have
\begin{align}
     (H \circ d\omega(h)) \cdot O &=  (\Cm \circ \Dm - \Dm \circ \Cm \circ d\omega(h)) \cdot O \\
     &= (\Cm \circ \Dm \circ d\omega(h)) \cdot O  - (\Dm \circ \Cm \circ d\omega(h)) \cdot O \\
     &= \Cm ((d\omega(h) \circ \Dm)\cdot O) - \Dm ((d\omega(h) \circ \Cm)\cdot O)\\
      &= (d\omega(h) \circ \Cm \circ \Dm  ) \cdot O  - (d\omega(h) \circ \Dm \circ \Cm ) \cdot O \\
     &= (d\omega(h) \circ H)\cdot O\;.
\end{align}

We have thus established that the actions of $U(m)$ and $\mathfrak{sl}_2$ commute on $\mathcal{W}$.

Since $U(m)$ is compact, the Peter--Weyl theorem \cite{hall_lie_2015} implies that  $\mathcal{W}$ decomposes as a direct sum of irreducible $U(m)$-modules: 
\begin{equation}
    \mathcal{W} \cong \bigoplus_{\nu \in \widehat{U(m)}} V_\nu \otimes M_\nu\;,
\end{equation}
where $\widehat{U(m)}$ denotes the set of inequivalent irreducible representations of $U(m)$, $V_\nu$ corresponds to the carrier space of the irreducible representation $\nu$  and $M_\nu$ is the corresponding multiplicity space, which can be infinite dimensional.

Moreover, since the $\mathfrak{sl}_2$-action on $\mathcal{W}$ commutes with the adjoint $U(m)$-action, each isotypic component $V_\nu \otimes M_\nu$ is invariant under $\mathfrak{sl}_2$. In particular, for any $x \in \mathfrak{sl}_2$, the map $\pi(x)$, defined in \autoref{eq:pi_sl2_rep}, acts as a $U(m)$-equivariant endomorphism of $V_\nu \otimes M_\nu$. By Schur's lemma, it follows that
\begin{equation}
    \pi(x)\big|_{V_\nu \otimes M_\nu}
    =
    \mathrm{id}_{V_\nu} \otimes \pi_\nu(x)\;,
\end{equation}
for a uniquely defined sub-representation $\pi_\nu(X)$ of $\pi(x)$.

Consequently, each multiplicity space $M_\nu$ carries a natural representation $\pi_\nu$ of $\mathfrak{sl}_2$.

\end{proof}

\autoref{th:commuting_action} shows that the operator space $\mathcal{W}$ decomposes into irreducible $U(m)$-representations $V_\nu$, with  corresponding multiplicity spaces carrying an $\mathfrak{sl}_2$-module structure.
This is suggestive of a dual-pair (Howe duality) setting, where two commuting actions generate mutual centralizers and yield a multiplicity-free decomposition into irreducible components \cite{debie2016howedualitypolynomialsolutions,brito2024quantummetaplectichoweduality}.

At this stage, however, we have only established the commutation between the $U(m)$ and $\mathfrak{sl}_2$ actions.
A full Howe duality would require the stronger double-commutant property, i.e., that the algebras they generate are mutual centralizers in $\mathrm{End}(\mathcal{W})$, which would in particular force each multiplicity space  $M_\nu$ to be  an irreducible $\mathfrak{sl}_2$-module and lead to a multiplicity-free decomposition. 

While we do not prove this property here, the structure will effectively emerge in the subsequent sections through an explicit ladder construction and a highest-weight analysis of the 
 $\mathfrak{sl}_2$ action.

\subsection{Properties of restricted raising and lowering maps}\label{app:properties}

In what follows, we further analyse the lowering and raising maps more locally by considering their restrictions on fixed photon-number operator subspaces. 
Formally, let $ \Cm_k$ and $ \Dm_k$ denote 
the following restricted maps,  $\forall k \geq 0$.
 \begin{align}
      \Cm_k &:= \Cm|_{W_k} \;,\label{eq:lowering_k}%\Cm \circ \mathcal{P}_k  \\
      \\
      \Dm_k &:= \Dm|_{W_{k-1}}\;.\label{eq:raising_k} %\Dm \circ \mathcal{P}_{k-1} 
 \end{align}
 Here, we adopt the convention that $\Dm_0$ is the zero map. Henceforth, we refer to these restricted maps simply as lowering and raising maps unless stated otherwise.

An important property of the lowering and raising maps which will be a key element in our analysis is the fact that they are $U(m)$-intertwiners.  This statement is formalized in \autoref{prop:app_interwiner_Ck} .

\begin{lemma}[The  lowering and raising maps are  $U(m)$-intertwiners]\label{prop:app_interwiner_Ck}
\hspace{1cm}

\begin{enumerate}
    \item The lowering map $\Cm_k : W_k \rightarrow W_{k-1}$ is a $U(m)-$ intertwiner with respect to the representations $\omega_k$ on $W_k$ and $\omega_{k-1}$ on $W_{k-1}$, introduced in \autoref{eq:decomp_omega_m}. Namely, the lowering map satisfy the following equivariance property, 
    \begin{equation}\label{eq:equiv_Ck}
        \Cm_k\circ \omega_{k} = \omega_{k-1} \circ \Cm_k\;.
    \end{equation}

    \item The raising map $\Dm_k : W_{k-1} \rightarrow W_k$ is a $U(m)-$ intertwiner with respect to the representations $\omega_k$ on $W_k$ and $\omega_{k-1}$ on $W_{k-1}$, introduced in \autoref{eq:decomp_omega_m}. Namely, the lowering map satisfy the following equivariance property, 
    \begin{equation}\label{eq:equiv_Dk}
        \Dm_k\circ \omega_{k-1} = \omega_{k} \circ \Dm_k\;.
    \end{equation}
\end{enumerate}

\end{lemma}

\begin{proof}
   In \autoref{th:commuting_action}, we showed that the action of $U(m)$ on the operator space $\mathcal{W}$ via the representation $\omega(U)$, defined in \autoref{eq:adjoint_app}, commutes with the $\mathfrak{sl}_2$ action generated by the raising and lowering maps introduced in \autoref{eq:lowering_map_A} and \autoref{eq:raising_map_A}. Equivalently, the global maps $\Cm$ and $\Dm$ are $U(m)$-intertwiners with respect to the adjoint representation $\omega(U)$.
Furthermore, by construction the lowering map $\Cm_k$ maps $W_k$ within $W_{k-1}$, while the raising map $\Dm_k$ maps $W_{k-1}$  within $W_k$.
Since the restrictions of $\Cm$ and $\Dm$ to fixed photon-number subspaces preserve this intertwining property, the result follows.
\end{proof}

Now, we show that the maps $\Cm_k$ and $\Dm_k$ are adjoint maps, which we denote henceforth by $\Dm_k = \Cm_k^\dagger$.
\begin{proposition}{(The raising and lowering maps are adjoint maps)}\label{prop:adjoint_Ck_Dk}
    For a fixed $k \geq 0$,  the lowering and raising maps $\Cm_k$ and $\Dm_k$ defined in \autoref{eq:lowering_k} and \autoref{eq:raising_k} are adjoint maps, i.e. \begin{equation}\label{eq:adjoint_maps}
    \langle A, \Cm_k(B) \rangle_{\mathcal{W}} = \langle \Dm_k(A), B \rangle_{\mathcal{W}} \;, \forall A \in W_{k-1}, B \in W_k
\end{equation}
where we equipped the operator subalgebra $\mathcal{W}$ and hence the operator subspaces $W_k$ with the Hilbert--Schmidt inner product.
\end{proposition}

\begin{proof}
    We equip the the operator subalgebra $\mathcal{W}$  with the standard Hilbert--Schmidt inner product, defined as $\langle A,B \rangle_{\mathcal{W}}:= \Tr_{\mathcal{F}}[A^\dagger B]$. 

Given the direct sum decomposition of $\mathcal{W}$ given in \autoref{eq:photon_number_decomp}, we have that for all $A\in W_k, B \in W_l $,
\begin{align}
    \langle A,B \rangle_{\mathcal{W}} &=0 &\text{ if } k \neq l\;,\\
    \langle A,B \rangle_{\mathcal{W}} &= \langle A,B \rangle_{W_k}:= \Tr_{\mathcal{F}_k}[A^\dagger B] &\text{ if } k = l\;.
\end{align}

   Let $A \in W_{k-1}, B \in W_k$. We have
   \begin{align}
       \langle A, \Cm_k(B) \rangle_{W_{k-1}} &= \Tr[A^\dagger \Cm_k(B)]\\
       &= \sum_{s=1}^m \Tr[A^\dagger a_s B a_s^\dagger]\\
       &=  \Tr[\Dm_k(A^\dagger)  B ]\\
       &=  \Tr[\Dm_k(A)^\dagger  B ] := \langle \Dm_k(A), B \rangle_{W_k}\;,
   \end{align}
   where in the second equality we used the definition of the lowering map given in \autoref{eq:lowering_k} along with the fact that $B \in W_k$. In the third equality, we used the cyclic property of the trace, the definition of the raising map given in \autoref{eq:raising_k} and the fact that $A \in W_{k-1}$. The last equality simply follows for the fact that the map $\Dm_k$ is hermitian.
   
    Therefore, the raising map $\Dm_k$ is the adjoint of the lowering map $\Cm_k$, which we denote by $\Dm_k = \Cm_k^\dagger$.
\end{proof}

\subsection{Analysis of irreducible representations of the photon-number preserving subalgebra}\label{app:exp_irreps}

In this section, we further characterize the irreducible components $\lambda_r^{(k)}$ arising in the decomposition of the operator subalgebra $\mathcal{W}$ of photon-number preserving bosonic operators, given in \autoref{eq:irrep_diagram}.
In particular, by exploiting the raising and lowering maps, we provide an alternative and more constructive realization of these irreps.

We first establish a recursive relation linking the fixed photon-number operator subspaces $W_k\;,\forall k \geq 0$.	
 Indeed, this relation is an immediate consequence of the adjointness of the raising and lowering maps, proven in \autoref{prop:adjoint_Ck_Dk}.

\begin{proposition}{(Recursive structure of the photon-number preserving subalgebra).}\label{prop:recursive_decomp}
For all $k \geq 0$, the operator subspace $W_k$ defined in \autoref{eq:irreps} decomposes under the $U(m)$-representation $\omega_k$  as 
    \begin{equation}
        W_k = \rm Ker(\Cm_k) \oplus \Dm_k(W_{k-1})\;.
    \end{equation}
\end{proposition}
\begin{proof}
 Using the fundamental theorem of linear algebra, we have
    \begin{align}
        W_k =\rm Ker(\Cm_k) \oplus\rm Im(\Cm_k^\dagger)  \;.  
    \end{align}
In \autoref{prop:adjoint_Ck_Dk}, we have shown that $\Cm_k^\dagger = \Dm_k$. Hence, we get

    \begin{align}
        W_k &= \rm Ker(\Cm_k) \oplus Im(\Dm_k)\\
        &= \rm Ker(\Cm_k) \oplus \Dm_k(W_{k-1})\;.
    \end{align}

Moreover, since the maps $\Cm_k$ and $\Dm_k$ are $U(m)$-intertwiners according to \autoref{prop:app_interwiner_Ck},  $\rm Ker(\Cm_k)$ and $\Dm_k(W_{k-1})$ are invariant subspaces of $W_k$,
    which concludes the proof.
\end{proof}

Now, we show that the outermost irrep of $W_k$ given by $\lambda_k^{(k)}$ is nothing but the kernel of the lowering map $\Cm_k$.

\begin{proposition}{(Top irrep of the $k$-photon operator subspace)} \label{prop:kernel_top_irrep}
 The kernel of the lowering map $\Cm_k$ coincides with the top irrep in $W_k$, i.e.
    \begin{equation}\label{eq:top_irrep_ker}
   {\rm Ker}(\Cm_k) = \lambda_k^{(k)}\;.
    \end{equation}    
\end{proposition}

\begin{proof}
    The first step in this proof consists in proving the the lowering map $\Cm_k$ is a surjective map. To do so, it is sufficient to prove that its adjoint map given by the raising map $\Dm_k$ is injective. Indeed, it is more convenient to study the raising map $\Dm_k$ because we don't have to deal with the subtlety of annihilating the vacuum state.

    Since $\{\ketbra{p}{p'}\}_{p,p' \in \mathcal{F}_{k-1}}$ forms a basis of $W_{k-1}$, proving the injectivity of the raising map  $\Dm_k$ boils down to proving that
\begin{align}
    \Dm_k(\ketbra{p}{p'}) = \Dm_k(\ketbra{q}{q'}) \implies p=q \text{ and } p'=q' \;.
\end{align}

Let $\ketbra{p}{p'}, \ketbra{q}{q'} \in W_{k-1}$ such that
\begin{equation}\label{eq:inj_cond}
    \Dm_k(\ketbra{p}{p'}) = \Dm_k(\ketbra{q}{q'})\;.
\end{equation}
 
The map $\Dm_k$ acts on $\ketbra{p}{p'}$ as follows
\begin{align}
    \Dm_k(\ketbra{p}{p'}) &= \sum_{s=1}^m a_s^\dagger \ketbra{p}{p'} a_s\\
    &= \sum_{s=1}^m \sqrt{p_s +1} \sqrt{p'_s +1}  \ketbra{p+e_s}{p'+e_s} \;.
\end{align}

Let us fix $t \in [|1,m|]$ and compute the Hilbert--Schmidt inner product between $\ketbra{p'+e_t}{p+e_t}$ and $\Dm_k(\ketbra{p}{p'}) - \Dm_k(\ketbra{q}{q'})$. From the condition in \autoref{eq:inj_cond} , we already know that this overlap should be zero, i.e.
\begin{align}
  0&=  \bra{p+e_t} \left(\Dm_k(\ketbra{p}{p'}) - \Dm_k(\ketbra{q}{q'})\right) \ket{p'+e_t} \\
  &= \bra{p+e_t} \left(\sum_{s=1}^m \sqrt{p_s +1} \sqrt{p'_s +1}  \ketbra{p+e_s}{p'+e_s} - \sqrt{q_s +1} \sqrt{q'_s +1}  \ketbra{q+e_s}{q'+e_s}\right) \ket{p'+e_t}\\
  &= \sqrt{p_t +1} \sqrt{p'_t +1} - \sum_{s=1}^m \sqrt{q_s +1} \sqrt{q'_s +1} \bra{p+e_t}\ket{q+e_s}  \bra{q'+e_s}\ket{p'+e_t}\;. \label{eq:D_inj}
\end{align}

We assume that $(p,p') \neq (q,q')$.
Let us now focus on identifying the conditions for which $\bra{p+e_t}\ket{q+e_s} =1$. For $s=t$, one can easily see that
\begin{equation}
    \bra{p+e_t}\ket{q+e_t} \bra{q'+e_t}\ket{p'+e_t}= \braket{p}{q} \braket{p'}{q'} = 0\;.
\end{equation}

For $s \neq t$, the overlap is equal to one if  
\begin{align} 
    q_r &= p_r \;, \forall r \notin \{s,t\} \label{eq:overlap1}\\
    q_s &= p_s - 1 \label{eq:overlap2}\\
    q_t &= p_t + 1\;. \label{eq:overlap3}
\end{align}
For a given pair $(p,q)$, there exists at most one pair $(s_{p,q}, t_{p,q})$ with $s_{p,q} \neq t_{p,q}$ and such that the conditions in \autoref{eq:overlap1}, \autoref{eq:overlap2}, and \autoref{eq:overlap3} are verified. If such a pair $(s_{p,q}, t_{p,q})$ exists, we use the notation $\Tilde{\delta}_{p,q}^{s^*,t^*} =1$.

By plugging in these conditions in the second term of \autoref{eq:D_inj}, we get
\begin{align}
     \sum_{s=1}^m \sqrt{q_s +1} \sqrt{q'_s +1} \bra{p+e_t}\ket{q+e_s}  \bra{q'+e_s}\ket{p'+e_t} &= \Tilde{\delta}_{p,q}^{s^*,t^*} \delta_{t,t^*} \delta_{s,s^*} \sqrt{q_{s^*} +1} \sqrt{q'_{s^*} +1}\\
     &=  \Tilde{\delta}_{p,q}^{s^*,t^*} \delta_{t,t^*} \delta_{s,s^*} \sqrt{p_{t^*} +2} \sqrt{p'_{t^*} +2}\;.\label{eq:eq1}
\end{align}

By plugging \autoref{eq:eq1} back in \autoref{eq:D_inj} , we get 
\begin{equation}
     0 = \sqrt{p_{t^*} +1} \sqrt{p'_{t^*} +1}  -  \Tilde{\delta}_{p,q}^{s^*,t^*} \delta_{t,t^*} \delta_{s,s^*} \sqrt{p_{t^*} +2} \sqrt{p'_{t^*} +2}  \;,
\end{equation}
which is a contradiction.
Consequently, this implies that $(p,p') = (q,q')$, thus the injectivity of the raising map $\Dm_k$ and the surjectivity of its adjoint $\Cm_k$.

    Having established the surjectivity of the lowering map $\Cm_k$, we are one step closer to proving the statement in \autoref{eq:top_irrep_ker}.

    Let us consider now the restriction of the contraction map to the irrep $\lambda_k^{(k)}$, i.e.
\begin{align}
    \Cm_k|_{\lambda_k^{(k)}} : \lambda_k^{(k)} \rightarrow W_{k-1}\;.
\end{align}
Clearly, the above map stays an intertwiner. Hence, its kernel ${\rm Ker}( \Cm_k|_{\lambda_k^{(k)}}) \subseteq \lambda_k^{(k)}$ is an invariant subspace. Or, $\lambda_k^{(k)}$ is an irrep. Consequently, we have $ {\rm Ker}( \Cm_k|_{\lambda_k^{(k)}}) \in \{\{0\}, \lambda_k^{(k)}\}$. 

Let us assume that $ {\rm Ker}( \Cm_k|_{\lambda_k^{(k)}})  = \{0\}$. This implies that ${\rm Im}(\Cm_k|_{\lambda_k^{(k)}}) \simeq \lambda_k^{(k)} $. By using \autoref{lemma:lemma1}, we get that there exists $0 \leq r \leq k-1$ such that $\lambda_r^{(k-1)} \simeq \lambda_k^{(k)}$, which is absurd (given that the lowering map is a non-zero map). Thus, we get $ {\rm Ker}( \Cm_k|_{\lambda_k^{(k)}})  = \lambda_k^{(k)}$ and consequently 
\begin{equation}\label{eq:ker_inclusion}
   \lambda_k^{(k)} \subseteq {\rm Ker}(\Cm_k)\;.
\end{equation}

Moreover, since $\Cm_k$ is surjective, we have 
\begin{align}
    \dim (Ker(\Cm_k)) &= \dim(W_k) - \dim(Im(\Cm_k))\\
    &= \dim(W_k) - \dim(W_{k-1})\\
    &= \dim(\lambda_k)\;.
\end{align}

Together with \autoref{eq:ker_inclusion}, we finally get ${\rm Ker}(\Cm_k) = \lambda_k^{(k)}$.
\end{proof}

\autoref{prop:kernel_top_irrep} identifies the top irreducible component of each subspace $W_k$ with the kernel of the lowering map $\Cm_k$. In particular, This shows that the outermost irrep contributions can be eliminated by iteratively applying the lowering map on a photon-number preserving bosonic observable. This observation forms the basis of the iterative projection procedure onto irreducible components.

Having established this characterization of the top component $\lambda_k^{(k)} \subset W_k$, we now turn to constructive properties of the lower irreducible components $\lambda_r^{(k)} $ for all $r<k$.

\begin{proposition}{(Irreducible component characterization).}\label{prop:lower_irreps}
For every fixed $r \geq 0$, the irreducible components $\lambda_r^{(i)}$ and $\lambda_r^{(j)}$ are isomorphic for all $i,j \geq r$. Moreover, for any $k \geq 0$ and every $r < k$, the irreducible component $\lambda_r^{(k)}$ is obtained from $\lambda_r^{(r)}$ by successive application of the raising maps, namely
\begin{equation}\label{eq:irrep_expression}
\lambda_r^{(k)} 
= \Dm_k \circ \Dm_{k-1} \circ \dots \circ \Dm_{r+1}\bigl(\lambda_r^{(r)}\bigr) \; .
\end{equation}
\end{proposition}

\begin{proof}

Let us  consider the restriction of the lowering map $\Cm_k$ to the irreducible subspace $\lambda_r^{(k)}$ for a fixed $r<k$. The equivariance property of the restricted map $\Cm_k|_{\lambda_r^{(k)}}$ follows from the equivariance  of $\Cm_k$.
Hence, the map $\Cm_k|_{\lambda_r^{(k)}}$ is a $U(m)-$intertwiner.

Since ${\rm Ker}(\Cm_k) = \lambda_k^{(k)}$, it follows that for every $r <k$ the restricted maps $\Cm_k|_{\lambda_r^{(k)}}$  are non zero intertwiners.

Fix $r <k$. We thus have a non zero $U(m)-$intertwiner  between the irreducible representation $\omega_r^{(k)}$ acting on $\lambda_r^{(k)}$ and the completely reducible representation $\omega_{k-1}$ acting on $W_{k-1}$, which decomposes into inequivalent multiplicity-free irreducible components  as $W_{k-1} = \bigoplus_{r=0}^{k-1} \lambda_r^{(k-1)}$. Invoking \autoref{lemma:lemma1}, we conclude that  there exists an index $0 \leq r^* \leq k-1$ and an isomorphism $T^{(k-1,k)}: \lambda_r^{(k)} \rightarrow \lambda_{r^*}^{(k-1)}$ such that 
\begin{align}\label{eq:C_k_iso}
    \Cm_k: &\lambda_r^{(k)} \rightarrow \lambda_{r^*}^{(k-1)} \\
        &X \quad  \mapsto \quad c_{r,k} T^{(k-1,k)}X (T^{(k-1,k)})^{-1}\;,
\end{align}
for some scalar $ c_{r,k} \in \mathbb{C}$.
 Finally, by dimension argument, we have that $r^* = r$.

An analogous statement holds for the raising map $\Dm_k$, which satisfy the same properties as the lowering map. In particular, the map $\Dm_k$ is a $U(m)-$intertwiner by \autoref{prop:app_interwiner_Ck} and it is non-zero since it is  injective (see the proof of \autoref{prop:kernel_top_irrep}).
Therefore, the same reasoning as for the lowering map applies. We conclude that, for each $r \leq k-1$,
\begin{align}\label{eq:D_k_iso}
   \Dm_k: &\lambda_r^{(k-1)} \rightarrow \lambda_r^{(k)} \\
        &X \quad  \mapsto \quad d_{r,k-1} (T^{(k-1,k)})^{-1} X T^{(k-1,k)} \;,
\end{align}
for some scalar $d_{r,k-1} \in \mathbb{C}$.

Using the fact that $\Dm_k$ is an isomorphism from $\lambda_r^{(k-1)}$ to $ \lambda_r^{(k)}$, it follows that the iterative composition of  raising maps yield also an isomorphism. Specifically,  the map $\Dm_k \circ \dots \circ \Dm_{j+1}$ for $r\leq j<k$ is an isomorphism  between $\lambda_r^{(j)}$ and $\lambda_r^{(k)}$.
This implies that the irreducible components $\lambda_r^{(j)}$ and $\lambda_r^{(k)}$ are isomorphic. Moreover, for $j=r$, we get
\begin{equation}
    \lambda_r^{(k)} = \Dm_k \dots D_{r+1}(\lambda_r^{(r)})\;.
\end{equation}

\end{proof}

\autoref{prop:lower_irreps} shows that the irreducible component $\lambda_r^{(k)} \subset W_k$ for $r < k $ can be generated from its isomorphic representation $\lambda_r^{(r)} \subset W_r$, acting on the $r$ photons subspace . Moreover,  \autoref{prop:kernel_top_irrep} establishes that the irreducible component $\lambda_r^{(r)}$ coincides with the kernel of the lowering map $\Cm_r$. Consequently, for a fixed photon subspace $W_k$, each lower irreducible component $\lambda_r^{(k)}$, with $0 \leq r <k$ the lower is obtained from the top irreducible component of  $W_r$ through successive application of the raising map.

With these ingredients in place, we are one step away from establishing the recursive trace removal procedure. In the following proposition, we show that the composed maps $\Cm_k \circ \Dm_k$ and $\Dm_k \circ \Cm_k$ act as scalar multiples of identity on each irreducible component $\lambda_r^{(k)}$ and we identify the associated scalars. 

\begin{proposition}{(Eigenvalues of raising and lowering maps on irreducible components).}\label{prop:CD_DC}    
         The maps $\Cm_{k+1} \circ \Dm_{k+1} : W_{k} \rightarrow W_k$ and  $\Dm_k \circ \Cm_k : W_{k} \rightarrow W_k$ act as scalar multiples of the identity on each  irreducible component  $\lambda_r^{(k)}$ for $0 \leq r\leq k$, i.e.
        \begin{align}
            \Cm_{k+1} \circ \Dm_{k+1}|_{\lambda_r^{(k)}} &= \beta_{r,k} \Id_{\lambda_r^{(k)}} \\
            \Dm_k \circ \Cm_k|_{\lambda_r^{(k)}} &= \beta_{r,k-1} \Id_{\lambda_r^{(k)}}\;,
        \end{align}
        where $\beta_{r,k} =  (k-r+1) (m+r+k)$.    \end{proposition}

\begin{proof}
     By combining the results from \autoref{eq:C_k_iso} and \autoref{eq:D_k_iso}, it follows that the maps $\Cm_{k+1} \circ \Dm_{k+1} : W_{k} \rightarrow W_k$ and  $\Dm_k \circ \Cm_k : W_{k} \rightarrow W_k$ are isomorphisms.
Moreover, they act as scalar multiples of identity non each irreducible component of $W_k$. Namely,

 \begin{align}
     \Cm_{k+1} \circ \Dm_{k+1}|_{\lambda_r^{(k)}} &= d_{r,k} c_{r,k+1} \Id_{\lambda_r^{(k)}}\\
     \Dm_k \circ \Cm_k |_{\lambda_r^{(k)}} &= d_{r,k-1} c_{r,k} \Id_{\lambda_r^{(k)}}\;.
 \end{align}
For convenience, we introduce the shorthand
\begin{equation}\label{eq:beta}
    \beta_{r,k}=d_{r,k} c_{r,k+1}\;.
\end{equation}

In what follows, we focus on computing the coefficients $ \beta_{r,k}$.
To this end, we exploit the $\mathfrak{sl}_2$ commutation structure established in  \autoref{lemma:sl2_commutation} for the raising and lowering maps acting on the full operator algebra $\mathcal{W}$. In particular, this structure allows one to generate eigenvectors of the map  $H= [\Cm,\Dm]$ together with their associated eigenvalues from a single eigenvector , as described in \autoref{lemma:sl2_commutation}.

    Let $X_r \in \lambda^{(r)}_r$. It follows from \autoref{prop:kernel_top_irrep} that $\Cm_r(X_r)=0$. Hence , we get
    \begin{align} \label{eq:H_on_ker}
        H(X_r) = [\Cm,\Dm](X_r) =  \Cm\Dm(X_r) =    \sum_{s,t=1}^m a_t  a_s^\dagger X_r a_s a_t^\dagger\;.
    \end{align}
  
  For $s,t =1 ,\dots,m$, we have
  \begin{align}\label{eq:commutators}
      a_t  a_s^\dagger X_r a_s a_t^\dagger =  ( a_s^\dagger a_t  + \delta_s^t \Id) X_r ( a_t^\dagger a_s + \delta_s^t \Id)   \;.   
  \end{align}
  Hence by summing over $s$ and $t$, we get
  \begin{align}
      \sum_{s,t=1}^m a_t  a_s^\dagger X_r a_s a_t^\dagger &= \sum_{s,t=1}^m ( a_s^\dagger a_t  + \delta_s^t \Id) X_r ( a_t^\dagger a_s + \delta_s^t \Id) \\
      &= \Dm\Cm(X_r) + \sum_{s=1}^m (  a_s^\dagger a_s X_r  + X_r + X_r a_s^\dagger a_s)\\
      &= m X_r + \{\hat{N}, X_r\}\\
      &= (m + 2r ) X_r := h_r X_r\;,\label{eq:useful_id}
  \end{align}
  where we introduced the shorthand $h_r= m + 2r$ in the last equality.
   This implies that $X_r$ is an eigenvector of $H$ with eigenvalue $h_r$.  

Here, one can easily show that if $X\in \mathcal{W}$ is eigenvector of $H$ with eigenvalue $h$, i.e. $H(X) = h X$, then $\Dm^k(X)$ is also an eigenvector of $H$ with eigenvalue $h+2k$.
Specifically, let  $X \in \mathcal{W}$ be an eigenvector of $H$  with eigenvalue $h$. We then show that $\Dm(X)$ is also an eigenvector of $H$. Precisely, we have
 \begin{align}
     H \Dm(X) &= [H,\Dm](X) + \Dm H(X) \\
    &= 2 \Dm(X) + h \Dm(X)\\
    &= (2+h) \Dm(X)\;,
 \end{align}
 where in the second equality we used the commutation relation from \autoref{eq:proof_D_Sl2}.

   Consequently, it follows  that the iterative application of the raising map on the vector $X_r \in \lambda^{(r)}_r$ also yields an eigenvector of $H$. 
   Specifically, by applying sequentially the raising map $k-r$ times, we have
   \begin{align}
       H (\Dm_{k} \dots \Dm_{r+1}(X_r)) = (m + 2r +2(k-r)) X_r= (m+2k) X_r
       :=  h_{k} X_r\;.
   \end{align}
Moreover, since irreducible components  $\lambda_r^{(k)}$ of the $k$-photon operator subspace are generated each from the isomorphic irreducible components from the $r$-photon operator subspace, then we have that the map $H$ acts as a scalar multiple of identity on each of these irreps where the scalar is given by $h_k$ i.e.
\begin{equation}
    H|_{\lambda_r^{(k)}} = h_k \Id_{\lambda_r^{(k)}}\;.
\end{equation}

Here, we notice that the eigenvalue on a given irreducible component $\lambda_r^{(k)}$  is independent of its index $r$ and only depend on the photon number $k$. 

Having identified the action of the the map $H= [\Cm,\Dm]$  restricted  on an irreducible component $\lambda_r^{(k)}$, we obtain a recursive relation of the coefficients $\beta_{r,k}$ introduced in \autoref{eq:beta}.
\begin{align}
    ( \beta_{r,k}- \beta_{r,k-1}) \Id|_{\lambda_r^{(k)}}  &=  \left(\Cm_{k+1}\Dm_{k+1}|_{\lambda_r^{(k)}} - \Dm_{k}\Cm_{k}|_{\lambda_r^{(k)}}\right) \\
    &= H_k|_{\lambda_r^{(k)}}\\
    &= h_k \Id|_{\lambda_r^{(k)}}\;,
\end{align}
hence the recursive relation
\begin{equation}
    h_k = \beta_{r,k}- \beta_{r,k-1}\;.
\end{equation}

By means of a telescopic sum, we obtain 
\begin{align}
    \sum_{j=r}^{k} h_{r,j}= \sum_{j=r}^{k} (\beta_{r,j}- \beta_{r,j-1} ) = \beta_{r,k}-  \beta_{r,r-1} = \beta_{r,k}\;,
\end{align}
where we used the fact that $\beta_{r,r-1} =0$.
Consequently, the expression of $\beta_{r,k}$ is given by 
\begin{equation}\label{eq:eigenvalues_irreps}
    \beta_{r,k} =  \sum_{j=r}^{k} m +2j = m(k-r+1) +2 \sum_{j=r}^{k} j = (k-r+1) (m+r+k)\;,
\end{equation}
which concludes the proof.

\end{proof}

\subsection{$U(m)$--$\mathfrak{sl}_2$ duality on the space of photon-number-preserving bosonic operators}\label{app:howe_proof}

Having identified the irreducible components of the operator space through the $\mathfrak{sl}_2$ ladder structure underlying its highest-weight irrep constructions, we are now ready to formally state the Howe-duality structure for the operator space $\mathcal{W}$. The following theorem is a formal version of \autoref{th:howe_duality} given in the main text. 

\begin{theorem}[{\texorpdfstring{$U(m)$--$\mathfrak{sl}_2$}{U(m)-sl2} Howe duality on $\mathcal W$}]\label{th:howe_duality_app}
Let $\mathcal W$ be the space of photon number preserving operators defined in \autoref{eq:subalgebra_commute}. Then the commuting actions of  $U(m)$ and $\mathfrak{sl}_2$ form a dual pair on $\mathcal{W}$. In particular, the space $\mathcal{W}$ admits a multiplicity-free decomposition into irreducible $U(m) \times \mathfrak{sl}_2$ modules of the form 

\begin{equation}
    \mathcal W = \bigoplus_{k\geq 0}\, \mathcal W_k\;,
    \qquad
    \mathcal W_k = \bigoplus_{n\geq k} \lambda_k^{(n)} \cong \lambda_k \otimes S_k\;,
\end{equation}
where $\lambda_k$ denotes an irreducible   $U(m)$-representation  whose multiplicity space $S_k$ corresponds to an irreducible $\mathfrak{sl}_2$ module generated by the highest weight space $\lambda_k^{(k)}$, given by
\begin{equation}\label{eq:highest_weight}
     \lambda_k^{(k)} = {\rm Ker}(\Cm|_{W_k})\;,
\end{equation}
and the full  $U(m) \times \mathfrak{sl}_2$ module denoted by $\mathcal{W}_k$ is generated by successive applications of the raising operator $\Dm$, yielding the ladder structure

\begin{equation}%\label{eq:ladder}
   \lambda_k^{(k)}  \xrightarrow{\Dm} \lambda_k^{(k+1)} \xrightarrow{\Dm} \lambda_k^{(k+2)} \xrightarrow{\Dm} \cdots %\lambda_r^{(n)}
\end{equation}
\end{theorem}

\begin{proof}
We begin by combining the decomposition of the photon-number-preserving operator space
\begin{equation}
    \mathcal{W} = \bigoplus_{n \geq 0} W_n\;,
\end{equation}
with the multiplicity-free decomposition of each fixed photon-number sector $W_n$ into irreducible $U(m)$-modules (\autoref{lemma:irrep_decomp}). Reorganizing the resulting double sum yields
\begin{equation}\label{eq:ladder_decomp_proof}
    \mathcal{W} \cong \bigoplus_{k \geq 0} \bigoplus_{n \geq k} \lambda_k^{(n)} 
    = \bigoplus_{k \geq 0} \mathcal{W}_k\;,
\end{equation}
where we introduced the subspaces
\begin{equation}\label{eq:MathcalW_k}
    \mathcal{W}_k := \bigoplus_{n \geq k} \lambda_k^{(n)}\;.
\end{equation}

By \autoref{lemma:irrep_decomp}, for each fixed $n$, the irreducible components $\lambda_k^{(n)}$ with $0 \leq k \leq n$ are pairwise inequivalent $U(m)$-modules. Moreover, by \autoref{prop:lower_irreps}, these modules satisfy
\begin{equation}\label{eq:D_tower}
    \lambda_k^{(n)} = \Dm^{\,n-k}\bigl(\lambda_k^{(k)}\bigr),\;
\end{equation}
and they are isomorphic as $U(m)$-modules, i.e.
\begin{equation}
    \lambda_k^{(n)} \cong \lambda_k^{(k)}
    \quad \text{ for all } n \geq k\;.
\end{equation}

Consequently, each $\mathcal{W}_k$ is generated from the primitive component $\lambda_k^{(k)}$ by repeated applications of the raising operator $\Dm$.

Furthermore, \autoref{prop:kernel_top_irrep} shows that
\begin{equation}
    \lambda_k^{(k)} = \rm{Ker}(\Cm|_{W_k})\;,
\end{equation}
so that $\lambda_k^{(k)}$ identifies the highest-weight space of the $\mathfrak{sl}_2$-action within $W_k$.

\medskip

Since all irreducible components $\lambda_k^{(n)}$ with fixed $k$ are isomorphic as $U(m)$-modules, we introduce an abstract irreducible representation $\lambda_k \cong \lambda_k^{(k)}$ and write
\begin{equation}
    \lambda_k^{(n)} \cong \lambda_k \otimes \mathbb{C} e_n\;,
\end{equation}
where $\mathbb{C} e_n$ labels the copy of $\lambda_k$ appearing in the sector $W_n$.

This yields the factorization
\begin{equation}
    \mathcal{W}_k \cong \lambda_k \otimes M_k\;,
    \qquad
    M_k := \bigoplus_{n \geq k} \mathbb{C} e_n\;,
\end{equation}
and therefore
\begin{equation}\label{eq:abstract_decomp}
    \mathcal{W} \cong \bigoplus_{k \geq 0} \lambda_k \otimes M_k\;.
\end{equation}

\medskip

We now relate this decomposition to the abstract one induced by the commuting $U(m)$- and $\mathfrak{sl}_2$-actions. By \autoref{lemma:sl2_commutation} and the Peter--Weyl theorem, we have
\begin{equation}\label{eq:decomp_commuting}
    \mathcal{W} \cong \bigoplus_{\nu \in \widehat{U(m)}} V_\nu \otimes M_\nu\;,
\end{equation}
where $V_\nu$ are irreducible $U(m)$-modules and $M_\nu$ the corresponding multiplicity spaces.

Comparing \autoref{eq:abstract_decomp} with \autoref{eq:decomp_commuting}, and using the multiplicity-free structure of each $W_n$, we identify
\begin{equation}
\widehat{U(m)} = \{\, \lambda_k \mid k \geq 0 \,\}\;,
\qquad
V_\nu \cong \lambda_k\;,
\qquad
M_\nu \cong M_k\;.
\end{equation}

It remains to identify the $\mathfrak{sl}_2$-structure on $M_k$. The operators $\Cm$ and $\Dm$ preserve each subspace $\mathcal{W}_k$ and act trivially on the $U(m)$-factor, hence induce operators on $M_k$. In the basis $\{e_n\}_{n \geq k}$, the action is given by
\begin{equation}
\Dm(e_n) \in \mathbb{C} e_{n+1}\;, 
\qquad
\Cm(e_n) \in \mathbb{C} e_{n-1}\;,
\end{equation}
with $\Cm(e_k)=0$. Thus $e_k$ is a lowest-weight vector, and $M_k$ is generated by repeated applications of $\Dm$. Moreover, since any non-zero $\mathfrak{sl}_2$-invariant subspace must contain $e_k$, it follows that $M_k$ is irreducible.

Therefore,
\begin{equation}
M_k \cong S_k\;,
\end{equation}
where $S_k$ is the irreducible highest-weight $\mathfrak{sl}_2$-module generated by $e_k$.

We conclude that
\begin{equation}
\mathcal{W}_k \cong \lambda_k \otimes S_k\;,
\qquad
\mathcal{W} \cong \bigoplus_{k \geq 0} \lambda_k \otimes S_k\;,
\end{equation}
where each summand is an irreducible $U(m)\times \mathfrak{sl}_2$-module. This establishes the multiplicity-free $U(m)$–$\mathfrak{sl}_2$ decomposition and completes the proof.
\end{proof}

\section{Projection onto irreducible components and their norms}\label{app:projections}

In this section, we briefly recall the standard approach based on Clebsch--Gordan coefficients for constructing projections onto irreducible components of $\mathcal{W}$, and introduce a recursive method that bypasses their explicit computation. Although both projection methods rely on similar representation-theoretic structure, the latter operates directly at the level of invariant maps and is better suited to our purposes.

\subsection{Projection with CG coefficients}\label{sec:CG}

The standard technique for computing projections into irreducible representations of the operator space is based on the Clebsch-Gordon transform. The Clebsch-Gordon decomposition appears in the decomposition of sums of angular momenta in quantum mechanics \cite{sakurai_modern_2021}. In particular, the Clebsch-Gordon transformation is a unitary transformation $C$ that realises the basis change 
\begin{equation}
    C W_n C^\dagger = \bigoplus_{k = 0}^n \lambda_k\;,
\end{equation}
of \autoref{eq:comp_red}. A given irreducible component $\lambda$ admits an orthonormal basis, identified by so-called Gelfand-Tsetlin (GT) patterns
\cite{moshinsky_gelfand_1966}, that is to say, the projector onto $\lambda_k$ is given by 
\begin{equation}
    P_k = \sum_{M \in GT(\lambda_k)} \ketbra{M}\;.
\end{equation}
These basis elements $\ket{M}$ can be expressed in the Fock basis \cite{arienzo_bosonic_2024} using Clebsh-Gordon coefficients as 
\begin{equation}
    \ket{M} \simeq \sum_{S,T \in \Phi_m^n} C_{S, T}^M \ket{S,T}\;,
\end{equation}
where $\simeq$ indicates a basis change (from GT patterns of $\varphi$ and $\bar\varphi$ to the Fock basis). Each of these coefficients can be computed numerically efficiently \cite{alex_numerical_2011}. However, as the number of photons and modes grows the projectors need exponential space to store in memory \cite{thomas_shedding_2025}, as they act on the vectorised operator space via matrix-multiplication. In particular, as Fock states have support on all irreps, computing their projection via this methods is inefficient. Nonetheless, improvements could be done if one is only interested in irreps of polynomial dimension, but this methods is not scalable in general. 

\subsection{Iterative projection procedure}\label{app:iterative_removal}

While the Clebsch--Gordan construction provides an explicit, basis-level realization of the irreducible decomposition, it relies on coefficient-level computations and becomes cumbersome at large system sizes. By contrast, the commuting $U(m)$ and $\mathfrak{sl}_2$ actions ensure, via Howe duality, that irreducible components can be characterized intrinsically. The recursive projection method leverages the same highest-weight structure but operates directly at the level of invariant maps, isolating irreducible contributions through the ladder structure without explicit computations of coefficients. This basis-independent formulation is particularly well suited for evaluating quantities such as Hilbert--Schmidt norms.

To this end, we combine the structural results established in \autoref{app:properties} and \autoref{app:exp_irreps} to derive a recursive projection procedure. The key observation is that the kernel of the lowering map coincides with the lowest-weight (primitive) irreducible component of the operator subspace. Consequently, successive applications of the contraction map remove contributions from the highest irreps, allowing one to iteratively isolate lower components.

We also introduce the following shorthands $\forall n \geq 0 $ and $ 0 \leq  j \leq n$:
\begin{align}
    \Cm_{ j+1 \rightarrow n } &:=   \Cm_{j+1} \circ  \dots \circ \Cm_{n-1} \circ \Cm_n \;, \label{eq:comp_lowering_notation}\\
    \Dm_{n\rightarrow j+1} &:= \Dm_n \circ\dots \circ\Dm_{j+2} \circ \Dm_{j+1}\;. \label{eq:comp_raising_notation}
\end{align}
where we adopt the convention that $\Cm_{ n+1 \rightarrow n }$ and $\Dm_{ n+1 \rightarrow n }$ are the identity maps.

\begin{theorem}\label{th:iterative_removal}
Consider an  observable $O$ acting on the $n$-photon Fock sector, i.e. $O\in W_n$. 
The component $O_j^{(n)}$ of the observable $O$ on each of the irreducible representations $\lambda_j^{(n)}$ introduced in \autoref{lemma:irrep_decomp} is iteratively given by
    \begin{align}
    O_j^{(n)}&= \frac{1}{\alpha_{j,j+1,n}} \left(\Dm_{n\rightarrow j+1} \circ \Cm_{j+1\rightarrow n} (O) - \sum_{r=0}^{j-1} \alpha_{r,j+1,n} O_r^{(n)} \right)\;,
\end{align}
where for all $r < j\leq n$, the coefficients $\alpha_{r,j+1,n}$ are given by 
\begin{equation}\label{eq:alpha_exp_th_app}
    \alpha_{r,j+1,n}  = \frac{(n-r)! (m+n+r-1)!}{(j-r)! (j+m+r-1)!}\;.
\end{equation}

\end{theorem}

\begin{proof}
    The proof follows by combining the propositions established in \autoref{app:exp_irreps}.

    First, using the recursive decomposition of the fixed photon-number operator subspaces established in \autoref{prop:recursive_decomp}, i.e. $W_n = Ker(\Cm_n) \oplus \Dm_n(W_{n-1})$, it follows that there exist operators $O_n^{(n)}\in \rm \lambda_n^{(n)} \subseteq W_n$ and  $\hat{O}^{(n-1)} \in W_{n-1}$ such that 
\begin{equation}
    O = O_n^{(n)} + \Dm_n(\hat{O}^{(n-1)})\;.
\end{equation}

In this decomposition,  $O_n^{(n)}$ corresponds to the  component of the observable supported on the highest irreducible subspace $\lambda_n^{(n)}$, while $\hat{O}^{(n-1)} \in W_{n-1}$ is an operator in the lower photon-number subspace whose image under the raising map $\Dm_n$ reconstructs the remaining components of $O$ in $W_n$. 
We emphasize that $\hat{O}^{(n-1)}$ should not be confused with the projection of the original observable onto the $(n-1)$-photon sector; precisely  $\hat{O}^{(n-1)} \neq P_{n-1} \Tilde{O} P_{n-1}$.  

Iterating this decomposition, we obtain that for every $j \leq n$ there exist operators $O_j^{(j)} \in \rm \lambda_j^{(j)} \subset W_j$ and $\hat{O}^{(j)} \in W_{j}$ such that the following recursive relation holds for all $j\leq n$.
\begin{equation}\label{eq:decomp_Wj}
    \hat{O}^{(j)} = O_j^{(j)} + \Dm_j(\hat{O}^{(j-1)})\;, 
     \hat{O}^{(n)}:= O \;.
\end{equation}
Consequently, the operator $O$ decomposes as follows,
\begin{align}
    O &= O_n^{(n)} + \Dm_n(O_{n-1}^{n-1}) + \Dm_n \Dm_{n-1}(O_{n-2}^{n-2})+ \dots + \Dm_{n \rightarrow j+1}(O_j^{(j)}) + \dots + \Dm_{n \rightarrow 2}(O_1^{(1)})+ \Dm_{n \rightarrow 1}(O_0^{(0)}) \\
    &:= \sum_{r=0}^n \Dm_{n \rightarrow r+1}(O_r^{(r)})\;.\label{eq:obs_decomp}
\end{align}
Using the construction of the irreducible components by the iterative application of the raising map established in \autoref{prop:lower_irreps}, each element in the above sum exactly coincides with the original observable components. Specifically, the component of the original observable $O \in W_n$ on the irreducible component $\lambda_j^{(n)}$ is given by
\begin{equation}\label{eq:form_irrep_comp}
    O_j^{(n)} = \Dm_{n \rightarrow j+1}(O_j^{(j)})\;.
\end{equation}

   Given that the lowering map annihilates the highest irrep component according to \autoref{prop:kernel_top_irrep}, the consecutive application of the lowering map $n-j$ times eliminates all highest irrep components indexed from  $j+1$ onwards. This observation forms the basis of the iterative trace removal procedure. Precisely, we have 
   \begin{align}
       \Cm_{ j+1  \rightarrow n} (O) &= \sum_{r=0}^n \Cm_{j+1  \rightarrow n}\Dm_{n \rightarrow r+1}(O_r^{(r)})\\
       &= \sum_{r=0}^{j} \Cm_{j+1  \rightarrow n}\Dm_{n \rightarrow r+1}(O_r^{(r)}) + \sum_{r=j+1}^{n} \Cm_{j+1  \rightarrow n}\Dm_{n \rightarrow r+1}(O_r^{(r)})\\
       &= \sum_{r=0}^{j} \Cm_{j+1  \rightarrow n} \Dm_{n \rightarrow j+1} \Dm_{j \rightarrow r+1}(O_r^{(r)}) + \sum_{r=j+1}^{n} \Cm_{j+1  \rightarrow r} \Cm_{r+1  \rightarrow n}\Dm_{n \rightarrow r+1}(O_r^{(r)})\\
       &= \sum_{r=0}^{j} \Cm_{j+1  \rightarrow n} \Dm_{n \rightarrow j+1}( \Dm_{j \rightarrow r+1}(O_r^{(r)}))\;, \label{eq:contraction_O}
   \end{align}
   where in the first equality, we substitute the observable $O$ with its decomposition given in \autoref{eq:obs_decomp}.  In the last equality, the second sum vanishes. Indeed, the map  $\Cm_{r+1  \rightarrow n}\Dm_{n \rightarrow r+1}$ acts as a scalar multiple of the identity on $O_r^{(r)} \in \lambda_r^{(r)} \subset W_r$. Consequently, the resulting operator remains in $ \lambda_r^{(r)}$ and a single subsequent application of the lowering map $\Cm_r$ suffices to eliminate this term. Hence all terms with $r \geq j$ are eliminated.
   
   For the terms appearing in the sum in \autoref{eq:contraction_O},  the map $\Dm_{j \rightarrow r+1}$ applied to the operator $O_r^{(r)}\in \lambda_r^{(r)}\subset W_r$ raises it to  its isomorphic component $\lambda_r^{(j)}$ in the $j$-photon number operator subspace.  Then, the map $ \Cm_{j+1  \rightarrow n} \Dm_{n \rightarrow j+1}$ simply acts as a scalar multiple of identity on the operator $ \Dm_{j \rightarrow r+1}(O_r^{(r)})$.

    We therefore focus on determining the scalar associated to the map $ \Cm_{j+1  \rightarrow n} \Dm_{n \rightarrow j+1}$ for every $0\leq j \leq n$. For $0\leq r \leq j-1$, we introduce the following notation
    \begin{equation}\label{eq:alpha_coeffs}
        \Cm_{j+1  \rightarrow n} \Dm_{n \rightarrow j+1}|_{\lambda_r^{(j)}} = \alpha_{r,j+1,n} \Id_{\lambda_r^{(j)}}\;.
    \end{equation}
    Plugging in this identity 
    in \autoref{eq:contraction_O} yields
 \begin{align}
      \Cm_{ j+1  \rightarrow n} (O) &=  \sum_{r=0}^{j} \Cm_{j+1  \rightarrow n} \Dm_{n \rightarrow j+1}( \Dm_{j \rightarrow r+1}(O_r^{(r)})) \\
      &=  \sum_{r=0}^{j} \alpha_{r,j+1,n}  \Dm_{j \rightarrow r+1}(O_r^{(r)}) \;,
 \end{align}

    To compute the coefficients $\alpha_{r,j+1,n}$ , we invoke  \autoref{prop:CD_DC}, which gives the eigenvalues of the map $\Cm_{i+1} \Dm_{i+1}$ on the irreducible components $\lambda_r^{(i)}$ for every $r\leq i$. Specifically, 
    \begin{equation}\label{eq:eigenvalue_beta}
        \Cm_{i+1} \Dm_{i+1}|_{\lambda_r^{(i)}} = \beta_{r,i} \Id_{\lambda_r^{(i)}}\;,
    \end{equation}
    where $\beta_{r,i} =  (i-r+1) (m+r+i)$.

    Therefore, determining the  coefficients $ \alpha_{r,j+1,n}$ introduced in \autoref{eq:alpha_coeffs} reduces to iteratively applying the above relation.

    Let  $X_r^{(j)} \in \lambda_r^{(j)}$ with $r \leq j$. Then
    \begin{equation}
         \Cm_{j+1} \Dm_{j+1} (X_r^{(j)}) = \beta_{r,j} X_r^{(j)}\;,
    \end{equation}
    and 
    \begin{equation}
        \Cm_{j+1} \Cm_{j+2} \Dm_{j+2} \Dm_{j+1} (X_r^{(j)}) = \Cm_{j+1} (\Cm_{j+2} \Dm_{j+2} (\Dm_{j+1} (X_r^{(j)}))) = \beta_{r,j+1} \Cm_{j+1} \Dm_{j+1} (X_r^{(j)})  =   \beta_{r,j} \beta_{r,j+1} X_r^{(j)}\;.
    \end{equation}

    Proceeding inductively, we obtain
\begin{align} 
     \Cm_{ j+1 \rightarrow n } \Dm_{n\rightarrow j+1} (X_r^{(j)}) = \left(\prod_{k=j}^{n-1} \beta_{r,k}\right)  X_r^{(j)} \;.
\end{align}

Using the explicit expression of $\beta_{r,k}$, we find
\begin{align}
     \alpha_{r,j+1,n} &= \prod_{k=j}^{n-1} \beta_{r,k}\\
     &= \prod_{k=j}^{n-1} (k-r+1) (m+r+k)\\
     &= \prod_{k=j-r+1}^{n-r} k  \prod_{k=j+m+r}^{n+m+r-1} k \\
     &= \frac{(n-r)! (m+n+r-1)!}{(j-r)! (j+m+r-1)!}\;. \label{eq:alpha_exp}
\end{align}
 
The expression of the component $O_j^{n}$ can be therefore derived by applying the raising maps $\Dm_{ n  \rightarrow j+1}$ on top of \autoref{eq:contraction_O} and identifying the irreducible components of the original observable using the form in \autoref{eq:form_irrep_comp}. Precisely, we have

\begin{align}
   \Dm_{ n  \rightarrow j+1} \Cm_{ j+1  \rightarrow n} (O) 
      &=  \sum_{r=0}^{j} \alpha_{r,j+1,n} \Dm_{ n  \rightarrow j+1} \Dm_{j \rightarrow r+1}(O_r^{(r)})\\
       &=  \sum_{r=0}^{j} \alpha_{r,j+1,n} \Dm_{ n  \rightarrow r+1}(O_r^{(r)})\\
      &= \alpha_{j,j+1,n} \Dm_{ n  \rightarrow j+1}(O_j^{(j)}) + \sum_{r=0}^{j-1} \alpha_{r,j+1,n}  \Dm_{n \rightarrow r+1}(O_r^{(r)}) \label{eq:proj_without_D}\\
      &= \alpha_{j,j+1,n} O_j^{(n)} + \sum_{r=0}^{j-1} \alpha_{r,j+1,n}  O_r^{(n)}\;.
\end{align}

 Consequently, the expression of $ O_j^{(n)}$ is given by
   \begin{align}
    O_j^{(n)}= \frac{1}{\alpha_{j,j+1,n}} \left(\Dm_{n\rightarrow j+1} \circ \Cm_{j+1\rightarrow n} (O) - \sum_{r=0}^{j-1} \alpha_{r,j+1,n} O_r^{(n)} \right)\;.
\end{align}
This concludes the proof.

\end{proof}

\subsection{Closed form expression of irrep norms}\label{app:irrep_norms}

Having established the iterative projection procedure and its theoretical underpinnings in \autoref{th:iterative_removal}, we can  now proceed to compute the Hilbert--Schmidt norm of irreducible components of a photon-number preserving operator appearing in the expression of the second moment in \autoref{prop:sec_moment}.

\begin{theorem}{(Irrep norms).}\label{th:recursive_irrep_norm_app}
Let $O$ be an observable acting on the $n$-photon Fock sector, i.e., $O \in W_n$. 
For every $0 \le k \le n$, define
\begin{equation}\label{eq:g_k_th}
    g_k(O) = \Tr\!\left[\big(\Cm_{k+1 \rightarrow n}(O)\big)^2\right]\;,
\end{equation}
where $\Cm_{k+1 \rightarrow n}$ denotes the composition of lowering maps introduced in \autoref{eq:comp_lowering_notation}. 

Let $O_j^{(n)}$ denote the component of $O$ supported on the irreducible subspace $\lambda_j^{(n)}$ (see \autoref{lemma:irrep_decomp}). 
Then the expression of Hilbert--Schmidt norm of $O_j^{(n)}$ is given by
\begin{equation}
   \Tr\left[(O_j^{(n)})^2\right] = \frac{1}{\alpha_{j,j+1,n}} \sum_{l=0}^k \frac{(-1)^{k-l}}{(k-l)! (k+l+ m-1)_{k-l}} g_l\;,
\end{equation}
where $(x)_p$ denotes the Pochhammer symbol defined as $(x)_p := x (x+1) \dots (x+p-1)$ and $\alpha_{j,j+1,n} = \frac{(n-j)! (m+n+j-1)!}{(m+2j-1)!}$.
\end{theorem}

\begin{proof}
    
To compute the Hilbert--Schmidt norm of the observable $O\in W_n$ projections onto the irreducible components $\lambda_j^{(n)}$ for $j \leq n$ denoted by $O_j^{(n)}$, we begin by invoking its expression derived through the iterative trace removal method in \autoref{th:iterative_removal}.

\begin{align}
    \Tr[(O_j^{(n)})^2]&= \Tr[\left(\frac{1}{\alpha_{j,j+1,n}} \left(\Dm_{n\rightarrow j+1} \circ \Cm_{j+1\rightarrow n} (O) - \sum_{k=0}^{j-1} \alpha_{k,j+1,n} O_k^{(n)} \right)\right)^2]\\
    &= \frac{1}{\alpha_{j,j+1,n}^2} \Tr[\left(\Dm_{n\rightarrow j+1}\left( \Cm_{j+1\rightarrow n} (O) - \sum_{k=0}^{j-1} \alpha_{k,j+1,n}\Dm_{j \rightarrow k+1}(O_k^{(k)}) \right)\right)^2]\;, \label{eq:exp1}
\end{align}
where in the last equality , we factored out the map $\Dm_{n\rightarrow j+1}$ since the lower components are given by $O_k^{(n)} = \Dm_{n\rightarrow j+1}(\Dm_{j \rightarrow k+1}(O_k^{(k)}))$. We also recall that
\begin{equation}
   O_j^{(j)} = \frac{1}{\alpha_{j,j+1,n}}  \left(\Cm_{j+1\rightarrow n} (O) - \sum_{k=0}^{j-1} \alpha_{k,j+1,n}\Dm_{j \rightarrow k+1}(O_k^{(k)})\right)\;,
\end{equation}
and introduce for convenience the associated operator $\tilde{O}_j^{(j)}$ defined recursively as
\begin{equation}\label{eq:component_tilde}
    \tilde{O}_j^{(j)} = \alpha_{j,j+1,n}  O_j^{(j)} = \Cm_{j+1\rightarrow n} (O) - \sum_{k=0}^{j-1} \alpha_{k,j+1,n}\Dm_{j \rightarrow k+1}(O_k^{(k)}) = \Cm_{j+1\rightarrow n} (O) - \sum_{k=0}^{j-1} \frac{\alpha_{k,j+1,n}}{\alpha_{k,k+1,n}}\Dm_{j \rightarrow k+1}(\Tilde{O}_k^{(k)})\;.
\end{equation}

Using this notation, \autoref{eq:exp1} becomes 
\begin{equation}
    \Tr[(O_j^{(n)})^2] = \frac{1}{\alpha_{j,j+1,n}^2} \Tr[\left(\Dm_{n\rightarrow j+1}\left( \tilde{O}_j^{(j)}\right)\right)^2]\;.
\end{equation}

An important property of the raising and lowering map that we will exploit in what follows is the fact that they are adjoint maps as shown in \autoref{prop:adjoint_Ck_Dk}. Moreover, one can easily see that the maps $\Cm_{j+1 \rightarrow n}$ and $\Dm_{n \rightarrow j+1}$  introduced in \autoref{eq:comp_lowering_notation} and \autoref{eq:comp_raising_notation} are similarly adjoint maps.

Using this property, we obtain

\begin{align}
   \Tr[(O_j^{(n)})^2] &= \frac{1}{\alpha_{j,j+1,n}^2} \Tr[\left(\Dm_{n\rightarrow j+1}\left( \tilde{O}_j^{(j)}\right)\right)^2] \\
   &= \frac{1}{\alpha_{j,j+1,n}^2} \Tr[ \tilde{O}_j^{(j)} \Cm_{j+1\rightarrow n}\Dm_{n\rightarrow j+1}\left( \tilde{O}_j^{(j)}\right)]\\
   &= \frac{1}{\alpha_{j,j+1,n}} \Tr[(\tilde{O}_j^{(j)})^2]\;,
\end{align}
where in the second equality we used the adjoint property and in the last equality we used the fact that the operator $\tilde{O}_j^{(j)}$ belongs to $\lambda_j^{(j)}$ and that the map $\Cm_{j+1\rightarrow n}\Dm_{n\rightarrow j+1}$ acts as a scalar multiple of identity on $\lambda_j^{(j)}$ with the scalar given by $\alpha_{j,j+1,n}$.

Now, we focus on further developing the expression of the Hilbert--Schmidt norm of the operator $\tilde{O}_j^{(j)}$.

\begin{align}\label{eq:reduced_purity_formula}
    \Tr\left[\left( \Tilde{O}_j^{(j)}  \right)^2\right] 
    &=  \Tr[\left( \Cm_{j+1\rightarrow n} (O) - \sum_{k=0}^{j-1} \frac{\alpha_{k,j+1,n}}{\alpha_{k,k+1,n}}\Dm_{j \rightarrow k+1}(\tilde{O}_k^{(k)})\right)^2] \\
    \\&= \Tr\left[\left( \Cm_{j+1\rightarrow n}(O)  \right)^2\right] + \Tr\left[\left( \sum_{k=0}^{j-1} \frac{\alpha_{k,j+1,n}}{\alpha_{k,k+1,n}} \Dm_{j\rightarrow k+1}( \tilde{O}_k^{(k)})  \right)^2\right] - 2 \Tr\left[ \Cm_{j+1\rightarrow n}(O)\left( \sum_{k=0}^{j-1} \frac{\alpha_{k,j+1,n}}{\alpha_{k,k+1,n}} \Dm_{j\rightarrow k+1}( \tilde{O}_k^{(k)})   \right)\right]\;. \label{eq:exp2}
\end{align}
For ease of notation, we introduce the shorthand $g_j$ for the first term appearing in the above sum.
\begin{equation}\label{eq:g_k}
    g_j := \Tr\left[\left( \Cm_{j+1\rightarrow n}(O)  \right)^2\right]\;.
\end{equation}
As will be detailed in later sections, this term can be easily computed for a given observable $O \in W_n$ since it simply consists of applying the lowering map $n-j$ times on $O$, removing at each application one photon in every mode.

Now, we focus on computing the other two terms appearing in the sum in \autoref{eq:exp2}.

The second term in \autoref{eq:exp2} further simplifies as 
\begin{align}
    \Tr\left[\left( \sum_{k=0}^{j-1} \frac{\alpha_{k,j+1,n}}{\alpha_{k,k+1,n}} \Dm_{j\rightarrow k+1}( \tilde{O}_k^{(k)})  \right)^2\right] &= \sum_{k=0}^{j-1} \left(\frac{\alpha_{k,j+1,n}}{\alpha_{k,k+1,n}}\right)^2 \Tr\left[  \Dm_{j\rightarrow k+1}( \tilde{O}_k^{(k)}) \Dm_{j\rightarrow k+1}( \tilde{O}_k^{(k)}) \right] \\
    & \quad + \sum_{k \neq p =0}^{j-1} \frac{\alpha_{k,j+1,n}}{\alpha_{k,k+1,n}} \frac{\alpha_{p,j+1,n}}{\alpha_{p,p+1,n}} \Tr\left[  \Dm_{j\rightarrow k+1}( \tilde{O}_k^{(k)}) \Dm_{j\rightarrow p+1}( \tilde{O}_p^{(p)}) \right]\\
    &= \sum_{k=0}^{j-1} \left(\frac{\alpha_{k,j+1,n}}{\alpha_{k,k+1,n}}\right)^2 \Tr\left[  \Dm_{j\rightarrow k+1}( \tilde{O}_k^{(k)}) \Dm_{j\rightarrow k+1}( \tilde{O}_k^{(k)}) \right]\\
     &=  \sum_{k=0}^{j-1} \left(\frac{\alpha_{k,j+1,n}}{\alpha_{k,k+1,n}}\right)^2 \Tr\left[ \tilde{O}_k^{(k)}  \Cm_{k+1 \rightarrow j}\Dm_{j\rightarrow k+1}( \tilde{O}_k^{(k)}) \right]\\
    &=  \sum_{k=0}^{j-1} \left(\frac{\alpha_{k,j+1,n}}{\alpha_{k,k+1,n}}\right)^2 \alpha_{k,k+1,j} \Tr[( \tilde{O}_k^{(k)})^2]\;. \label{eq:second_term}
\end{align}
In the second equality, the terms of the form $\Tr[  \Dm_{j\rightarrow k+1}( \tilde{O}_k^{(k)}) \Dm_{j\rightarrow p+1}( \tilde{O}_p^{(p)})]$  vanish since they are overlaps of two elements belonging to orthogonal subspaces, precisely $\lambda_k^{(j)}$ and  $\lambda_p^{(j)}$ for $k\neq p$. In the third equality, we use the adjoint property of the raising and lowering maps and in the fourth equality, we use the fact that the map $\Cm_{k+1 \rightarrow j}\Dm_{j\rightarrow k+1}$ acts as a scalar multiple of identity on the operator $\tilde{O}_k^{(k)} \in \lambda_k^{(k)}$ with the associated scalar being $\alpha_{k,k+1,j}$ according to \autoref{eq:alpha_coeffs}.

Using the expression of the coefficients of the form $\alpha_{k,j+1,n}$ defined in \autoref{eq:alpha_exp}, we further compute the coefficients appearing in the sum in \autoref{eq:second_term} for all $k\leq j-1$.
\begin{align}
    \left(\frac{\alpha_{k,j+1,n}}{\alpha_{k,k+1,n}}\right)^2 \alpha_{k,k+1,j} &= \left( \frac{(n-k)! (m+n+k-1)!}{(j-k)! (j+m+k-1)!} \frac{(m+2k-1)!}{(n-k)! (m+n+k-1)!} \right)^2 \frac{(j-k)! (j+m+k-1)!}{(m+2k-1)!}\\
    &= \left( \frac{(m+2k-1)!}{(j-k)! (j+m+k-1)!}  \right)^2 \frac{(j-k)! (j+m+k-1)!}{(m+2k-1)!}\\
    &= \frac{(m+2k-1)!}{(j-k)! (j+m+k-1)!}\;.
\end{align}
Here, we notice that the obtained result coincides with the ratio between the coefficients $\alpha_{k,j+1,n}$ and $\alpha_{k,k+1,n}$. Consequently, we have that the sum in \autoref{eq:second_term} can be rewritten as 
\begin{equation}\label{eq:second_term_not}
    \Tr\left[\left( \sum_{k=0}^{j-1} \frac{\alpha_{k,j+1,n}}{\alpha_{k,k+1,n}} \Dm_{j\rightarrow k+1}( \tilde{O}_k^{(k)})  \right)^2\right] =  \sum_{k=0}^{j-1} \frac{\alpha_{k,j+1,n}}{\alpha_{k,k+1,n}} \Tr[( \tilde{O}_k^{(k)})^2] =: \sum_{k=0}^{j-1} a_{k,j}\Tr[( \tilde{O}_k^{(k)})^2]\;,
\end{equation}
where we introduced the shorthand 
\begin{equation}\label{eq:a_kj}
    a_{k,j} = \frac{\alpha_{k,j+1,n}}{\alpha_{k,k+1,n}} = \frac{(m+2k-1)!}{(j-k)! (j+m+k-1)!}\;.
\end{equation}

Similarly, the third term in \autoref{eq:exp2}  simplifies as 
\begin{align}
    \sum_{k=0}^{j-1} \frac{\alpha_{k,j+1,n}}{\alpha_{k,k+1,n}} \Tr\left[ \Cm_{j+1\rightarrow n}(O)  \Dm_{j\rightarrow k+1}( \tilde{O}_k^{(k)})   \right] &=  \sum_{k=0}^{j-1} \frac{\alpha_{k,j+1,n}}{\alpha_{k,k+1,n}}\Tr[\Cm_{k+1\rightarrow j} \Cm_{j+1\rightarrow n}(O) \tilde{O}_k^{(k)}]\\
    &=  \sum_{k=0}^{j-1} \frac{\alpha_{k,j+1,n}}{\alpha_{k,k+1,n}}\Tr[\Cm_{k+1\rightarrow n}(O) \tilde{O}_k^{(k)}]\\
    &=:\sum_{k=0}^{j-1} a_{k,j}f_k\;,\label{eq:third_term_recursive}
\end{align}
where in the first equality, we used the fact that $\Cm_{k+1\rightarrow j} $ is the adjoint map of $\Dm_{j\rightarrow k+1}$ and in the last equality we used the shorthand $a_{k,j}$ introduced in \autoref{eq:a_kj} and further introduced the shorthand $f_k$ for every $k \leq j-1$ given by
\begin{equation}\label{eq:f_k_notation}
    f_k := \Tr[\Cm_{k+1\rightarrow n}(O) \tilde{O}_k^{(k)}]\;.
\end{equation}

By plugging the definition of the operator $\tilde{O}_k^{(k)}$ given in \autoref{eq:component_tilde} back into the expression of $f_k$ given in \autoref{eq:f_k_notation}, we observe a recursive pattern. Precisely, we have
\begin{align}
    f_k &:= \Tr[\Cm_{k+1\rightarrow n}(O) \tilde{O}_k^{(k)}]\\
    &= \Tr[\Cm_{k+1\rightarrow n}(O)  \left(\Cm_{k+1\rightarrow n} (O) - \sum_{l=0}^{k-1} \frac{\alpha_{l,k+1,n}}{\alpha_{l,l+1,k}}\Dm_{k \rightarrow l+1}(\Tilde{O}_l^{(l)})\right)]\\
    &= \Tr[(\Cm_{k+1\rightarrow n}(O))^2] - \sum_{l=0}^{k-1} \frac{\alpha_{l,k+1,n}}{\alpha_{l,l+1,k}} \Tr[\Cm_{k+1\rightarrow n}(O)\Dm_{k \rightarrow l+1}(\Tilde{O}_l^{(l)})]\\
    &= g_k(O) - \sum_{l=0}^{k-1} a_{l,k} f_l\;, \label{eq:third_term}
\end{align}
where in the last equality, we used the shorthands $g_k, a_{l,k}$ and $f_l$ introduced in \autoref{eq:g_k}, \autoref{eq:a_kj}, and \autoref{eq:f_k_notation} respectively.

By plugging in the notation from \autoref{eq:g_k} and the results from \autoref{eq:second_term_not} and \autoref{eq:third_term_recursive} back into \autoref{eq:exp2}, we finally obtain
\begin{align}\label{eq:norm_dumb}
    \Tr\left[\left( \Tilde{O}_j^{(j)}  \right)^2\right] &= g_j +  \sum_{k=0}^{j-1} a_{k,j} \Tr[( \tilde{O}_k^{(k)})^2]   - 2 \sum_{k=0}^{j-1} a_{k,j} f_k \;.
\end{align}

Moreover, by using the recursive definition of $f_k$ established in \autoref{eq:third_term}, the expression Hilbert--Schmidt norm in \autoref{eq:norm_dumb} becomes 
\begin{align}
      \Tr\left[\left( \Tilde{O}_j^{(j)}  \right)^2\right] &= f_j +  \sum_{k=0}^{j-1} a_{k,j} \Tr[( \tilde{O}_k^{(k)})^2]   -  \sum_{k=0}^{j-1} a_{k,j} f_k\\
      &= f_j +  \sum_{k=0}^{j-1} a_{k,j} (\Tr[( \tilde{O}_k^{(k)})^2]   -   f_k)\;.
\end{align}

By introducing the shorthand $t_j:= \Tr\left[\left( \Tilde{O}_j^{(j)}  \right)^2\right]-f_j$, the above equation simplifies to the following recursive form
\begin{equation}
    t_j = \sum_{k=0}^{j-1} a_{k,j} t_k\;.
\end{equation}

For $j=0$, we have $t_0 = \Tr\left[\left( \Tilde{O}_0^{(0)}  \right)^2\right]-f_0 =0$. Hence, we get that $t_j=0, \forall 0 \leq j \leq n$. This implies that the Hilbert--Schmidt norm of the observable projection into irreducible components is nothing but the term $f_j$ defined recursively in \autoref{eq:third_term}.

In what follows, we focus on deriving an explicit expression of $f_k$.
Concretely, the recursion in \autoref{eq:third_term}  can be rewritten in the following matrix form 
\begin{align}
    g = Af
\end{align}
where $A$ is an $(n+1 \times n+1)$ lower triangular matrix with ones on its diagonal given by
\begin{align}
    A_{k,k}&=1\\
    A_{k,l} &= a_{l,k} \;, \forall 0 \leq l <k\\
    A_{k,l} &=0 \;, \forall k < l \leq n\;.
\end{align}

Obtaining a closed form expression for $f$ boils down to computing the entries of the matrix $B:= A^{-1}$, i.e. $(AB)_{i,j} = \delta_{i,j}$. Since the matrix $A$ is lower triangular, its inverse $B$ will also be lower triangular. More precisely, we have that $\forall j \leq i$ the following identity holds
\begin{align}
   (AB)_{i,j} = \sum_{l=0}^n A_{i,l} B_{l,j} = \sum_{l=j}^i A_{i,l} B_{l,j} = \sum_{l=j}^i a_{l,i} B_{l,j} = \delta_{i,j}\;.
\end{align}

For $i=j$, we get $a_{i,i} B_{i,i} =1 $, which implies that $B_{i,i} = 1\;, \forall \; 0 \leq i \leq n$.

For a fixed $j$, we now derive the expression of $B_{j+s,j}$ for all $s \geq 0$  such that $ j+s \leq n$. In what follows, we show by induction that
\begin{equation}
    B_{j+s,j} =  \frac{(-1)^{s} \Gamma(2j + s+m-1)}{s! \Gamma(2j+2s+m-1)} = \frac{(-1)^s}{ s! (2j+s+m-1)_s}\;.
\end{equation}

To do so, we use the following recursive relation, i.e.
\begin{align}\label{eq:recursive_B}
    0 = (AB)_{j+s,j} = \sum_{l=j}^{j+s} a_{l,j+s} B_{l,j} = \sum_{l=0}^{s} a_{l+j,j+s} B_{l+j,j}\;.
\end{align}

For $s=0$, we have 
\begin{equation}
    B_{j+0,j}= \frac{1}{(2j+m-1)_0}= 1\;.
\end{equation}

Let us assume that $B_{j+l,j} =  \frac{(-1)^{l} \Gamma(2j + l+m-1)}{l! \Gamma(2j+2l+m-1)} = \frac{(-1)^l}{ l! (2j+l+m-1)_l}$ for all $0\leq l \leq s$ and let us prove the result for $l=s+1$.
Henceforth, we introduce the shorthand $x= 2j+m-1$. Using this notation we have 
\begin{align}
    B_{j+l,j} &= \frac{(-1)^l \Gamma(x+l)}{l! \Gamma(x+2l)} \;, \forall \; 0\leq l \leq s \\
    a_{l+j,j+s+1} &=  \frac{\Gamma(x+2l+1)}{(s+1-l)! \Gamma(x+l+s+2)}\;.
\end{align}

Under this assumption and using the recursive equation in \autoref{eq:recursive_B}, we get
\begin{align}
    B_{j+s+1,j} &= - \frac{1}{a_{j+s+1,j+s+1}} \sum_{l=j}^{j+s} a_{l,j+s+1} B_{l,j}\\
    &=  -  \sum_{l=0}^{s} a_{l+j,j+s+1} B_{l+j,j}\\
    &= - \sum_{l=0}^{s} \frac{(-1)^l \Gamma(x+l) \Gamma(x+2l+1)}{l! (s+1-l)! \Gamma(x+2l) \Gamma(x+l+s+2)}\\
     &= - \sum_{l=0}^{s} \frac{(-1)^l (x+2l)}{l! (s+1-l)! (x+l)_{s+2}}\\
     &= - \sum_{l=0}^{s+1} \frac{(-1)^l (x+2l)}{l! (s+1-l)! (x+l)_{s+2}} + \frac{(-1)^{s+1} (x+2s+2)}{(s+1)! (x+s+1)_{s+2}}\\
     &= - \sum_{l=0}^{s+1} \frac{(-1)^l \binom{s+1}{l} (x+2l)}{(s+1)! (x+l)_{s+2}} + \frac{(-1)^{s+1} }{(s+1)! (x+s+1)_{s+1}}\\
     &= \frac{(-1)^{s+1} }{(s+1)! (x+s+1)_{s+1}}\;,
    \label{eq:trinag_recursive}
\end{align}
where we used $\frac{\Gamma(x+k)}{\Gamma(x)} = (x)_k:= x (x+1) \dots (x+k-1)$ and in the last equality we used the property from \cite{szablowski2026identitiesintegralsinvolvingpochhammer}, which we recall here: %in \autoref{prop:sum_zero_hyper}.
\begin{equation}
    \sum_{j=0}^n (-1)^j \binom{n}{j} \frac{(x+2j -1)}{(x+j-1)_{(n+1)}}=\delta_{n,0}\;.
\end{equation}

Consequently, we obtain the following closed form of $f_k$ defined recursively in \autoref{eq:third_term}, i.e.
\begin{align}
    f_k &= \sum_{l=0}^k B_{k,l} g_l\\
    &= \sum_{l=0}^k \frac{(-1)^{k-l}}{(k-l)! (k+l+ m-1)_{k-l}} g_l\;.
\end{align}

\end{proof}

An important quantity appearing in the expression of irrep norms established in the previous  \autoref{th:recursive_irrep_norm_app} is $g_k(O)$ defined for each irrep $k$, as in \autoref{eq:g_k_th}. In the following Lemma, we compute this quantity for projectors onto Fock states and show how it can have a very simple expression for collision free and maximal bunching states.

\begin{lemma}{(Irrep purities  for Fock states).}\label{lemma:puR_fock}
    Consider a Fock state $R$ with a total number of $n$ photons, i.e $R \in \Phi_m^n$. For every $0\leq k \leq n$, the term $g_k(\ketbra{R}{R})$ defined in \autoref{eq:g_k_th} can be expressed as 
    \begin{equation}
        g_k(\ketbra{R}{R}) = ((n-k)!)^2 \sum_{\substack{|\boldsymbol{b}| = n-k\\0 \leq \boldsymbol{b}\leq R}} \prod_{i=1}^m \binom{R_i}{b_i}^2\;.
    \end{equation}

    For  collision free states $R_{\rm cf}$, the expression of $g_k$ simplifies to 
    \begin{equation}\label{eq:g_cf}
        g_k(\ketbra{R_{\rm cf}}{R_{\rm cf}}) = ((n-k)!)^2  \binom{n}{n-k}\;,
    \end{equation}
    whereas for  states with maximal bunching $R_b$, i.e.\ all photons are in a single mode, we get
    \begin{equation}
        g_k(\ketbra{R_b}{R_b}) = ((n-k)!)^2  \binom{n}{n-k}^2\;.
    \end{equation}
\end{lemma} 

\begin{proof}
We begin by recalling the expression of $ g_k(\ketbra{R}{R})$ introduced in \autoref{eq:g_k_th}, which corresponds to the Hilbert Schmidt norm of the operator obtained after applying the lowering map $\Cm$, $n-k$ times on the Fock state $\ketbra{R}{R}$, i.e.

\begin{equation}\label{gk_R}
    g_k(\ketbra{R}{R}) = \Tr\!\left[\big(\Cm_{k+1 \rightarrow n}(\ketbra{R}{R})\big)^2\right]\;.
\end{equation}

    Let us focus on computing the action of applying the lowering map $k$ times on a Fock state  $\ketbra{R}{R}$. By exploiting the fact that creation/annihilation operators commute on different modes, we obtain
    \begin{align}
        \Cm^k(\ketbra{R}{R}) &= \sum_{s_1,\dots,s_k=1}^m a_{s_1} \dots a_{s_k} \ketbra{R}{R} a_{s_k}^\dagger   \dots a_{s_1}^\dagger\\
        &= \sum_{\substack{|\boldsymbol{b}| = k\\0 \leq \boldsymbol{b}\leq  R}} \binom{k}{b_1, \dots , b_m} \bigotimes_{i=1}^m \left( a_i^{b_i}  \ketbra{R_i}{R_i} (a_i^\dagger)^{b_i}\right)\\
        &= \sum_{\substack{|\boldsymbol{b}| = k\\0 \leq \boldsymbol{b}\leq R}} \binom{k}{b_1, \dots , b_m}  \left( \prod_{i=1}^m R_i (R_i-1) \dots (R_i- b_i +1) \right)   \ketbra{R - \sum_{i=1}^m b_i e_i}{R - \sum_{i=1}^m b_i e_i}\\
        &= \sum_{\substack{|\boldsymbol{b}| = k\\0 \leq \boldsymbol{b}\leq R}} k!  \left( \prod_{i=1}^m \frac{R_i!}{b_i! (R_i - b_i)!} \right)   \ketbra{R - \boldsymbol{b}}{R -\boldsymbol{b}}\\
        &= k! \sum_{\substack{|\boldsymbol{b}| = k\\0 \leq \boldsymbol{b}\leq R}} \prod_{i=1}^m \binom{R_i}{b_i} \ketbra{R - \boldsymbol{b}}{R -\boldsymbol{b}}\;.
        \label{eq:contract_Fock_state} 
    \end{align}
In the second equality, we used the multinomial identity and the fact that configurations $\boldsymbol{b} \in \Phi_m^k$ with non zero contribution corresponds to the ones satisfying the pairwise order $\boldsymbol{b} \leq R$. In the third equality, we simply used the fact that $a \ket{n} = \sqrt{n} \ket{n-1}$. Finally, in the final equality, we identify the product of binomial coefficients over all modes.

By substituting $k$ in \autoref{eq:contract_Fock_state} by $n-k$, we finally obtain 

\begin{align}\label{eq:result_g}
     g_k(\ketbra{R}{R}) := \Tr\!\left[\big(\Cm_{k+1 \rightarrow n}(\ketbra{R}{R})\big)^2\right] = (n-k)!^2 \sum_{\substack{|\boldsymbol{b}| = n-k\\0 \leq \boldsymbol{b}\leq R}} \prod_{i=1}^m \binom{R_i}{b_i}^2    
\end{align}

The generic expression established in   \autoref{eq:result_g} is clearly invariant under mode permutations. 

Consequently, the value of $g_k$ for all collision free states is the same and is given by 
\begin{align}
     g_k(\ketbra{R_{\rm cf}}{R_{\rm cf}}) = (n-k)!^2  \sum_{\substack{|\boldsymbol{b}| = n-k\\ \boldsymbol{b}\in \{0,1\}^{n}}} \prod_{i=1}^n \binom{1}{b_i}^2  =  (n-k)!^2  \sum_{\substack{|\boldsymbol{b}| = n-k\\ \boldsymbol{b}\in \{0,1\}^{n}}} 1 = (n-k)!^2  \binom{n}{n-k}\;.
\end{align}

For a state with maximal bunching, the permutation invariance implies that the result  for all such state is the same as $\ket{R_b}= \ket{n,0,\dots,0}$. In this case, we get
\begin{align}
    g_k(\ketbra{R_{\rm b}}{R_{\rm b}}) = (n-k)!^2 \binom{n}{n-k}^2\;. 
\end{align}

This concludes the proof.
    
\end{proof}

A useful consequence of the previous lemma is the evaluation of the sum of the quantity $g_k(\ketbra{R}{R})$ over all Fock states $R \in \Phi_m^n$. This quantity will play a key role in subsequent proofs of anti-concentration, and we formalize it in the following proposition.

\begin{proposition}\label{prop:sum_gk}
Consider the quantity $g_{n-k}(\ketbra{R}{R})$ introduced in \autoref{eq:g_k_th} for all Fock states $R \in \Phi_m^n$ and $0 \leq k\leq n $. Their sum is given by 
    \begin{equation}
       \sum_{R \in \Phi_m^n}  g_{n-k}(\ketbra{R}{R}) = |\Phi_m^n| \frac{n!k! }{ (n-k)!}  \frac{(n-k + (m+1)/2)_k}{((m+1)/2)_k}\;.
    \end{equation}
\end{proposition}

\begin{proof}
  In \autoref{lemma:puR_fock}, we showed that for each Fock state $R \in \Phi_m^n$ , the expression of   $ g_k(\ketbra{R}{R})$ is given by
  \begin{equation}
      g_{n-k}(\ketbra{R}{R}) = (k!)^2 \sum_{\substack{|\boldsymbol{b}| = k\\0 \leq \boldsymbol{b}\leq R}} \prod_{i=1}^m \binom{R_i}{b_i}^2\;.
  \end{equation}

  In what follows, we show that when summing the above quantity over all possible Fock states with $n$ photons in $m$ modes, the sum can be identified with a coefficient of Gregenbauer polynomials whose coefficients are expressed as hypergeometric functions.

  To this end, we introduce the shorthand $S_{n,k}:= \sum_{\substack{|\boldsymbol{b}| = k\\0 \leq \boldsymbol{b}\leq R}} \prod_{i=1}^m \binom{R_i}{b_i}^2$ and 
  begin by  showing that $S_{n,k}$ corresponds to the coefficient of the following $2$--variate polynomial
  \begin{equation}
      S_{n,k}= [u^n v^k](1-2u(1+v)+u^2(1-v)^2)^{-\frac{m}{2}}\;.
  \end{equation}

  Let us introduce the $2$--variate generating function $F(u,v)$ given by 
  \begin{equation}\label{eq:f_gen}
      F(u,v) = \sum_{R \in \mathbb{N}^m} \sum_{ \boldsymbol{0}\leq \boldsymbol{b} \leq R} \prod_{i=1}^m \binom{R_i}{b_i}^2 u^{|R|} v^{|b|} \;.
  \end{equation}
 where $R \in \mathbb{N}^m$ denotes the multi-index of length $m$ whose components take positive integer values, $\boldsymbol{b}$ denotes as well a multi-index of length $m$ whose components are pairwise smaller than those of $R$ and $|R|$ and $|\boldsymbol{b}|$ correspond to the $1$--norm of each vector.

Here, one can easily see that the coefficient corresponding to the term $u^n v^k$ in the generating function $F(u,v)$ coincides with $S_{n,k}$, which we denote henceforth by 
\begin{equation}\label{eq:F_S}
    [u^n v^k]F(u,v) =  S_{n,k}\;.
\end{equation}

Since entries in distinct modes are independent in \autoref{eq:f_gen}, the generating function $F(u,v)$ can be simply rewritten as 
  \begin{align}
      F(u,v) = \prod_{i=1}^m \left( \sum_{r \geq 0} \sum_{b=0}^r \binom{r}{b}^2 u^r v^b \right):= G(u,v)^m\;.
  \end{align}

  Focusing on the single-mode generating function $G(u,v)$, its coefficient corresponding to the term $u^r$ can be expressed as a hypergeometric function. Concretely, we have 
  \begin{align}
     [u^r]G(u,v)= \sum_{b=0}^r \binom{r}{b}^2  v^b &= {}_2F_1(-r,-r;1;v)\;.
  \end{align}
  where we recall that the hypergeometric function ${}_2F_1(-r,b;c;v)$ is given by 
  \begin{equation}
      {}_2F_1(-r,a;c;v) = \sum_{b=0}^r (-1)^b \binom{r}{b} \frac{(a)_b}{(c)_b} v^b\;.
  \end{equation}

  Using the hypergeometric representation of Legendre polynomials together with the Pfaff transformation, one finds that ${}_2F_1(-r,-r;1;v) = (1-v)^r P_r\left(\frac{1+v}{1-v}\right)$ \cite{nationalinstituteofstandardsandtechnology_nist_2010}.
  
   This implies that the generating function $G(u,v)$ can be identified with the generating function of Legendre polynomials \cite{nationalinstituteofstandardsandtechnology_nist_2010}, given by 
  \begin{equation}\label{eq:G_func}
      G(u,v) = \sum_{r \geq 0} {}_2F_1(-r,-r;1;v) u^r = \frac{1}{\sqrt{1-2u(1+v)+u^2(1-v)^2}}\;.
  \end{equation}

   By combining \autoref{eq:F_S} and \autoref{eq:G_func},  we get that 
  \begin{align}
      S_{n,k} &= [u^n v^k]F(u,v)\\
      &= [u^n v^k]G(u,v)^m\\
      &= [u^n v^k](1-2u(1+v)+u^2(1-v)^2)^{-\frac{m}{2}}\;.\label{eq:S_exp1}
  \end{align}

The expression of \(G(u,v)\) given in \autoref{eq:G_func} can also be identified with the generating function of Gegenbauer polynomials \cite{nationalinstituteofstandardsandtechnology_nist_2010} given by 
\begin{equation}\label{eq:gen_grenbauer}
    \sum_{p \geq 0} C_p^{(\alpha)}(x) t^p = (1-2xt + t^2)^{-\alpha}\;,
\end{equation}
 where the Gegenbauer polynomials \(C_p^{(\alpha)}(x)\) admit the hypergeometric representation
\begin{equation}
    C_p^{(\alpha)}(x) = \frac{(2\alpha)_p}{p!} \, {}_2F_1\!\left(-p,p + 2 \alpha ; \alpha + \tfrac{1}{2};\frac{1-x}{2}\right).
\end{equation}

By introducing the change of variables
\begin{equation}\label{eq:subs}
    x= \frac{1+v}{1-v}\;,\quad t= (1-v)u\;,\quad \alpha=\frac{m}{2}\;,
\end{equation}
 \autoref{eq:S_exp1} can be brought into the canonical Gegenbauer generating function form, i.e.
\begin{equation}
    (1-2u(1+v)+u^2(1-v)^2)^{-\frac{m}{2}}
    = \sum_{p \geq 0} C_p^{(m/2)}\!\left(\frac{1+v}{1-v}\right) (1-v)^p u^p\;.
\end{equation}

Consequently, the expression of \(S_{n,k}\) given in \autoref{eq:S_exp1} can be connected to coefficients of Gegenbauer polynomials as follows:
\begin{align}
    S_{n,k}
    &= [u^n v^k](1-2u(1+v)+u^2(1-v)^2)^{-\frac{m}{2}}\nonumber\\
    &= [u^n v^k]\left( \sum_{p \geq 0} C_p^{(m/2)}\!\left(\frac{1+v}{1-v}\right) (1-v)^p u^p  \right)\nonumber\\
    &= [v^k]\left( C_n^{(m/2)}\!\left(\frac{1+v}{1-v}\right) (1-v)^n \right)\;.
\end{align}
Substituting the hypergeometric expression of \(C_n^{(m/2)}\), we obtain
\begin{equation}
    S_{n,k}
    = \frac{(m)_n}{n!} [v^k]\left( (1-v)^n \, {}_2F_1\!\left(-n,n + m ; \tfrac{m+1}{2};\frac{v}{v-1}\right) \right),
\end{equation}
where we used the standard relation between Gegenbauer and hypergeometric functions \cite{nationalinstituteofstandardsandtechnology_nist_2010}.

To simplify the argument of the hypergeometric function, we use the Pfaff transformation
\begin{equation}
    {}_2F_1(a,b;c;z)
    = (1-z)^{-a} \, {}_2F_1\!\left(a,c-b;c;\frac{z}{z-1}\right),
\end{equation}
a classical identity for  hypergeometric functions \cite{nationalinstituteofstandardsandtechnology_nist_2010}, which yields
\begin{align}
    S_{n,k}
    &= \frac{(m)_n}{n!} [v^k] \left( {}_2F_1\!\left(-n,-n - \tfrac{m-1}{2} ; \tfrac{m+1}{2};v\right) \right).
\end{align}
Using the series expansion of the hypergeometric function, the coefficient extraction is immediate, giving
\begin{align}
    S_{n,k}
    &= \frac{(m)_n}{n!} (-1)^k \binom{n}{k} \frac{(-n - (m-1)/2)_k}{((m+1)/2)_k}.
\end{align}
Finally, using the elementary identity for Pochhammer symbols \( (-a)_k = (-1)^k (a-k+1)_k \), we obtain
\begin{equation}
    S_{n,k}
    = \binom{n+m-1}{n} \binom{n}{k}
    \frac{(n-k + (m+1)/2)_k}{((m+1)/2)_k}\;.
\end{equation}
  
Multiplying by the prefactor $k!^2$, we finally obtain the desired result. Precisely, we have
\begin{equation}
     \sum_{R \in \Phi_m^n}  g_{n-k}(\ketbra{R}{R}) =  k!^2 S_{n,k} = k!^2 \binom{m+n-1}{n} \binom{n}{k} \frac{(n-k + (m+1)/2)_k}{((m+1)/2)_k} = |\Phi_m^n| \frac{n!k! }{ (n-k)!} \frac{(n-k + (m+1)/2)_k}{((m+1)/2)_k}
\end{equation}
where we used the shorthand  $|\Phi_m^n| = \binom{m+n-1}{n} $. This concludes the proof.

\end{proof}
\section{Second moment expressions for passive linear optics transformations}\label{app:sec_moment_proof}

In this section, we provide a proof for the \autoref{prop:sec_moment} which we recall here for completeness.

\begin{proposition}\label{prop:sec_moment_app}
    Consider a particle-number preserving bosonic observable $O \in \mathcal{W}$ and its  expectation value $f_U(\rho,O)$ of the form in \autoref{eq:linear_form} with respect to an  initial state $\rho$ on $m$ modes evolved under passive linear optics transformations. Further assume that either the initial state or the observable acts exactly on $n$ photons.  Then, the second moment over  Haar random interferometers $U \sim U(m)$ of $f_U(\rho,O)$ is given by  

    \begin{equation}\label{eq:sec_m_exp_app}
        \e[U \sim U(m)]{f_U(\rho,O)^2} =  \sum_{k=0}^n \frac{1}{d_{k}^{(n)}}\|P_k^{(n)}(\rho)\|_2^2 \|P_k^{(n)}(O)\|_2^2,
    \end{equation}
    where $P_k^{(n)}$ is a projection map onto the isotypic component $\lambda_k^{(n)}$ given in  \autoref{eq:irrep_diagram} and $d_{k}^{(n)}$ is its dimension given in \autoref{eq:irrep_dim_app}. 
\end{proposition}

\begin{proof}
    We consider an observable $O$ acting  exactly on $n$ photons, i.e.\ $O\in W_n$. This same assumption can be transferred to the initial state without loss of generality.
    
    Then its expectation value with respect to an arbitrary initial state evolved using a linear optical circuit is given by 
    \begin{align}\label{eq:vectorized_f}
        f_U(\rho,O) &= \Tr[\rho \varphi(U)^\dagger O  \varphi(U)]\;.
    \end{align}

    Using the vectorization formalism, the above expectation value can be expressed in the following forms 
    \begin{align}
        f_U(\rho,O) =\vbra{O}  \omega(U)\vket{\rho} = \vbra{\rho}  \omega(U)^\dagger\vket{O}\;,
    \end{align}
    where we recall that $\omega(U)$ is the (vectorized) adjoint representation of the photonic homomorphism introduced in \autoref{eq:adjoint_rep} for $U\in U(m)$.

    Given that $O\in W_n$ and $\omega(U)$ decomposes block-diagonally in the Fock basis as detailed in \autoref{eq:decomp_omega_m}, only the contribution from the $n$-photon sector survive, i.e.
    \begin{align}\label{eq:vectorized_f_n}
         f_U(\rho,O) =\vbra{O}  \omega_n(U) \hat{P}^{(n)}\vket{\rho} = \vbra{\rho} \hat{P}^{(n)}  \omega_n(U)^\dagger\vket{O}\;,
    \end{align}
    where $\hat{P}^{(n)}$ is the orthogonal projector onto the Fock sector $\mathcal{F}_n$, acting on the vectorized space $W_n = \mathcal{F}_n \otimes \mathcal{F}_n^*$. In particular, by considering the projector onto the Fock sector $\mathcal{F}_n$ denoted by $\Pi^{(n)}= \sum_{R \in \Phi_{n,m} } \ketbra{R}{R}$, the vectorized projector $\hat{P}^{(n)}$ acting on $W_n$ can be expressed as follows
    \begin{equation}
        \hat{P}^{(n)} := \Pi^{(n)} \otimes \Pi^{(n)}\;.
    \end{equation}

    The second moment of $f_U(\rho,O)$ with respect to Haar random interferometers $U\sim U(m)$ is thus given by
    \begin{align}
        \mathbb{E}_{U\sim U(m)}[f_U(\rho,O)^2]&=   \mathbb{E}_{U\sim U(m)}[\vbra{O}  \omega_n(U)\hat{P}^{(n)}\vket{\rho} \vbra{\rho} \hat{P}^{(n)} \omega_n(U)^\dagger\vket{O}]\\
        &= \mathbb{E}_{U\sim U(m)}\left[\Tr[\vket{O}\vbra{O}  \omega_n(U)\hat{P}^{(n)}\vket{\rho} \vbra{\rho} \hat{P}^{(n)} \omega_n(U)^\dagger]\right]\;,
    \end{align}
    where in the first equality we used the vectorized forms of $f_U$ given in \autoref{eq:vectorized_f_n} and in the second equality we used the cyclicity of the trace.

    To further evaluate the average over Haar random interferometers, we invoke the $G$- twirl Theorem which we recall in \autoref{th:G_twirl} for the compact group $U(m)$ and its unitary representation $\omega_n$ acting on the $n$-particle operator space $W_n$. 

    In \autoref{lemma:irrep_decomp}, we established that the representation $\omega_n$ and subsequently the space $W_n$ decompose into multiplicity free irreps as
    \begin{equation}
        W_n \simeq \bigoplus_{k=0}^n \lambda_k^{(n)}\;.
    \end{equation}

    Applying the $G$-twirl \autoref{th:G_twirl} yields that the second moment is given by
    \begin{align}
         \mathbb{E}_{U\sim U(m)}[f_U(\rho,O)^2] &=  \Tr[\vket{O}\vbra{O} \int_{U \sim U(m)} \omega_n(U) \hat{P}^{(n)}\vket{\rho}\vbra{\rho}\hat{P}^{(n)} \omega_n(U)^\dagger]\\
         &= \sum_{k=0}^n \frac{\Tr[\hat{P}_k^{(n)}\vket{O}\vbra{O} ] \Tr[\hat{P}_k^{(n)}\vket{\rho}\vbra{\rho} ]}{d_k^{(n)}} \;,
    \end{align}
where $\hat{P}_k^{(n)}$ is the vectorized projection map onto the irreducible component $\lambda_k^{(n)}$. Precisely, by considering the projection operator $\Pi_k^{(n)}$ acting on the operator space $W_n$, the expression of $\hat{P}_k^{(n)}$ is given by 
\begin{equation}
    \hat{P}_k^{(n)}:= \Pi_k^{(n)} \otimes (\Pi_k^{(n)})^*\;.
\end{equation}

By plugging this definition back inside the trace, we obtain
\begin{align}
    \Tr[\hat{P}_k^{(n)}\vket{O}\vbra{O} ]  &= \vbra{O} \Pi_k^{(n)} \otimes (\Pi_k^{(n)})^* \vket{O}\\
    &= \Tr[O^\dagger \Pi_k^{(n)} O \Pi_k^{(n)} ]\\
    &= \Tr[O_k^{(n) \dagger} O_k^{(n)}]\;,
\end{align}
where $O_k^{(n)}$ denotes the component of the the observable O supported on the irreducible subspace $\lambda_k^{(n)}$, i.e. ${O_k^{(n)}:= \Pi_k^{(n)} O \Pi_k^{(n)}}$.

\end{proof}

\section{Anti-concentration of boson sampling}\label{sec:anticoncentrationFockBS}

\subsection{Anti-concentration in hardness arguments}\label{app:anticon_hardness}

In this section, we review the role of anticoncentration in complexity-theoretic arguments, particularly for hardness of boson sampling in the saturated regime, and clarify the distinction between forms of anticoncentration.

\paragraph{Sampling-to-counting reduction and error structure.}

The hardness of boson sampling is based on a reduction from approximate sampling to average-case estimation of output probabilities. Fix an interferometer $U \in U(m)$ and assume there exists an efficient classical algorithm $\mathcal{A}$ that samples from a distribution $q_U$ that is $\varepsilon$-close in total variation distance to the target distribution $p_U$, i.e.,
\begin{equation}
    \|q_U - p_U\|_1 \leq \varepsilon.
\end{equation}
Using Stockmeyer's approximate counting algorithm, one can then estimate $q_U(\s)$ for any outcome $S \in \Phi_m^n$ by a value $\tilde{q}_U(S)$ satisfying
\begin{equation}
    |\tilde{q}_U(S) - q_U(S)| \leq \frac{1}{\mathrm{poly}(m,n)}\, q_U(S) \leq \frac{1}{\mathrm{poly}(m,n)} p_U(S) + \frac{1}{\mathrm{poly}(m,n)}|q_U(S)-p_U(S)|.
\end{equation}

The deviation between the target output probability and the estimated one through Stockmeyer's algorithm can be bounded  as

\begin{align}
      |\tilde{q}_U(S) - p_U(S)| &\leq  |\tilde{q}_U(S) - q_U(S)| +  |q_U(S) - p_U(S)|\\
      &\leq \frac{1}{\mathrm{poly}(m,n)} p_U(S) +  \left(1+ \frac{1}{\mathrm{poly}(m,n)}\right) |q_U(S)-p_U(S)|.
\end{align}

Combining the above error with 
 Markov’s inequality yields that, with probability at least $1-\delta$ over $(U,S)$,
\begin{align}\label{eq:stockmeyer_error_app}
    |\tilde{q}_U(S) - p_U(S)|
    \leq \frac{1}{\mathrm{poly}(m,n)}\, p_U(S)
    + \frac{2\varepsilon}{\delta |\Phi_m^n|}\left(1 + \frac{1}{\mathrm{poly}(m,n)}\right).
\end{align}
Thus, the reduction produces an estimate combining both multiplicative and additive errors. Establishing hardness of estimating output probabilities within this error threshold would rule out the existence of an efficient classical sampler $\mathcal{A}$. The key remaining step is therefore to control the additive contribution relative to $p_U(S)$.

\paragraph{Anticoncentration.}

To convert the error bound in \autoref{eq:stockmeyer_error_app} into a purely multiplicative approximation, one requires that typical output probabilities are not too small relative to their mean value:
\begin{equation}\label{eq:weak_ac}
    \Pr_{\substack{U \sim U(m)\\ S \sim \Phi_m^n}}
    \left[
    p_U(S) \geq \frac{\alpha}{|\Phi_m^n|}
    \right]
    \geq \gamma(\alpha),
\end{equation}
for some constants $\alpha, \gamma(\alpha)> 0$. 
The above formulation implies that for almost all pairs $(U,S)$ (specifically a constant fraction), the output probabilities are at least a constant fraction of their mean value $1/|\Phi_m^n|$. 

Under this condition, the additive term in the error bound in \autoref{eq:stockmeyer_error_app} becomes of the same order as $p_U(S)$
and can therefore be absorbed into the multiplicative contribution. More precisely, with probability $\gamma(\alpha)$ we get the upper bound 
\begin{equation}
    \frac{p_U(S)}{\alpha} \geq \frac{1}{|\Phi_m^n|}.
\end{equation}
Set 
\begin{equation}\label{eq:defEpsilonConc}
    \epsilon(\alpha, \delta) = \frac{1}{\mathrm{poly}(m,n)}
    + \frac{2\varepsilon}{\alpha\delta}\left(1 + \frac{1}{\mathrm{poly}(m,n)}\right), 
\end{equation}
then, combining the above in \autoref{eq:stockmeyer_error_app} yields the multiplicative error bound
\begin{equation}
    \Pr_{\substack{U \sim U(m)\\ S \sim \Phi_m^n}}
    \left[ |\tilde{q}_U(S) - p_U(S)|  \leq \epsilon(\alpha, \delta) p_U(S)\right]
    \geq \gamma(\alpha) - \delta.
\end{equation}

Given that $\delta$ is by definition also of the order $\delta \sim 1/ \rm poly(n)$, the estimate of the target output probability becomes a purely multiplicative approximation for almost all inputs $U$ and output configurations $S$.

While this condition guarantees that a non-negligible (constant) fraction of outcomes have probabilities of the correct scale, it does not suffice to eliminate the additive term with arbitrarily high probability. Instead, the conversion to multiplicative error only holds on this subset of \emph{good} instances, yielding a multiplicative approximation with constant probability rather than with overwhelming probability. 
This form of anti-concentration still supports a reduction to multiplicative estimation, albeit in a weaker sense than what considered for hardness proofs of boson sampling \cite{aaronson_computational_2011,bouland2025exponentialimprovementsaveragecasehardness}. In this setting, to rule out an efficient classical sampler, one must establish multiplicative hardness restricted for a larger fraction of outputs, namely a polynomially large fraction of the outcomes. This corresponds to a strictly stronger requirement on the underlying hardness assumption, as it demands that multiplicative estimation remains intractable even when restricted to a comparatively smaller fraction of instances. More precisely, the hardness proof demand the following form of anti-concentration,
\begin{equation}
    \Pr_{\substack{U \sim U(m)\\ S \sim \Phi_m^n}}
    \left[
    p_U(S) \geq \frac{1}{\text{poly}(n)|\Phi_m^n|}
    \right]
    \geq 1-\frac{1}{\text{poly}(n)},
\end{equation}
which we do not tackle in this work. This form of anti-concentration is essential for establishing approximate sampling hardness based on high-probability multiplicative hardness of probability estimation in the context of boson sampling. However, it has been so far only conjectured \cite{aaronson_computational_2011,nezami2021permanentrandommatricesrepresentation}. 
It is nonetheless important to stress that other sampling schemes, whose computational hardness proof rely on similar properties, instead conjecture anti-concentration, as defined in \autoref{eq:weak_ac} and proved in \autoref{cor:ac_app}
\cite{oszmaniec_fermion_2022,bremner_averagecase_2016,Hangleiter_2018}.

In the collision-free regime, collision-free states dominate the output distributions. Thus, it is sufficient to consider uniformly sampled collision-free configurations in the above analysis. Moreover, the permutation invariance of such states along with the Haar invariance preserved under the photonic homomorphism action, we get that 
\begin{equation}
\Pr_{\substack{U \sim U(m)\\ S \sim \Omega_m^n}}[S] = \Pr_{U \sim U(m)}[S_{\text{ref}}],
\end{equation}
 where $\Omega_m^n$ denotes the set of collision-free configurations and $S_{\text{ref}}$ is a reference state..
 In this regime, the multiplicative hardness conjecture reduces to the multiplicative hardness of estimating permanents of Gaussian matrices with iid entries. However in the saturated regimes, the conjecture is rather formalized for permanents of submatrices of Haar random unitaries over $U(m)$ of size $n \times n$ with repeated rows and columns. This problem is dubbed by $|\rm SUPER|^2_\times$ in \cite{bouland_complexitytheoretic_2023}, which we recall here, adapted to our settings, for completeness.

 \begin{conjecture}[|$\text{SUPER}|_\times$, \cite{bouland_complexitytheoretic_2023}]\label{cj:multiplicative}
    The following family of problems, denoted |$\text{SUPER}|_\times$, is \#P-hard: given $U \sim U(m)$
as above, output $z \in \mathbb{C}$ such that 
\begin{equation}
    \left|z- |p_U(S)|^2\right| \leq \epsilon(\alpha, \delta) p_U(S) 
\end{equation}
with probability $\geq \gamma(\alpha) - \delta$ over the choice of $U$ and $S$, where $\epsilon$ is defined in \autoref{eq:defEpsilonConc}.
 \end{conjecture}

The significance of the above conjecture, together with a formal proof of anti-concentration in this regime, lies in providing strong evidence toward closing the robustness gap in complexity-theoretic hardness arguments for the saturated regime. In \cite{bouland_complexitytheoretic_2023}, the authors develop alternative techniques to establish hardness of sampling when collisions are prevalent and standard Gaussian approximations break down. However, as in other sampling hardness results, their proof still exhibits a robustness gap, stemming from the mismatch between the additive error achievable for average-case probability estimation (this problem is denoted by $|\text{SUPER}|_\pm$ in \cite{bouland_complexitytheoretic_2023}) and that required for hardness. Establishing the multiplicative hardness conjecture formulated in \autoref{cj:multiplicative}, together with anti-concentration, would bridge this gap.

\subsection{Normalized average outcome collision probability}\label{sec:P2}

The weak form of anti-concentration defined in \autoref{eq:weak_ac} is usually diagnosed through second moments of the output distribution, and in particular through the normalized average outcome collision probability
\begin{equation}\label{eq:P2}
    P_2(m,n) = |\Phi_m^n| \sum_{S \in \Phi_m^n} \mathbb{E}_{U \sim U(m)}[p_U(S)^2].
\end{equation}
Computing $P_2(m,n)$ via second-moment methods allows one to establish lower bounds of the form in \autoref{eq:weak_ac} through inequalities such as Paley–Zygmund. Indeed, Paley–Zygmund yields 
\begin{equation}\label{eq:PZ_app}
     \Pr_{\substack{U \sim U(m)\\ S \sim \Phi_m^n}}
    \left[
    p_U(S) \geq \frac{\alpha}{|\Phi_m^n|}
    \right] \geq \frac{(1-\alpha)^2}{P_2(m,n)},
\end{equation}
as detailed in the proof of \autoref{cor:ac_app}. 

This expression of $P_2(m,n)$ as the ratio of moments explains its lengthy denomination; normalized average outcome collision probability. Specifically, we the outcome collision probability refers to the probability that two independent samples coincide, averaged over the outcome configurations $S \sim \Phi_m^n$ and random interferometer $U \sim U(m)$. The normalization arises from dividing by the square of the first moment, ensuring that the resulting quantity is bounded by one.

Beyond its role in diagnosing anti-concentration, the average outcome collision probability arises naturally in linear cross entropy benchmarking. Specifically, linear cross-entropy benchmarking is a method for assessing the fidelity of a sampling device by comparing experimentally observed output probabilities with the ideal distribution \cite{boixo_characterizing_2018}.In the ideal, noise-free case,  linear cross-entropy benchmarking score is given by
\begin{equation}
    F_{\rm XEB} = |\Phi_m^n| \sum_{S \in \Phi_m^n} p_U(S) \tilde{p}_U(S) -1\;,
\end{equation}
where $p_U(S)$ is the output probability of the ideal distribution and  $\tilde{p}_U(S)$ is the probability of the potentially noisy experiment. Therefore, the expected value of the ideal cross-entropy over all possible
unitaries is therefore expressed in terms of the normalized average outcome collision probability, i.e.
\begin{equation}
    \mathbb{E}_{U \sim U(m)}[F_{\rm XEB}^{\rm ideal}] = |\Phi_m^n| \sum_{S \in \Phi_m^n} \mathbb{E}_{U \sim U(m)}[p_U(S)^2] -1 := P_2(m,n)-1\;.
\end{equation}

In the collision-free regime, where $m$ scales at least quadratically in $n$, the dominance of collision-free output configurations provide a key simplification for analysing second moments and thus establishing anti-concentration via the computation of the normalized average outcome collision probability. Specifically, in this setting one can only consider uniformly sampling from collision-free states replace the mean over all the Fock space by the mean over collision-free configurations in the anti-concentration definition given in \autoref{eq:PZ_app}. Consequently, a rough approximation, which has been used in \cite{Ehrenberg_2025}, of the normalized average collision outcome probability yields 
\begin{equation}
     P_2(m,n) \approx |\Omega_m^n| \sum_{S \in \Omega_m^n} \mathbb{E}_{U \sim U(m)}[p_U(S)^2],
\end{equation}
where $ \Omega_m^n := \{S \in \Phi_m^n \,:\, s_i \in \{0,1\}, \;\forall 1\leq i \leq m\}$ denote the set of collision-free output configurations. Indeed,  approximation error occurring from restricting the sum to collision-free states converges to zero thanks to the bosonic birthday paradox \cite{aaronson_computational_2011,arkhipov_bosonic_2012}.

Moreover, by exploiting the Haar invariance under mode permutation (which still holds after the application of the photonic homomorphism), the sum over collision-free states reduces to
\begin{equation}\label{eq:simple_P2}
     P_2(m,n) \approx |\Omega_m^n|^2 \mathbb{E}_{U \sim U(m)}[p_U(S_{\rm cf})^2],
\end{equation}
for any reference collision free configuration $S_{\rm cf} \in \Omega_m^n$.

Beyond this regime, however, this approximation breaks down, as collision events become prevalent. To formalize this statement, we analyse the ratio $|\Omega_m^n|/|\Phi_m^n|$ in different regimes of $m$ as a function of $n$.

We begin by recalling that the cardinality of collision-free states  is given by
\begin{equation}
    |\Omega_m^n| = \binom{m}{n},
\end{equation}
while the full Fock space has cardinality is 
\begin{equation}
    |\Phi_m^n| = \binom{m+n-1}{n}.
\end{equation}
Hence, the fraction of collision-free states is
\begin{equation}\label{eq:ratio_cf_full}
    \frac{|\Omega_m^n|}{|\Phi_m^n|}
    = \frac{\binom{m}{n}}{\binom{m+n-1}{n}}
    = \prod_{j=0}^{n-1}\frac{m-j}{m+j}.
\end{equation}
This quantity also coincides with the probability of observing no collisions in the bosonic birthday paradox.
In the asymptotic limit, we obtain
\begin{equation}
    \log \frac{|\Omega_m^n|}{|\Phi_m^n|} = \sum_{j=0}^{n-1} \log\left(\frac{m-j}{m+j}\right) = \sum_{j=0}^{n-1} \left(-\frac{2j}{m}+ \mathcal{O}\left(\frac{j^3}{m^3}\right)\right) = -\frac{n(n-1)}{m} +\mathcal{O}\left(\frac{n^4}{m^3}\right)
\end{equation}

We now specialize to the scaling  $m = c n^\beta$ with $c \geq 1$.

\paragraph*{Collision-free regime \texorpdfstring{(\(\beta>2\))}{(beta>2)}.}
In this regime, one has
\begin{equation}\label{eq:collision_free}
   \frac{|\Omega_m^n|}{|\Phi_m^n|}
    = 1 - \frac{n(n-1)}{m} + O\!\left(\frac{n^2}{m^2}\right)
    = 1 - \frac{1}{c}n^{2-\beta} + o\!\left(n^{2-\beta}\right),
\end{equation}
so the collision-free sector asymptotically exhausts the full Fock space.

\paragraph*{Quadratic regime \texorpdfstring{(\(\beta=2\))}{(beta=2)}.}
At the bosonic birthday paradox threshold, one finds
\begin{equation}\label{eq:quadratic}
    \frac{|\Omega_m^n|}{|\Phi_m^n|}
    \longrightarrow e^{-1/c}.
\end{equation}
Thus, collision-free configurations retain a finite asymptotic weight, but no longer dominate the sample space.

\paragraph*{Intermediate saturated regime \texorpdfstring{(\(1<\beta<2\))}{(1<beta<2)}.}
In this regime,
\begin{equation}\label{eq:intermediate1}
    \log \frac{|\Omega_m^n|}{|\Phi_m^n|}
    = -\frac{n(n-1)}{m} + O\!\left(\frac{n^4}{m^3}\right)
    = -\frac{1}{c}n^{2-\beta} + o\!\left(n^{2-\beta}\right),
\end{equation}
and therefore
\begin{equation}\label{eq:intermediate2}
    \frac{|\Omega_m^n|}{|\Phi_m^n|}
    = \exp\!\left(-\frac{1}{c}n^{2-\beta} + o\!\left(n^{2-\beta}\right)\right).
\end{equation}
Hence the fraction of collision-free states is already asymptotically negligible.

\paragraph*{Linear regime \texorpdfstring{(\(\beta=1\), \(m=cn\))}{(beta=1, m=cn)}.}
Using Stirling's formula, one obtains
\begin{equation}\label{eq:linear}
    \frac{|\Omega_m^n|}{|\Phi_m^n|}
    \sim
    \sqrt{\frac{c+1}{c-1}}\,
    \exp\!\Big(
    n\big[
    2c\log c -(c-1)\log(c-1) -(c+1)\log(c+1)
    \big]
    \Big),
\end{equation}
for \(c>1\). Since
\begin{equation}
    2c\log c -(c-1)\log(c-1) -(c+1)\log(c+1) < 0,
\end{equation}
the ratio decays exponentially in \(n\). In particular, in the linear regime, restricting to collision-free outputs discards an exponentially large portion of the sample space.

Finally, the computation of the single expectation value for the reference collision-free state appearing in \autoref{eq:simple_P2} further simplifies under the \emph{hiding} property, which is expected to hold in collision-free regimes.
Informally, this property states that a sufficiently small submatrix of a Haar-random interferometer $U\in U(m)$ behaves approximately as a matrix with independent complex Gaussian entries \cite{aaronson_computational_2011}. Since output amplitudes of boson sampling are given by permanents of such submatrices (See \autoref{eq:permanent}), this allows one to replace analysis over Haar-random unitaries by analysis over Gaussian random matrices, for which powerful tools from random matrix theory are available. Concretely, we have 
\begin{equation}
    \mathbb{E}_{U \sim U(m)}[p_U(S_{\rm cf})^2] \simeq \frac{1}{m^{2n}} \mathbb{E}_{X \sim \mathcal{G}^{n \times n}} [|\rm Per(X)|^2]
\end{equation}
where $\mathcal{G}^{n \times n}$ denotes the normalized Gaussian distribution over $n \times n$ matrices. This follows from establishing that this gaussian distribution in close in total variation distance to the distribution induced by the Haar measure over $U(m)$ on submatrices of size $n \times n$. 
In Aaronson and Arkhipov’s seminal work, this property has been proven in the regime where $m = \Omega(n^5)$ and conjectured to be correct in the quadratic regime $m \sim n^2$. Follow up works have established the hiding property in regimes such as  $n^3 = O(m)$ \cite{Leverrier_2018} and $n^2 = O(m)$ for random orthogonal matrices \cite{jiang_distances_2019}.

This analysis yields that in the regime where collision free outcomes dominate and the hiding property holds that the output distribution of a Fock state boson sampler is anti-concentrated over a polynomially vanishing fraction of instances $U \in U(m)$, as shown in \cite{aaronson_computational_2011}.

In the more experimentally relevant linear regime $m = O(n)$, the hiding property breaks down due to the presence of photon collisions. Consequently, one must treat different classes of outcomes separately, and the analysis becomes significantly more involved. While recent works have established hardness of sampling in this regime using alternative techniques \cite{bouland_complexitytheoretic_2023}, anti-concentration properties have so far been investigated primarily through numerical evidence.

\subsection{Proof of anti-concentration beyond the collision-free regime}\label{app:P2}

In this section, our goal is to establish \autoref{cj:scalingP}, proving anti-concentration of boson sampling in the linear regime.

To this end, we first derive a series of results on the average outcome collision probability, which will serve as the main ingredient in obtaining anti-concentration via an application of the Paley–Zygmund inequality, as outlined in \autoref{sec:ac_back}.

The section is organized as follows. We begin by deriving closed-form expressions for the outcome collision probability. Beyond their role in proving anti-concentration, these expressions are of independent interest, as this quantity also appears in linear cross-entropy benchmarking. We present these equivalent formulations both in the order in which they arise in our derivation and because each provides a distinct perspective that can be used to evaluate the quantity, allowing practitioners—particularly in linear cross-entropy benchmarking—to adopt the most convenient representation for their purposes.

In \autoref{th:P2_hypergeometric}, we express the outcome collision probability as a sum of hypergeometric functions, where each term corresponds to the contribution of a given irreducible representation. In \autoref{pr:TkExpBeta}, we show that these contributions can be identified with moments of the Beta distribution. Moreover, we  derive an integral representation of the full outcome collision probability, which enables us to analyse its scaling across different regimes of modes and photons.
Building on these results, we finally establish the asymptotic scaling of the normalized average outcome collision probability in \autoref{th:asymptotoic_P2_app} and hence the
desired anti-concentration statement in \autoref{cor:ac_app}.

\begin{theorem}[Average outcome collision probability]\label{th:P2_hypergeometric}
    The normalized average outcome-collision probability for linear optical
circuits can be expressed as
 \begin{equation}
     P_2(m,n) =  \sum_{k=0}^n \frac{(m+k-1)_{k}}{(m+n)_k} {}_2 F_1(-k,n-k+1;2-m-2k;-1)\;.
 \end{equation}
 The shorthand $(x)_p:= x (x+1) \dots (x+p-1)$ denotes the Pochhammer symbol and ${}_2 F_1$ is the hypergeometric function defined as 
 \begin{equation}
     {}_2 F_1(-k,a;b;z) = \sum_{p=0}^k (-1)^p \binom{k}{p} \frac{(b)_p}{(c)_p} z^p\;,
 \end{equation}
 for any positive integer $k$.
\end{theorem}

\begin{proof}

We begin by recalling the definition of the average outcome collision probability $P_2(m,n)$ introduced in \autoref{eq:P2},
\begin{equation}\label{eq:P2_app}
    P_2(m,n) = |\Phi_m^n| \sum_{R \in \Phi_m^n} \mathbb{E}_{U\sim U(m)} [p_U(R)^2]\;,
\end{equation}
where $p_U(R)$ denotes the probability of starting with the collision free state $\ket{R_{\rm cf}}= \ket{1,\dots,1,0,\dots,0}$ and measuring $R \in \Phi_m^n$ at the output of the interferometer given by
\begin{equation}\label{eq:prob}
    p_U(R) =  \Tr[\ketbra{R_{\rm cf}}{R_{\rm cf}} \varphi_n(U)^\dagger \ketbra{R}{R}\varphi_n(U)]\;.
\end{equation}

One can easily notice that the above expression of $p_U(R)$ is of the form in \autoref{eq:linear_form_n} for $\rho := \ketbra{R_{\rm cf}}{R_{\rm cf}}$ and $O := \ketbra{R}{R}$.

Consequently, by combining the general second moment from  \autoref{prop:sec_moment} and the reduced purities expression established in \autoref{th:recursive_irrep_norm}, the expectation value appearing in the expression of $P_2(m,n)$ in \autoref{eq:P2_app} can be expressed as 
\begin{align}
    \mathbb{E}_{U\sim U(m)} [p_U(R)^2] &= \sum_{k=0}^n \frac{\Tr[\left(P_k^{(n)}\left(\ketbra{R_{\rm cf}}{R_{\rm cf}}\right)\right)^2] \Tr[\left(P_k^{(n)}(\ketbra{R}{R})\right)^2]}{d_k^{(n)}}\\
    &=  \sum_{k=0}^n \frac{1}{ d_k^{(n)} \alpha_{k,k+1,n}^2} \left(\sum_{l=0}^k B_{k,l} g_l(\ketbra{R_{\rm cf}}{R_{\rm cf}})\right) \left(\sum_{l=0}^k B_{k,l} g_l(\ketbra{R}{R})\right)\;,\label{eq:pp1}
\end{align}
where we introduced the shorthand $B_{k,l}$
\begin{equation}
     B_{k,l} :=\frac{(-1)^{k-l}}{(k-l)! (k+l+ m-1)_{k-l}}\;,
\end{equation}
and recalling the expressions of $\alpha_{k,k+1,n}$ and $d_k^{(n)}$,   respectively, 
\begin{align}
     \alpha_{k,k+1,n} &= (n-k)! (m+2k)_{n-k}\\
      d_k^{(n)} &= \frac{2k+m-1}{m-1} \binom{k+m-2}{k}^2\;.
\end{align}

According to \autoref{lemma:puR_fock} and in particular \autoref{eq:g_cf}, we have 
\begin{equation}\label{eq:free_gl}
    g_l(\ketbra{R_{\rm cf}}{R_{\rm cf}}) = (n-l)!^2 \binom{n}{n-l} = (n-l)!^2 \frac{n!}{(n-l)! l!} = \frac{n! (n-l)!}{l!}\;.
\end{equation}

Consequently, by plugging in \autoref{eq:free_gl} in \autoref{eq:pp1} and further summing over $R \in \Phi_m^n$, we obtain the following expression for $P_2(m,n)$

\begin{align}
     P_2(m,n) &= |\Phi_m^n| \sum_{R \in \Phi_m^n} \mathbb{E}_{U\sim U(m)} [p_U(R)^2]\\
     &= |\Phi_m^n| \sum_{R \in \Phi_m^n} \sum_{k=0}^n \frac{1}{ d_k^{(n)} \alpha_{k,k+1,n}^2} \left(\sum_{l=0}^k B_{k,l} \frac{n! (n-l)!}{l!}\right) \left(\sum_{l=0}^k B_{k,l} g_l(\ketbra{R}{R})\right)\\
      &= |\Phi_m^n|  \sum_{k=0}^n \frac{1}{ d_k^{(n)} \alpha_{k,k+1,n}^2} \left(\sum_{l=0}^k B_{k,l} \frac{n! (n-l)!}{l!}\right) \left(\sum_{l=0}^k B_{k,l} \sum_{R \in \Phi_m^n} g_l(\ketbra{R}{R})\right)\;. \label{eq:P2_exp2}
\end{align}

In the above equation, we identify the term $\sum_{R \in \Phi_m^n} g_l(\ketbra{R}{R})$ which has been further evaluated in \autoref{prop:sum_gk}. 
Consequently, \autoref{eq:P2_exp2} becomes
\begin{equation}
    P_2(m,n) = |\Phi_m^n|^2 \sum_{k=0}^n \frac{1}{ d_k^{(n)} \alpha_{k,k+1,n}^2} \left(\sum_{l=0}^k B_{k,l} \frac{n! (n-l)!}{l!}\right) \left(\sum_{l=0}^k B_{k,l} S_{m,n,n-l}\right)  \;,
\end{equation}
where we introduced the shorthand 

\begin{equation}
     S_{m,n,n-l}  := \frac{n! (n-l)!}{l!} \frac{  (l + (m+1)/2)_{n-l}}{((m+1)/2)_{n-l}}\;.
\end{equation}

For convenience, we further introduce the shorthand
\begin{equation}\label{eq:T_k_def}
    T_k := \frac{|\Phi_m^n|^2 }{d_k^{(n)} ((n-k)! (m+2k)_{n-k})^2}  \left( \sum_{l=0}^k B_{k,l} \frac{n! (n-l)!}{l!}\right)\left(\sum_{l=0}^k B_{k,l} S_{m,n,n-l}\right)\;.
\end{equation}

We note that the term $\sum_{l=0}^k B_{k,l} \frac{n! (n-l)!}{l!}$ can be identified with a hypergeometric function as follows
\begin{align}\label{eq:cf_sum}
    \sum_{l=0}^k B_{k,l} \frac{n! (n-l)!}{l!} = n! \frac{(n-k)!}{k!} {}_2 F_1(-k,n-k+1;2-m-2k;-1)\;.
\end{align}

Now, we focus on further simplifying the sum $\sum_{l=0}^k B_{k,l} S_{m,n,n-l}$.
First , we introduce the change of variable $a=\frac{m+1}{2}$ and show the following identity
\begin{align}
    \frac{(a+l)_{n-l}}{(a)_{n-l}} = \frac{\Gamma(a+n)\Gamma(a)}{\Gamma(a+l)\Gamma(a+n-l)} = \frac{(a+n-l)_l}{(a)_l}\;.
\end{align}
Consequently, the expression of $S_{m,n,n-l}$ can be rewritten as 
\begin{align}
    S_{m,n,n-l} =\frac{n! (n-l)!}{l!} \frac{(a+n-l)_l}{(a)_l}\;.
\end{align}

By plugging this into the sum we obtain 
\begin{align}
    \sum_{l=0}^k B_{k,l} S_{m,n,n-l} &= n! \sum_{l=0}^k \frac{(-1)^{k-l}}{(k-l)! (k+l+ m-1)_{k-l}} \frac{ (n-l)!}{l!} \frac{(a+n-l)_l}{(a)_l}\\
    &= n! \sum_{j=0}^k \frac{(-1)^{j}}{(j)! (2k-j+ m-1)_{j}} \frac{ (n-k+j)!}{(k-j)!} \frac{(a+n-k+j)_{k-j}}{(a)_{k-j}}\\
    &=  n! (n-k)! \sum_{j=0}^k \frac{(-1)^{j}}{j! (k-j)!}  \frac{(n-k+1)_j(a+n-k+j)_{k-j}}{(2k-j+ m-1)_{j}(a)_{k-j}}\;.\label{eq:S_sum}
\end{align}

Using the following identities for $j\leq k$
\begin{align}
    (x)_{k} &= (x)_j (x+j)_{k-j}\\
    (x)_k &= (x)_{k-j} (x+k-j)_j\;,
\end{align}
and applying it to the numerator and denominator in the term $\frac{(a+n-k+j)_{k-j}}{(a)_{k-j}}$ yields
\begin{align}
    \frac{(a+n-k+j)_{k-j}}{(a)_{k-j}} = \frac{(a+n-k)_k}{(a+n-k)_j} \frac{(a+k-j)_j}{(a)_k} = \frac{(a+n-k)_k}{(a)_k} \frac{(a+k-j)_j}{(a+n-k)_j} = \frac{(a+n-k)_k}{(a)_k} \frac{(-1)^j (1-a-k)_j}{(a+n-k)_j}\;,
\end{align}
where we used the identity 
\begin{equation}\label{eq:pochhammer_neg}
    (x-j)_j = (-1)^j (1-x)_j\;.
\end{equation}

By plugging this back into \autoref{eq:S_sum}, we obtain
\begin{align}
    \sum_{l=0}^k B_{k,l} S_{m,n,n-l} &= n! (n-k)! \sum_{j=0}^k \frac{(-1)^{j}}{j! (k-j)!}  \frac{(n-k+1)_j}{(2k-j+ m-1)_{j}} \frac{(a+n-k)_k}{(a)_k} \frac{(-1)^j (1-a-k)_j}{(a+n-k)_j}\\
    &= n! (n-k)! \frac{(a+n-k)_k}{(a)_k} \sum_{j=0}^k \frac{1}{j! (k-j)!}  \frac{(n-k+1)_j}{(2k-j+ m-1)_{j}} \frac{(1-a-k)_j}{(a+n-k)_j}\\
    &= n! (n-k)! \frac{(a+n-k)_k}{(a)_k} \sum_{j=0}^k \frac{(-1)^j}{j! (k-j)!}  \frac{(n-k+1)_j}{(2-m-2k)_{j}} \frac{(1-a-k)_j}{(a+n-k)_j}\\
    &=\frac{n! (n-k)!}{k!} \frac{(a+n-k)_k}{(a)_k} \sum_{j=0}^k   \frac{(-k)_j(n-k+1)_j}{j! (2-m-2k)_{j}} \frac{(1-a-k)_j}{(a+n-k)_j}\;.
\end{align}

The above sum can be identified with the terminating hypergeometric series ${}_3F_2$
\begin{align}
    \sum_{j=0}^k   \frac{(-k)_j(n-k+1)_j}{j! (2-m-2k)_{j}} \frac{(1-a-k)_j}{(a+n-k)_j} &= {}_3 F_2\left(\begin{matrix}&-k, n-k+1,1-a-k\\&2-m-2k,a+n-k&\end{matrix};1\right)\;.
\end{align}
Here, we notice that the parameters of the hypergeometric function  obey the following relation 
\begin{equation}
    a+n-k = 1 + (n-k+1) + (1-a-k) - (2-m-2k) -k\;. 
\end{equation}

Consequently, we can invoke the Pfaﬀ--Saalschutz summation formula for a terminating ${}_3F_2$ \cite{nationalinstituteofstandardsandtechnology_nist_2010}
\begin{equation}
    {}_3 F_2\left(\begin{matrix}&-k, A,B\\&C,1+A+B-C-k&\end{matrix};1\right) = \frac{(C-A)_k (C-B)_k}{(C)_k (C-A-B)_k}\;,
\end{equation}
and obtain 
\begin{align}
    \sum_{l=0}^k B_{k,l} S_{m,n,n-l}  &= \frac{n! (n-k)!}{k!} \frac{(a+n-k)_k}{(a)_k} {}_3 F_2\left(\begin{matrix}&-k, n-k+1,1-a-k\\&2-m-2k,a+n-k&\end{matrix};1\right)\\
    &= \frac{n! (n-k)!}{k!} \frac{(a+n-k)_k}{(a)_k} \frac{(1-m-n-k)_k (2-a-k)_k}{(2-m-2k)_k (1-a-n)_k}\;.\label{eq:S_sum_1}
\end{align}

Using the identity given in \autoref{eq:pochhammer_neg}, we have
\begin{align}
    (1-m-n-k)_k &= (-1)^k (m+n)_k\\
    (2-a-k)_k &= (-1)^k (a-1)_k\\
    (2-m-2k)_k &= (-1)^k (m+k-1)_k\\
    (1-a-n)_k &= (-1)^k  (a+n-k)_k\;.
\end{align}

Consequently, \autoref{eq:S_sum_1} simplifies to 
\begin{align}
    \sum_{l=0}^k B_{k,l} S_{m,n,n-l} &= \frac{n! (n-k)!}{k!} \frac{(a+n-k)_k}{(a)_k} \frac{(1-m-n-k)_k (2-a-k)_k}{(2-m-2k)_k (1-a-n)_k}\\
    &= \frac{n! (n-k)!}{k!} \frac{(a+n-k)_k}{(a)_k} \frac{(m+n)_k (a-1)_k}{(m+k-1)_k (a+n-k)_k}\\
    &= \frac{n! (n-k)!}{k!} \frac{(a-1)_k}{(a)_k} \frac{(m+n)_k }{(m+k-1)_k }\\
    &= \frac{n! (n-k)!}{k!} \frac{m-1}{m+2k-1} \frac{(m+n)_k }{(m+k-1)_k} \;.\label{eq:sum2_simp}
\end{align}

Combining \autoref{eq:cf_sum} and \autoref{eq:sum2_simp}, the expression of $T_k$ simplifies to
\begin{align}
    T_k &= \frac{\binom{m+n-1}{n}^2 }{d_k^{(n)} ((n-k)! (m+2k)_{n-k})^2}  \left( \sum_{l=0}^k B_{k,l} \frac{ n!(n-l)!}{l!}\right)\left(\sum_{l=0}^k B_{k,l} S_{m,n,n-l}\right)\\
    &= \left( \frac{ (m)_n(n-k)!}{k! (n-k)! (m+2k)_{n-k}} \right)^2 \frac{1}{d_k^{(n)}} \frac{(m-1)}{(m+2k-1)} \frac{(m+n)_k }{(m+k-1)_{k}} {}_2 F_1(-k,n-k+1;2-m-2k;-1)\\
    &=  \left(\frac{(m)_n}{k!  (m+2k)_{n-k}} \right)^2  \frac{k!^2}{(m-1)_k^2} \left( \frac{(m-1)}{(m+2k-1)}\right)^2 \frac{(m+n)_k }{(m+k-1)_{k}} {}_2 F_1(-k,n-k+1;2-m-2k;-1)\\
    &=   \left(\frac{(m)_n}{(m)_{k-1}  (m+2k)_{n-k}  (m+k-1)_{k} (2k+m-1)} \right)^2 \frac{(m+k-1)_{k}}{(m+n)_k } {}_2 F_1(-k,n-k+1;2-m-2k;-1)\\
    &= \frac{(m+k-1)_{k}}{(m+n)_k} {}_2 F_1(-k,n-k+1;2-m-2k;-1)\;.\label{eq:Tk_final}
\end{align}

Summing over terms $T_k$ from $k=0$ to $n$ concludes the proof.
\end{proof}

In the following proposition, we show that the terms $T_k$ defined in \autoref{eq:T_k_def} whose sum corresponds to the average collision outcome probability can be expressed as moments of Beta distribution. This result allows us to obtain a compact integral form of the average outcome collision probability.

\begin{proposition}[Integral form of average outcome collision probability]\label{pr:TkExpBeta}
For $0\leq k \leq n $, the terms $T_k$ defined in \autoref{eq:Tk_final} can be expressed as moments of the Beta distribution. In particular, we have
    \begin{equation}
        T_k = \e[X \sim \text{Beta}(n-k+1, m+k-1)]{(1-2X_n)^k}\;,
    \end{equation}
    where $\text{Beta}(\alpha, \beta)$ is the Beta distribution with support over the interval $[0, 1]$ whose probability density function is given by
    \begin{equation}
        f(\alpha, \beta) = \frac{x^{\alpha - 1}(1-x)^{\beta-1}}{B(\alpha, \beta)}\;,
    \end{equation}
    with $B(\alpha, \beta) = \frac{\Gamma(\alpha)\Gamma(\beta)}{\Gamma(\alpha+\beta)}$  the \emph{Beta} function.

    Moreover, the average outcome collision probability admits the integral form 
    \begin{equation}\label{eq:P2_int_app}
                P_2(m, n) = (n+m-1)\int_0^{\pi/2} d\theta \cos(\theta)^{n+m-2} \sin((n+1)\theta)\;.
    \end{equation}
\end{proposition}
\begin{proof}

    Using the Pochhammer identity $(2-m-2k)_j = (-1)^j (m+2k-j-1)_j$,
    we express the terms $T_k$, defined in \autoref{eq:Tk_final},  as 
    \begin{align}
       T_k &= 
        \frac{(m+k-1)_{k}}{(m+n)_k} {}_2 F_1(-k,n-k+1;2-m-2k;-1)\\
        & =  \frac{(m+k-1)_{k}}{(m+n)_k} \sum_{j=0}^k\binom{k}{j}\frac{(-1)^j(n-k+1)_j}{(m+2k-j-1)_j} \label{eq:secondToLastStepTkToBeta}\\
        & = \sum_{j=0}^k(-1)^j \binom{k}{j} \frac{(n-k+1)_j(m+k-1)_{k-j}}{(m+n)_k} \label{eq:lastStepTkToBeta}\;,
    \end{align}
    where the last step follows from 
    \begin{equation}
        \frac{(m+k-1)_{k}}{(m+2k-j-1)_j} = (m+k-1)_{k-j}\;.
    \end{equation}
    From definition of the Beta function $B(\alpha, \beta) = \frac{\Gamma(\alpha)\Gamma(\beta)}{\Gamma(\alpha+\beta)}$,
    we obtain the following identity:
    \begin{equation}\label{eq:betaPoch}
        \frac{(\alpha)_j(\beta)_{k-j}}{(\alpha + \beta)_k} = \frac{B(\alpha + j, \beta+k-j)}{B(\alpha, \beta)}.
    \end{equation}
    Moreover, the Beta function admits an integral definition as 
    \begin{equation}\label{eq:betaInt}
        B(\alpha, \beta) = \int_0^1 x^{\alpha - 1}(1-x)^{\beta - 1}dx.
    \end{equation}
    We set $a = n-k+1$ and $b = m+k-1$, such that $a +b = m+n$. Thus, plugging \autoref{eq:betaPoch} and \autoref{eq:betaInt} into \autoref{eq:lastStepTkToBeta} yields
    \begin{align}
        T_k & = \frac{1}{B(a, b)} \sum_{j=0}^{k}(-1)^j \binom{k}{j}\int_0^1x^{a+j-1}(1-x)^{n+k-j-1} dx\\
        & = \frac{1}{B(a, b)} \int_0^1x^{a-1}(1-x)^{b-1} \sum_{j=0}^{k} \binom{k}{j} (-1)^jx^k (1-x)^{k-j} dx \\
        & = \frac{1}{B(a, b)}  \int_0^1 x^{a-1}(1-x)^{b-1} (1-2x)^k dx \\
        & = \e[X_n \sim \text{Beta}(a, b)]{(1-2X_n)^k}\;,
    \end{align}
    where the first step exchange the sum and the integral, the second step uses the binomial expansion of $(1-2x)^k$, and the last step exploits the probability density function of the Beta distribution. We keep track of $n$ by considering the random variable $X_n$ explicitly.

    Recall that we expressed $P_2(m)$ as 
    \begin{equation}
        P_2(m, n) = \sum_{k=0}^n T_k\;.
    \end{equation}

    Moreover,  moments of a Beta-distributed random variable are given by \cite{johnson_continuous_1994}
    \begin{equation}
        \e[X \sim \text{Beta}(a,b)]{X^j} = \frac{\binom{a-1+j}j}{\binom{a+b-1+j}j}\;.
    \end{equation}
    Thus, we obtain 
    \begin{align}
        P_2(m, n)
            & =\sum_{k=0}^n\e[X \sim \text{Beta}(a,b)]{(1-2X)^j} \\
            & =\sum_{k=0}^n\sum_{j=0}^k(-2)^j\binom kj \e[X \sim \text{Beta}(a,b)]{X^j} \\
            & = \sum_{k=0}^n\sum_{j=0}^k(-2)^j\binom kj \frac{\binom{n-k+j}{j}}{\binom{n + m-1+j}{j}} \\
            & =\sum_{j=0}^n(-2)^j\frac{1}{\binom{n + m-1+j}{j}}\sum_{k=j}^n \binom kj \binom{n-k+j}{j} \\
            & =\sum_{j=0}^n(-2)^j\frac{\binom{n+j+1}{2j+1}}{\binom{n + m-1+j}{j}}\label{eq:lastStepFirstSN}\;,
    \end{align}
    where the last step follows from Vandermonde's identity \cite{nationalinstituteofstandardsandtechnology_nist_2010}.
    As 
    \begin{equation}
    \frac{1}{\binom{N+j}{j}} = N \int_0^1 t^j (1-j)^{N-1} dt\;,
    \end{equation}
    setting $N = n + m - 1$ in \autoref{eq:lastStepFirstSN} yields
    \begin{equation}\label{eq:SnIntegralLast}
        P_2(m, n) = N \int_0^1 dt (1-j)^{N-1} \sum_{j=0}^n(-2t)^j\binom{n+j+1}{2j+1}\;.
    \end{equation}
    Using the identity 
    \begin{equation}
        \sum_{j=0}^n \binom{n+j+1}{2j+1}\left(x + \frac{1}{x} - 2\right)^j 
            = \frac{x^{n+1} - \frac{1}{x^{n+1}}}{x - \frac{1}{x}}\;,
    \end{equation}
    we set $x = e^{\imath \theta}$ giving $(x  + \frac{1}{x} - 2)^j = 2(\cos\theta - 1)$. Thus, when $t = 1 - \cos \theta$ we have $x + \frac{1}{x} - 2 = -2t$, and using Chebyshev's polynomial expansion via hypergeometric function, we obtain 
    \begin{equation}
        \sum_{j=0}^n\binom{n+j+1}{2j+1}(-2t)^j=\frac{x^{n+1}-\frac{1}{x^{n+1}}}{x-\frac{1}{x}}=\frac{\sin((n+1)\theta)}{\sin(\theta)}\;.
    \end{equation}
    We therefore make the change of variable $t = 1 - \cos\theta$ in \autoref{eq:SnIntegralLast}, thus $dt = \sin \theta d\theta$ and we obtain our final expression for $P_2(m, n)$ as
    \begin{equation}
        P_2(m, n) = (n+m-1)\int_0^{\pi/2} d\theta \cos(\theta)^{n+m-2} \sin((n+1)\theta)\;.
    \end{equation}
\end{proof}

Having established the integral form of the normalized average outcome collision probability, we are now ready to prove (a slightly more general version of) \autoref{thm:asymptoticP2}:

\begin{theorem}[Scaling  of normalized average outcome collision probability]\label{th:asymptotoic_P2_app}
Let $n \in \mathbb N$ and fix $m = cn^\beta\ge n$ for constants $c>0$ and $\beta \geq 1$.  Then the normalized average outcome-collision probability $P_2(m,n)$ satisfies the following asymptotic behaviour.

\begin{enumerate}

\item In the dilute regime with $\beta > 2$, the collision probability grows linearly with $n$, i.e.
\begin{equation}
    P_2(m,n) \sim n + o(1);.
\end{equation}

\item In the quadratic regime corresponding to $\beta = 2$, the scaling remains linear in $n$ with a constant prefactor depending on the constant $c$, i.e.
\begin{align}
    P_2(m,n) \sim \sqrt{2c}\,D_+\!\left(\frac{1}{\sqrt{2c}}\right)n + o(1)\;,
\end{align}
where $D_+$ denotes the Dawson function \cite{ARMSTRONG196761} defined as 
\begin{equation}
     D_+(y) = \frac{1}{2} \int_0^{+ \infty} dt \exp(-\frac{t^2}{4})\sin(yt)\;.
\end{equation}

\item In the non-dilute regime with
 $\beta<2$, the collision probability is dominated by the ratio $m/n$, i.e.
\begin{equation}\label{eq:mn}
    P_2(m,n)=\frac{m}{n}+1+O(n^{-\beta'/2})\;,
\end{equation}
for any $0<\beta'<\beta$.
In particular, in the linear regime $m=cn$, for any $0<\delta<1$, it holds that 
\begin{equation}\label{eq:P2mIsCn_app}
    P_2(m,n) =  c + 1 +  O(n^{-\delta/2})\;.
\end{equation}
\end{enumerate}
\end{theorem}
\begin{proof}
  
   We invoke the integral form of $P_2(n,m)$ established in \autoref{pr:TkExpBeta}.
    To find its limit, we split the interval of integration for some $\delta_n$ dependent on $n$ as 
    \begin{equation}
        P_2(m, n) = (n+m-1)\left[\int_0^{\delta_n} d\theta \cos(\theta)^{n+m-2} \sin((n+1)\theta) + \int_{\delta_n}^{\pi / 2} d\theta \cos(\theta)^{n+m-2} \sin((n+1)\theta)\right]\;.
    \end{equation}
    It holds that 
    \begin{equation}
        (n+m-1)\int_{\delta_n}^{\pi / 2} d\theta \cos(\theta)^{n+m-2} \sin((n+1)\theta) \leq \frac{\pi}{2}(n+m-1)\cos(\theta)^{n+m-2}\;,
    \end{equation}
    and moreover, the Taylor-Maclaurin series expansion of cosine is
    $\cos(\delta_n)  = 1 - \frac{1}{2}\delta_n^2 + O(\delta_n^4)$, thus
    \begin{equation}
        (\cos\delta_n)^{n+m-2}=\exp((n+m-2)\log(\cos\delta_n))=\exp(-\frac12(n+m-2)\delta_n^2+O(n^\beta \delta_n^4))\;.
    \end{equation}
    Picking $\delta_n = n^{-\alpha}$ with $\alpha \in (\frac{1}{4}, \frac{1}{2})$ ensures that 
    \begin{equation}
        \frac{\pi}{2}(n+m-1)\cos(\theta)^{n+m-2}
        \xrightarrow[n \to \infty]{}0\;.
    \end{equation}
    On the other hand, denote by
    \begin{equation}
        I_n=(n+m-1)\int_0^{\delta_n}d\theta(\cos\theta)^{n+m-2}\sin((n+1)\theta)=(n+m-1)\int_0^{\delta_n}d\theta\exp((n+m-2)\log(\cos\theta))\sin((n+1)\theta)\;.
    \end{equation}
    For $\theta \in [0, \delta_n]$, it holds that using cosine series expansion that
    \begin{align}
        \exp((n+m-2)\log(\cos\theta))
        & =\exp(-\frac12(n+m-2)\theta^2+O(n^\beta\theta^4)) \\
        & =\exp(-\frac12(n+m-2)\theta^2)\exp(O(n^\beta\theta^4))\\
        & =\exp(-\frac12(n+m-2)\theta^2)(1+O(n^\beta\theta^4))\;,
    \end{align}
    thus we obtain
    \begin{align}
        I_n
        & = (n+m-1)\int_0^{\delta_n}d\theta\exp(-\frac12(n+m-2)\theta^2)(1+O(n\theta^4))\sin((n+1)\theta)\\
        & = (n+m-1)\int_0^{\delta_n}d\theta\exp(-\frac12(n+m-2)\theta^2)\sin((n+1)\theta)+O(n^{2\beta}\delta_n^5)\;.
    \end{align}
    Taking $\delta_n=n^{-\alpha}$ with $\alpha\in(\frac{2\beta}5,\frac\beta2)$ ensures that the error term goes to $0$ as $n$ goes to infinity. Making the change of variable $x=\theta\sqrt{2(n+m-2)}$, we have 
    \begin{equation}
        I_n=\frac{n+m-1}{\sqrt{2(n+m-2)}}\int_0^{\delta_n\sqrt{2(n+m-2)}}dx\exp(-\frac14x^2)\sin(\frac{(n+1)x}{\sqrt{2(n+m-2)}})\;.
    \end{equation}
    Since $\delta_n=n^{-\alpha}$ with $\alpha\in(\frac{2\beta}5,\frac\beta2)$, we have $\lim_{n \to \infty} \delta_n\sqrt{2(n+m-2)} = \infty$, and
    \begin{equation}\label{eq:upperBoundErrExp}
        \left|\int_{\delta_n\sqrt{2(n+m-2)}}^{+\infty}dx\exp(-\frac14x^2)\sin(\frac{(n+1)x}{\sqrt{n+m-2}})\right|\le\exp(-\frac14(n+m-2)\delta_n^2)\int_{\delta_n\sqrt{2(n+m-2)}}^{+\infty}dx\exp(-\frac18x^2)\;.
    \end{equation}
    Hence, $I_n$ reads
    \begin{equation}\label{eq:InSeeDawson}
        I_n=\frac{n+m-1}{\sqrt{2(n+m-2)}}\int_0^{+\infty}dx\exp(-\frac14x^2)\sin(\frac{(n+1)x}{\sqrt{2(n+m-2)}}) + O(n^{2\beta}\delta_n^5)\;,
    \end{equation}
     up to the error term of \autoref{eq:upperBoundErrExp} which goes exponentially fast to $0$.
     We recognise in \autoref{eq:InSeeDawson} the Dawson function $D_+$ \cite{ARMSTRONG196761} defined as 
     \begin{equation}
         D_+(y) = \frac{1}{2} \int_0^{+ \infty} dt \exp(-\frac{t^2}{4})\sin(yt)\;,
     \end{equation}
     whose asymptotic expansion reads
     \begin{equation}\label{eq:expDawsonOrigin}
         D_+(y) = y + O(y^3)\;,
     \end{equation}
     near the origin and 
     \begin{equation}\label{eq:expDawsonInfty}
         D_+(y) = \frac{1}{2y} + O\left(\frac{1}{y^3}\right)\;,
      \end{equation}
    for large $y$.
     In particular, 
     \begin{equation}
         I_n = \frac{2(n+m-1)}{\sqrt{2(n+m-2)}}D_+\left(\frac{(n+1)}{\sqrt{2(n+m-2)}}\right) +O(n^{2\beta}\delta_n^5)\;. 
     \end{equation}
     When $\beta < 2$, we conclude using the asymptotic expansion of the Dawson function at large values (see \autoref{eq:expDawsonInfty}), giving
     \begin{equation}
         I_n=\frac{2(n+m-1)}{\sqrt{2(n+m-2)}}\frac{\sqrt{2(n+m-2)}}{2(n+1)}+O(n^{2\beta}\delta_n^5)=\frac mn+1+O(n^{2\beta}\delta_n^5)\;,
     \end{equation}
     which is valid for any $\delta_n = n^{-\alpha}$ with $\alpha \in (\frac{2\beta}5,\frac\beta2)$. In particular, 
     taking $\alpha$ as close as possible to $\beta / 2$ shows that the error term is $O(n^{-\beta '/2})$ for any $\beta '<\beta$.
     
     When $\beta > 2$, we conclude using the asymptotic expansion of the Dawson function at the origin (see \autoref{eq:expDawsonOrigin}), giving
     \begin{equation}
         I_n = \frac{(n+m-1) (n+1)}{(n+m-2)} + O (n^{2\beta}\delta_n^5) \sim n\;.
     \end{equation}
     Finally, when $\beta = 2$, setting $m=cn^2$ yields
    \begin{equation}
        I_n = \frac{2(n+c n^2-1)}{\sqrt{2(n+c n^2-2)}}D_+\left(\frac{(n+1)}{\sqrt{2(n+c n^2-2)}}\right) +O(n^{2\beta}\delta_n^5) \sim \sqrt{2c}D_+ \left(\frac{1}{\sqrt{2c}}\right)n\;.
    \end{equation}
\end{proof}

\begin{corollary}[Anti-concentration beyond the dilute regime]\label{cor:ac_app}
Let $n \in \mathbb N$ and for constants $c\geq 1$ and $1 \leq \beta < 2$, fix $m = cn^\beta$. In this regime, the output distribution of a boson sampler is anti-concentrated in the asymptotic limit, i.e.

\begin{equation}
    \Pr_{\substack{U \sim U(m)\\ S \sim \Phi_m^n}}
    \left[
    p_U(S) \geq \frac{\alpha}{|\Phi_m^n|}
    \right] \geq (1-\alpha)^2 \frac{1}{1+ m/n}\;.
\end{equation}
    
\end{corollary}

\begin{proof}
    The proof is primarily based on invoking the Paley–Zygmund inequality given by 
\begin{equation}\label{eq:PZ_app_proof}
     \Pr_{\substack{U \sim U(m)\\ S \sim \Phi_m^n}}
    \left[
    p_U(S) \geq \alpha \mathbb{E}[ p_U(S)]
    \right] \geq (1-\alpha)^2 \frac{\mathbb{E}[ p_U(S)]^2}{\mathbb{E}[ p_U(S)^2]}\;,
\end{equation}
where expectations are taken over $U \sim U(m)$ and $S \sim \Phi_m^n$.

Let us focus on showing that the ratio of first and second moments appearing in the right hand side of \autoref{eq:PZ_app_proof}.
The first moment of output probabilities over Haar random unitaries and uniformly samples configurations reads 
as
\begin{align}
    \mathbb{E}[ p_U(S)] = \frac{1}{|\Phi_m^n|} \sum_{S \in \Phi_m^n}\mathbb{E}_{U\sim U(m)}[ p_U(S)]\;.
\end{align}

For any outcome $S \in \Phi_m^n$, the first moment of its output probabilities can be easily computed by invoking the twirl theorem which we recall in \autoref{th:G_twirl} and the fact that the irreducibility of  representation $\varphi_n(U)$ . 
We obtain 
\begin{align}\label{eq:first_moment}
    \mathbb{E}_{U\sim U(m)}[ p_U(S)] = \mathbb{E}_{U\sim U(m)}[ \Tr[\ketbra{R_{\rm cf}}{R_{\rm cf}} \varphi_n^\dagger(U) \ketbra{S}{S} \varphi_n(U)]]= \frac{\Tr[(\ketbra{R_{\rm cf}}{R_{\rm cf}})^2] \Tr[(\ketbra{S}{S})^2]}{|\Phi_m^n|} = \frac{1}{|\Phi_m^n|}\;.
\end{align}

Consequently, the ratio of first and second moments can be expressed in terms of the average outcome collision probability as 
\begin{align}\label{eq:ratio}
    \frac{\mathbb{E}[ p_U(S)]^2}{\mathbb{E}[ p_U(S)^2]} &= \frac{|\Phi_m^n|}{|\Phi_m^n|^2  \sum_{S \in \Phi_m^n}\mathbb{E}_{U\sim U(m)}[ p_U(S)^2]}= \frac{1}{|\Phi_m^n|  \sum_{S \in \Phi_m^n}\mathbb{E}_{U\sim U(m)}[ p_U(S)^2]} := \frac{1}{P_2(m,n)}\;.
\end{align}

Hence, the probability lower bound in \autoref{eq:PZ_app_proof} becomes
\begin{equation}\label{eq:ac_end}
     \Pr_{\substack{U \sim U(m)\\ S \sim \Phi_m^n}}
    \left[
    p_U(S) \geq \frac{\alpha}{|\Phi_m^n|}
    \right] \geq (1-\alpha)^2 \frac{1}{P_2(m,n)}\;.
\end{equation}

Combining \autoref{eq:ac_end} with the asymptotic scaling of $P_2(m,n)$ established in \autoref{th:asymptotoic_P2_app} for the regime $m= c n^\beta$ with $1 \leq \beta <2$ (See \autoref{eq:mn}), we obtain the final result
\begin{equation}
    \Pr_{\substack{U \sim U(m)\\ S \sim \Phi_m^n}}
    \left[
    p_U(S) \geq \frac{\alpha}{|\Phi_m^n|}
    \right] \geq (1-\alpha)^2 \frac{1}{1+ m/n}\;.
\end{equation}

\end{proof}

\end{document}